\documentclass[11pt]{article}
\usepackage{amssymb,amsfonts}
\usepackage{amsthm}
\usepackage{amsmath}
\usepackage{graphicx}
\usepackage[curve,matrix,arrow,color]{xy}
\usepackage[mathscr]{eucal}

\usepackage[all,dvips,color]{xypic}
 \usepackage[usenames,dvipsnames]{color}
\xyoption{line}

\usepackage{epsf}
 \usepackage{epsfig}

\usepackage{tikz}

\DeclareSymbolFont{extraup}{U}{zavm}{m}{n}
\DeclareMathSymbol{\vardiamond}{\mathalpha}{extraup}{87}

\makeatletter
\@addtoreset{equation}{section}
\makeatother
\renewcommand{\theequation}{\thesection.\arabic{equation}}
\def\begeqar{\begin{eqnarray}}
\def\endeqar{\end{eqnarray}}
\def\begeq{\begin{equation}}
\def\endeq{\end{equation}}

\def\wgta#1#2#3#4{\hbox{\rlap{\lower.35cm\hbox{$#1$}}
\hskip.2cm\rlap{\raise.25cm\hbox{$#2$}}
\rlap{\vrule width1.3cm height.4pt}
\hskip.55cm\rlap{\lower.6cm\hbox{\vrule width.4pt height1.2cm}}
\hskip.15cm
\rlap{\raise.25cm\hbox{$#3$}}\hskip.25cm\lower.35cm\hbox{$#4$}\hskip.6cm}}

\def\wgtb#1#2#3#4{\hbox{\rlap{\raise.25cm\hbox{$#2$}}
\hskip.2cm\rlap{\lower.35cm\hbox{$#1$}}
\rlap{\vrule width1.3cm height.4pt}
\hskip.55cm\rlap{\lower.6cm\hbox{\vrule width.4pt height1.2cm}}
\hskip.15cm
\rlap{\lower.35cm\hbox{$#4$}}\hskip.25cm\raise.25cm\hbox{$#3$}\hskip.6cm}}

\def\begeqar{\begin{eqnarray}}
\def\endeqar{\end{eqnarray}}

\newcommand{\eqrefej}{(3.9)~} 
\newcommand{\eqrefuu}{(3.12)~} 
\newcommand{\eqrefuuu}{(4.22)~} 
\newcommand{\secrefcr}{4}

\newcommand{\proj}{p}
\newcommand{\modd}{\,\mathrm{mod}\,}

\newcommand{\one}{\boldsymbol{1}}
\newcommand{\tensor}{\otimes}

\newcommand{\q}{\mathfrak{q}}
\newcommand{\ffrac}[2]{\mbox{\footnotesize$\displaystyle\frac{#1}{#2}$}}
\newcommand{\idem}{\boldsymbol{e}}

\newcommand{\PTL}[1]{TL^a_{#1}}
\newcommand{\TL}[1]{TL_{#1}}
\newcommand{\JTL}[1]{JTL_{#1}}
\newcommand{\aJTL}[1]{\JTL{#1}^{(au)}}

\newcommand{\ATL}[1]{T^a_{#1}}
\newcommand{\oATL}[1]{O_{#1}}

\newcommand{\inv}{\mathrm{inv}}

\newcommand{\LQG}{U_{\q} s\ell(2)}
\newcommand{\LQGi}{U_{i} s\ell(2)}

\newcommand{\LQGodd}{U^{\text{odd}}_{\q} s\ell(2)}
\newcommand{\LQGoddi}{U^{\text{odd}}_{i} s\ell(2)}

\newcommand{\K}{\mathsf{K}}
\newcommand{\F}{\mathsf{F}}
\newcommand{\FF}[1]{\F_{#1}}
\newcommand{\f}{\mathsf{f}}

\newcommand{\E}{\mathsf{E}}

\newcommand{\EE}[1]{\E_{#1}}
\newcommand{\h}{\mathsf{h}}
\newcommand{\e}{\mathsf{e}}

\newcommand{\ext}[1]{\boldsymbol{#1}}

\newcommand{\half}{%
  \mathchoice{\ffrac{1}{2}}{\frac{1}{2}}{\frac{1}{2}}{\frac{1}{2}}}

\newcommand{\inn}[2]{\langle#1,#2\rangle}

\newcommand{\smatrix}[4]{\mbox{\scriptsize%
    $\displaystyle\begin{pmatrix}#1&#2\\
      #3&#4\end{pmatrix}$}}
\newcommand{\sigmap}{\smatrix{0}{1}{0}{0}}
\newcommand{\sigmam}{\smatrix{0}{0}{1}{0}}

\newcommand{\sigmaq}{\smatrix{\q}{0}{0}{-\q}}

\newcommand{\Hilb}{\mathcal{H}}

\newcommand{\veven}[1]{|v^{\text{even}}\rangle}
\newcommand{\vodd}[1]{|v^{\text{odd}}\rangle}

\newcommand{\vac}{\boldsymbol{\Omega}}
\newcommand{\lvac}{\boldsymbol{\omega}}

\newcommand{\oN}{\mathbb{N}}
\newcommand{\oC}{\mathbb{C}}

\newcommand{\oP}{\mathbb{P}}

\newcommand{\Endo}{\mathrm{End}}

\newcommand{\Hom}{\mathrm{Hom}}
\newcommand{\HomQodd}{\mathrm{Hom}_{\rule{0pt}{5.5pt}%
{U^{\text{odd}}_{\q}}}}

\newcommand{\gl}{g\ell}

\newcommand{\Ext}{\mathrm{Ext}_{\rule{0pt}{5.5pt}%
{\LQGodd}}^1}

\newcommand{\ExtJTL}{\mathrm{Ext}_{\rule{0pt}{7.5pt}%
{\JTL{N}}}^1}
\newcommand{\HomATL}{\mathrm{Hom}_{\rule{0pt}{6.5pt}%
{\ATL{N}}}}
\newcommand{\HomTL}{\mathrm{Hom}_{\rule{0pt}{6.5pt}%
{\JTL{N}}}}
\newcommand{\HomaJTL}{\mathrm{Hom}_{\rule{0pt}{6.5pt}%
{\aJTL{N}}}}
\newcommand{\EndJTL}{\mathrm{End}_{\rule{0pt}{6.5pt}%
{\JTL{N}}}}
\newcommand{\EndcJTL}{\mathrm{End}_{\rule{0pt}{6.5pt}%
{\centJTL}}}

\newcommand{\chVv}{\Hilb_{N}}

\newcommand{\repgl}{\pi_{\gl}}
\newcommand{\repXX}{\pi_{\mathrm{XX}}}
\newcommand{\repQG}{\rho_{\gl}}
\newcommand{\cent}{\mathfrak{Z}}
\newcommand{\centTL}{\cent_{\mathsf{TL}}}
\newcommand{\centJTL}{\cent_{\JTL{}}}

\newcommand{\N}{\mathrm N}

\newcommand{\toppr}{\mathsf{t}}
\newcommand{\botpr}{\mathsf{b}}
\newcommand{\leftpr}{\mathsf{l}}
\newcommand{\rightpr}{\mathsf{r}}
\newcommand{\stprp}{\mathsf{x}}

\newcommand{\Xodd}[1]{\mathsf{X}_{#1}}
\newcommand{\XX}{\mathsf{X}}
\newcommand{\PP}{\mathsf{P}}

\newcommand{\PPodd}[1]{\mathsf{T}_{#1}}
\newcommand{\TLX}{d^0}

\newcommand{\IrrTL}[1]{(d^0_{#1})}
\newcommand{\PrTL}[1]{\mathscr{P}_{#1}}
\newcommand{\StTL}[1]{\mathscr{W}_{#1}}
\newcommand{\AIrrTL}[2]{\mathscr{L}_{#1,#2}}
\newcommand{\APrTL}[1]{\widehat{\mathscr{P}}_{#1}}
\newcommand{\APrTLs}[1]{\widehat{\mathscr{P}}^*_{#1}}
\newcommand{\AStTL}[2]{\mathscr{W}_{#1,#2}}
\newcommand{\bAStTL}[2]{\overline{\mathscr{W}}_{#1,#2}}
\newcommand{\MmodJTL}[3]{\mathscr{M}^{(#1)}_{#2}(#3)}
\newcommand{\MJTL}[2]{\mathscr{M}_{#1}(#2)}
\newcommand{\NJTL}[2]{\mathscr{N}_{#1}(#2)}
\newcommand{\MJTLs}[2]{\mathscr{M}^*_{#1}(#2)}
\newcommand{\NJTLs}[2]{\mathscr{N}^*_{#1}(#2)}

\newcommand{\NNJTL}[1]{\mathscr{N}_{#1}}

\newcommand{\Bodd}{B_{\q}}
\newcommand{\Nodd}{N_{\q}}
\newcommand{\Boddm}{\Bodd^-}
\newcommand{\Boddp}{\Bodd^+}
\newcommand{\Noddm}{\Nodd^-}
\newcommand{\Noddp}{\Nodd^+}
\newcommand{\Bmod}[2]{\mathsf{B}^{(#1)}_{#2}}
\newcommand{\Rmod}[2]{\mathsf{R}^{(#1)}_{#2}}

 \newxyColor{col-p-1}{0.0 0.0 0.0}{rgb}{}
 \newxyColor{col-p-2}{0.0 0.0 0.0}{rgb}{}
 \newxyColor{col-p-3}{0.0 0.0 0.0}{rgb}{}
 \newxyColor{col-p-4}{0.0 0.0 0.0}{rgb}{}
  \newxyColor{col-p-5}{0.0 0.0 0.0}{rgb}{}
  \newxyColor{col-p-k}{0.0 0.0 0.0}{rgb}{}

 \newxyColor{col-m-1}{0.0 0.0 0.0}{rgb}{}
  \newxyColor{col-m-l}{0.0 0.0 0.0}{rgb}{}
 \newxyColor{col-m-2}{0.0 0.0 0.0}{rgb}{}
 \newxyColor{col-m-3}{0.0 0.0 0.0}{rgb}{}
 \newxyColor{col-m-4}{0.0 0.0 0.0}{rgb}{}
 \newxyColor{col-m-5}{0.0 0.0 0.0}{rgb}{}
 \definecolor{col-p-k}{rgb}{0.0,0.0,0.0}
 \definecolor{col-p-1}{rgb}{0.0,0.0,0.0}
 \definecolor{col-p-2}{rgb}{0.0,0.0,0.0}
 \definecolor{col-p-3}{rgb}{0.0,0.0,0.0}
 \definecolor{col-p-4}{rgb}{0.0,0.0,0.0}
  \definecolor{col-m-l}{rgb}{0.0,0.0,0.0}
 \definecolor{col-m-1}{rgb}{0.0,0.0,0.0}
 \definecolor{col-m-2}{rgb}{0.0,0.0,0.0}
 \definecolor{col-m-3}{rgb}{0.0,0.0,0.0}
 \definecolor{col-m-4}{rgb}{0.0,0.0,0.0}
 
\newcommand{\Approjmodbase}{B}

\newtheorem{Thm}[subsection]{Theorem}
\newtheorem{thm}[subsubsection]{Theorem}

\newtheorem{lemma}[subsubsection]{Lemma}

\newtheorem{prop}[subsubsection]{Proposition}

\theoremstyle{definition}

\newtheorem{dfn}[subsubsection]{Definition}
\newtheorem{rem}[subsubsection]{Remark}
\newtheorem{example}[subsubsection]{Example}

\begin{document}
\begin{center}

\Large{
 Bimodule structure in the periodic $\gl(1|1)$ spin chain}
\vskip 1cm

{\large A.M. Gainutdinov\,$^{a}$, N. Read\,$^{b}$, and
H. Saleur\,$^{a,c}$}

\vspace{1.0cm}

{\sl\small $^a$  Institut de Physique Th\'eorique, CEA Saclay,\\ 
Gif Sur Yvette, 91191, France\\}
{\sl\small $^b$ Department of Physics, Yale University, P.O. Box 208120,\\ New Haven, Connecticut 06520-8120, USA\\}
{\sl\small $^c$ Department of Physics and Astronomy,
University of Southern California,\\
Los Angeles, CA 90089, USA\\}

\end{center}

\begin{abstract}
This paper is the second  in a series devoted to the study of 
periodic super-spin chains.  In our first paper~\cite{GRS1}, we
have studied the symmetry algebra of the periodic $\gl(1|1)$ spin
chain. In technical terms, this spin chain is built out of the
alternating product of the $\gl(1|1)$ fundamental representation and
its dual. The local energy densities -- the nearest neighbor
Heisenberg-like couplings -- provide a representation of the Jones--Temperley--Lieb (JTL) algebra $\JTL{N}$. 
The symmetry algebra is then
the centralizer of $\JTL{N}$, and turns out to be smaller than for the
open chain, since it is now only a subalgebra of $\LQG$ at $\q=i$ --
dubbed $\LQGodd$ in~\cite{GRS1}.  A crucial step in our associative
algebraic approach to bulk logarithmic conformal field theory (LCFT)
is then the analysis of the spin chain as a bimodule over $\LQGodd$
and $\JTL{N}$. While our ultimate goal is to use this bimodule to
deduce properties of the LCFT in the continuum limit, its derivation
is sufficiently involved to be the sole subject of this
paper. We describe representation theory of the centralizer and then
use it to find a decomposition of the periodic $\gl(1|1)$ spin chain
over $\JTL{N}$ for any even $N$ and ultimately a corresponding
bimodule structure. Applications of our results to the analysis of the bulk LCFT
will then be discussed in the third part of this series.

\end{abstract}


\section{Introduction}

The general philosophy of the lattice approach to LCFTs in the
boundary case \cite{ReadSaleur07-1,ReadSaleur07-2} relies on the
analysis of microscopic models -- typically spin chains built out of
alternating representations of a super Lie algebra such as $\gl(m|n)$,
with a nearest neighbour ``Heisenberg" coupling -- as a bi-module over
two algebras. In physical terms, one of these algebras is generated by
the local hamiltonian densities, and the other is the ``symmetry"
commuting with these hamiltonian densities.

It is natural to try to extend this approach~\cite{ReadSaleur07-2} to
the bulk case, but considerable mathematical difficulties are
encountered in this endeavor. This is true even for the -- a priori
simplest -- case of $\gl(1|1)$, whose continuum limit is the
ubiquitous symplectic fermion theory. The local hamiltonian densities
then provide a (non-faithful) representation of the
Jones--Temperley--Lieb algebra $\JTL{N}$. Its centralizer $\centJTL$ was studied
in our previous paper~\cite{GRS1}, where we found that it is generated
 by a subalgebra -- dubbed $\LQGodd$ -- of $\LQG$ at $\q=i$ (recall that
the centralizer
of the ordinary Temperley--Lieb algebra  in the open case
is $\LQGi$) and two operators mapping to each other the lowest and highest spin states. Concerning the representations of $\centJTL$, the essential part is contained in the $\LQGodd$, and we will sometimes abuse notations by calling $\LQGodd$ the centralizer of $\JTL{N}$. Note that in the scaling limit~\cite{GRS1}, the difference between the two objects becomes irrelevant.  The next
step in the program consists thus in decomposing the spin chain as a
bimodule over this $\LQGodd$ and $\JTL{N}$ -- a rather technical task
we tackle in this paper, leaving the discussion of the (many) physical implications to a
sequel~\cite{GRS3}.

The plan of the paper is as follows. After preliminaries and
reminders of various definitions in Sec.~\ref{sec:prelim}, we
explore the representation theory of the centralizer in
Sec.~\ref{sec:rep-th-centJTL}. The representation theory of the
Jones--Temperley--Lieb algebra is then summarized in
Sec.~\ref{sec:JTL-stand}, largely based upon the seminal work of
Graham and Lehrer~\cite{GL}. Section~\ref{sec:sp-ch-decomp-JTL} is
devoted to the spin-chain decomposition over $\JTL{N}$. Considerable attention is paid to the
absence of the  double-centralizing property (a familiar aspect of  the semi-simple case),
and the ensuing technical complications for our analysis of  the  $\gl(1|1)$ spin-chain.  All
these elements are put together in Sec.~\ref{ind-chain-bimod-sec}
where the bimodule structure is finally obtained for the periodic
model. The twisted model with antiperiodic boundary conditions is also
decomposed as a (now, semisimple) bimodule over two centralizing
algebras in  Sec.~\ref{ind-chain-bimod-sec}. A few conclusions are gathered
in Sec.~\ref{sec:concl}.

\subsection{Notations}
To help the reader navigate through this long paper, we provide a (partial) list of notations:

\begin{itemize}

\item[\mbox{}] $\TL{N}$ --- the (ordinary) Temperley--Lieb algebra,

\item[\mbox{}] $\ATL{N}$ --- the periodic Temperley--Lieb algebra with the translation $u$, or the algebra of affine diagrams, 

\item[\mbox{}] $\JTL{N}(m)$ --- the Jones--Temperley--Lieb algebra with parameter $m$,

\item[\mbox{}] $\JTL{N}$ --- the Jones--Temperley--Lieb algebra at $m=0$,

\item[\mbox{}] $\centJTL$ --- the  centralizer of $\JTL{N}$ in the $\gl(1|1)$ spin chain,

\item[\mbox{}] $\repgl$ --- the spin-chain representation of $\JTL{N}$,

\item[\mbox{}] $\repQG$ --- the spin-chain representation of the quantum group $\LQG$,

\item[\mbox{}] $\E$, $\F$,  $\K^{\pm1}$ --- the standard quantum group generators,

\item[\mbox{}] $\e$, $\f$ --- the renormalized powers of the generators $\E$ and $\F$,

 \item[\mbox{}] $\XX_{1,n}$ ---  the simple $\LQGi$-modules,

 \item[\mbox{}] $\PP_{1,n}$ ---  the projective $\LQGi$-modules,

 \item[\mbox{}] $\Xodd{n}$ --- the simple $\LQGoddi$- and $\centJTL$-modules
 
 \item[\mbox{}] $\PPodd{n}$ ---  the indecomposable summands in the spin-chain decomposition over the centralizer $\centJTL$,

\item[\mbox{}] $\StTL{j}$ --- the standard modules over $\TL{N}$,
 
\item[\mbox{}] $\PrTL{j}$ --- the projective modules over $\TL{N}$,


\item[\mbox{}] $\AIrrTL{j}{(-1)^{j+1}}$ --- the simple modules over $\JTL{N}$
  for which we also use the notation $(\TLX_{j})$,

\item[\mbox{}] $\AStTL{j}{(-1)^{j+1}}$ --- the standard modules over $\JTL{N}$,

\item[\mbox{}] $\AStTL{j}{e^{2iK}}$ --- the standard modules over $\JTL{N}(m)$

\item[\mbox{}] $\bAStTL{0}{\q^2}$ --- the standard module over $\JTL{N}(m)$ for $j=0$. 

\item[\mbox{}] $\APrTL{j}$ --- the indecomposable summands in spin-chain decomposition over $\JTL{N}$.
\end{itemize}

\section{Preliminaries}\label{sec:prelim}

\subsection{The Temperley--Lieb algebras in the periodic case}
\label{sec:TL-alg-def}
The models we are interested in have transfer matrices (Hamiltonians)
expressed in terms of Temperley--Lieb generators. In the periodic case,
several variants of this algebra can be considered, and it is useful
to start by going over a few definitions.

We begin with an algebra  generated by the $e_j$'s together with the identity, subject to the usual relations
\begin{eqnarray}\label{TL-rel1}
e_j^2&=&me_j,\nonumber\\
e_je_{j\pm 1}e_j&=&e_j,\label{TL}\\
e_je_k&=&e_ke_j\qquad(j\neq k,~k\pm 1),\nonumber
\end{eqnarray}
where $j=1,\ldots,N$; $m$ is a (real) parameter,  and the indices are interpreted modulo $N$. This
algebra is a quotient of the affine Hecke algebra of $A$-type and
denoted by $\PTL{N}$ in the work of Graham and Lehrer \cite{GL,GL1}
whose definitions and notations we follow whenever
possible. The algebra $\PTL{N}$ is also known as the periodic
Temperley--Lieb algebra~\cite{MartinSaleur,MartinSaleur1}.

\begin{figure}
\begin{equation*}
 \begin{tikzpicture}
 	\draw[thick, dotted] (-0.05,0.5) arc (0:10:0 and -7.5);
 	\draw[thick, dotted] (-0.05,0.55) -- (2.65,0.55);
 	\draw[thick, dotted] (2.65,0.5) arc (0:10:0 and -7.5);
	\draw[thick, dotted] (-0.05,-0.85) -- (2.65,-0.85);
	\draw[thick] (0,0) arc (-90:0:0.5 and 0.5);
	\draw[thick] (0.9,0.5) arc (0:10:0 and -7.6);
	\draw[thick] (1.65,0.5) arc (0:10:0 and -7.6);
	\draw[thick] (2.6,0) arc (-90:0:-0.5 and 0.5);

	\draw[thick] (0.5,-0.8) arc (0:90:0.5 and 0.5);
	\draw[thick] (2.1,-0.8) arc (0:90:-0.5 and 0.5);
	\end{tikzpicture}\;\;,
	\qquad\qquad
 \begin{tikzpicture}
 	\draw[thick, dotted] (-0.05,0.5) arc (0:10:0 and -7.5);
 	\draw[thick, dotted] (-0.05,0.55) -- (2.65,0.55);
 	\draw[thick, dotted] (2.65,0.5) arc (0:10:0 and -7.5);
	\draw[thick, dotted] (-0.05,-0.85) -- (2.65,-0.85);
	\draw[thick] (0,0) arc (-90:0:0.5 and 0.5);
	\draw[thick] (0.8,0.5) arc (-180:0:0.5 and 0.5);
	\draw[thick] (2.6,0) arc (-90:0:-0.5 and 0.5);

	\draw[thick] (0.5,-0.8) arc (0:90:0.5 and 0.5);
	\draw[thick] (1.8,-0.8) arc (0:180:0.5 and 0.5);
	\draw[thick] (2.1,-0.8) arc (0:90:-0.5 and 0.5);

	\end{tikzpicture}\;\;,
	\qquad\qquad
 \begin{tikzpicture}
 	\draw[thick, dotted] (-0.05,0.5) arc (0:10:0 and -7.5);
 	\draw[thick, dotted] (-0.05,0.55) -- (2.65,0.55);
 	\draw[thick, dotted] (2.65,0.5) arc (0:10:0 and -7.5);
	\draw[thick, dotted] (-0.05,-0.85) -- (2.65,-0.85);
	\draw[thick] (0,0.1) arc (-90:0:0.5 and 0.4);
	\draw[thick] (0,-0.1) arc (-90:0:0.9 and 0.6);
	\draw[thick] (2.6,-0.1) arc (-90:0:-0.9 and 0.6);
	\draw[thick] (2.6,0.1) arc (-90:0:-0.5 and 0.4);
	
	\draw[thick] (0.5,-0.8) arc (0:90:0.5 and 0.5);
	\draw[thick] (1.8,-0.8) arc (0:180:0.5 and 0.5);
	\draw[thick] (2.1,-0.8) arc (0:90:-0.5 and 0.5);
	\end{tikzpicture}
\end{equation*}
\caption{Examples of affine diagrams for $N=4$, with the left and right sides of the framing rectangle identified. The first diagram represents the generator $e_4$, it has rank $2$ as well as the second one. The third diagram has  rank $3$.}
\label{fig:aff-diag}
\end{figure}
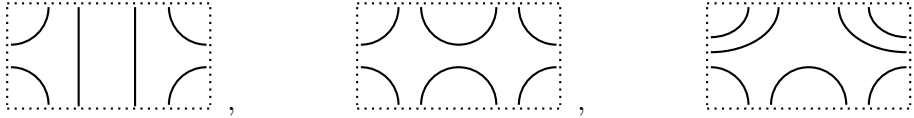

The $e_i$'s can be interpreted in terms of particular diagrams
on an annulus~\cite{MartinSaleur,MartinSaleur1,GL,Green,Jones} (a representation which is known to be faithful
\cite{FanGreen}).  A general basis element in the space of
diagrams we will be interested in is obtained by taking $N$ sites on the inner and $N$ sites on
the outer boundary of the annulus; these sites are connected in pairs,
and only configurations that can be represented using lines inside the
annulus without crossings are allowed. Diagrams related by an isotopy
leaving the labeled sites fixed are considered equivalent. We call such
(equivalence classes of) diagrams  \textit{affine} diagrams. Examples of affine diagrams are shown in Fig.~\ref{fig:aff-diag}, where we draw them in slightly different geometry: we cut the annulus and transform it to a rectangle which we call \textit{framing} so that the sites labeled by `$1$' are closest to the left and sites labeled by `$N$' are to the right sides of the rectangle. Multiplication of two diagrams
can be then defined by joining an inner to an outer annulus, and
removing the interior sites. Whenever a closed contractible loop is
produced  in this multiplication, it is 
replaced by a numerical factor $m$.  This defines abstractly an associative
algebra  which we denote as $\ATL{N}(m)$. Note that the  diagrams
in this algebra allow winding of through-lines around the annulus any
integer number of times, and different windings result in independent
algebra elements. Moreover, in the ideal of zero through-lines, any number of
non-contractible loops is allowed. The algebra $\ATL{N}(m)$ is
thus infinite-dimensional.

The action of the $e_i$ generators in the diagram basis is well known \cite{FanGreen}. 
Once in the sector with $N$ through-lines of the diagram algebra $\ATL{N}$, we consider the
generators $u$ and $u^{-1}$ of translations by one site to the right
and to the left, respectively.  The following additional defining
relations are then obeyed,
\begin{equation}\label{aTL-rel}
\begin{split}
ue_ju^{-1}&=e_{j+1},\\
u^2e_{N-1}&=e_1\ldots e_{N-1},
\end{split}
\end{equation}
and $u^{\pm N}$ is a central element.  The algebra generated by the
 $e_i$ and $u^{\pm1}$ with the defining relations~\eqref{TL}
 and~\eqref{aTL-rel} is isomorphic to $\ATL{N}(m)$  and
 called  the \textit{affine} Temperley--Lieb algebra.

We call \textit{rank}~\cite{GL} (see also~\cite{Green}) of an affine
diagram the minimal number of intersections with  
the left side of the framing rectangle, see examples shown in Fig.~\ref{fig:aff-diag}.
The algebra $\PTL{N}$
introduced in~\eqref{TL} is spanned by all affine diagrams of
even-rank in sectors with number of through-lines less than $N$ and by
the identity in the sector with $N$ through-lines. Nothing is said at
this stage about non contractible loops and windings of through-lines, and the algebra $\PTL{N}$ is
also infinite dimensional.

 \bigskip
 
 For the models we are interested in, with Hilbert spaces built out of
  (tensor products of) alternating representations, $N=2L$ is
  even. Moreover, the pattern of representations forces one to
  consider translations by an even number of sites only, {\it i.e.}, restrict
  to powers of $u^2$. This leads to a subalgebra
  $\oATL{N}(m)\subset\ATL{N}(m)$ spanned by all affine diagrams of even rank\footnote{The algebra
  $\oATL{N}(m)$ can be alternatively described as an algebra of diagrams
  with orientable lines (such that the arrows emanating from the even
  sites enter the odd sites on the inner boundary, and the reverse for
  the outer boundary), modulo odd-rank diagrams  in the ideal without
  through-lines.}.
 Physical applications require actually the consideration of further finite-dimensional quotients of the
$\oATL{N}(m)$. The   easiest way to define such
quotients is to consider a homomorphism $\psi$ to the Brauer
algebra~\cite{[Br]}. Recall first that the Brauer algebra  is defined as the algebra of diagrams drawn
inside a rectangle with lines connecting  two identical or opposite
edges, say the bottom and the top ones, with $N$ sites on each of them and
allowing any crossings, up to isotopy leaving the labeled sites fixed as usual. The homomorphism $\psi$ takes an even rank annular diagram and produces a rectangular diagram with crossings in the following way: we first cut the annulus such that the diagram is now inside the framing rectangle defined above and then we connect the point (of an arc or a through-line) on the left side with its corresponding point on the right side of the rectangle. For example, we have
\begin{equation}
  \xymatrix@C=8pt@R=1pt@M=-5pt@W=-2pt{
  &&	\mbox{}\quad\xrightarrow{{\mbox{}\quad\psi\quad} }\quad &\\
  & {\begin{tikzpicture}
 	\draw[thick, dotted] (-0.05,0.5) arc (0:10:0 and -7.5);
 	\draw[thick, dotted] (-0.05,0.55) -- (2.65,0.55);
 	\draw[thick, dotted] (2.65,0.5) arc (0:10:0 and -7.5);
	\draw[thick, dotted] (-0.05,-0.85) -- (2.65,-0.85);
	\draw[thick] (0,0.2) arc (-90:0:0.5 and 0.3);
	\draw[thick] (0.0,-0.1) arc (-90:0:1.1 and 0.6);
	\draw[thick] (1.65,0.5) arc (0:10:38.5 and -7.6);
	\draw[thick] (2.6,0.2) arc (-90:0:-0.5 and 0.3);
	\draw[thick] (0.5,-0.8) arc (0:90:0.5 and 0.3);
	\draw[thick] (1.7,-0.8) arc (0:90:-0.9 and 0.6);
	\draw[thick] (2.1,-0.8) arc (0:90:-0.5 and 0.3);
	\end{tikzpicture}\quad}&
	& {	\quad
  \begin{tikzpicture}
 	\draw[solid] (-0.0,0.51) arc (0:10:0 and -7.6);
 	\draw[solid] (-0.0,0.51) -- (2.65,0.51);
 	\draw[solid] (2.65,0.51) arc (0:10:0 and -7.6);
	\draw[solid] (-0.0,-0.81) -- (2.65,-0.81);
	\draw[thick] (0.5,0.5) arc (-180:0:0.8 and 0.5);
	\draw[thick] (1.0,0.5) arc (0:10:-44.5 and -7.6);
	\draw[thick] (1.65,0.5) arc (0:10:38.5 and -7.6);
	%
	\draw[thick] (2.1,-0.8) arc (0:180:0.8 and 0.5);
	\end{tikzpicture}}
  }
\end{equation}
where the right diagram is an element of the Brauer algebra.
The image of
$\oATL{N}(m)$ under the
homomorphism $\psi$ is thus a subalgebra of  diagrams that can be drawn in the annulus
without crossings (so they might have crossings in the rectangle, as a subalgebra in the Brauer algebra) {\sl plus} additional relations~\cite{Jones}: (i)  non contractible
loops are replaced
by the same numerical factor $m$ as  for contractible loops; (ii)
$u^N=1$ (this  allows one  to ``unwind" through-lines of the
affine diagrams); (iii) non-isotopic (in the annulus) diagrams connecting  the same sites
 are identified.
We call the finite-dimensional image of $\psi$ \textit{the
Jones--Temperley--Lieb algebra} $\JTL{N}(m)$ (actually used
in~\cite{ReadSaleur07-1}), and it is the object we mostly want to study  in this
paper.  This algebra was first introduced in~\cite{Jones} and called
oriented annular subalgebra in the Brauer algebra.

For further references, we gather all the mentioned algebras in the diagram
\begin{equation}\label{diag-alg}
     \xymatrix@C=28pt@R=16pt@M=5pt@W=4pt{
       &{\TL{N}\;\;}\ar@{^{(}->}[r] & {\;\;\PTL{N}\;\;}\ar@{^{(}->}[r]
       & {\;\;\ATL{N}\;\;} 
       &{\;\;\oATL{N}\;} \ar@{_{(}->}[l] \ar@{->>}[r]^(0.45){\psi} & {\;\JTL{N}}&
    } 
\end{equation}
where we also introduced the  notation for the open Temperley--Lieb algebra
$\TL{N}$ generated by $e_j$, for $1\leq j\leq N-1$;
the arrows
$\hookrightarrow$ denote embeddings of algebras while the doubled
arrows denote projections (surjective homomorphisms of algebras).

We will only be concerned in this paper with the case $m=0$ for which
 the algebra $\JTL{2L}(m)$ is non semi-simple; in the following we usually  suppress all
 reference to $m$. We will also mostly restrict to  a specific ``tensor product" representation - the alternating  $\gl(1|1)$ spin chain.

 \subsection{The closed $\gl(1|1)$ super-spin chain}\label{sec:super-spin-ch-def}
 
The closed  $\gl(1|1)$ super-spin chain \cite{ReadSaleur07-1,GRS1} is a
 tensor product 
 representation $\chVv=\tensor_{j=1}^{N}\oC^2$ of the algebra $\JTL{N}(0)$, which consists of
 $N=2L$ tensorands 
labelled $j=1,\ldots,2L$ with the fundamental representation of $\gl(1|1)$ on even sites and its dual on odd sites.
The representation of each $e_j$ is given by the operator mapping   the
product of two neighbour tensorands on the $\gl(1|1)$-invariant 
\begin{equation}\label{rep-JTL-1}
e_j^{\gl}= (f_j+f_{j+1})(f_j^\dagger+f_{j+1}^\dagger),\qquad 1\leq
j\leq 2L.
\end{equation}
Here we used a free fermion representation based on operators $f_j$
and $f_j^\dagger$ acting non-trivially only on $j$th tensorand and obeying
\begin{eqnarray}\label{f_j-rel}
\{f_j,f_{j'}\}=0,\quad\{f^{\dagger}_j,f^{\dagger}_{j'}\}=0,\quad\{f_j,f_{j'}^\dagger\}=(-1)^{j}\delta_{jj'},\qquad
f_{2L+1}=f_{1}, \quad f^\dagger_{2L+1}=f^{\dagger}_{1},
\end{eqnarray}
where the minus sign  for an odd $j$ is due to presence of the dual representations of $\gl(1|1)$.

The generators $e_j^{\gl}$ satisfy the (periodic) Temperley--Lieb algebra relations~\eqref{TL-rel1}
with $m=0$, and together with the generator $u^2$ translating the
periodic spin-chain by two sites $j\to j+2$, they provide a
representation of $\JTL{2L}(m=0)$ which we denote by
$\repgl:\JTL{2L}(0)\to\Endo_{\oC}(\chVv)$.
 The representation $\repgl$
is known to be non-faithful~\cite{ReadSaleur07-1}.
 
 The representation space $\Hilb_{2L}$ is equipped with an  inner product $\inn{\cdot}{\cdot}$ such that $\inn{f_j x}{y}=\inn{x}{f_j^{\dagger}y}$ for any $x,y\in\Hilb_{2L}$.
 We stress that the inner product is indefinite because of the sign factors in the relations~\eqref{f_j-rel}. Then, 
 the Hamiltonian operator
 \begin{equation}\label{hamil-def}
 H=-\sum_{j=1}^{2L}e^{\gl}_j,
 \end{equation}
with the ``hamiltonian densities" $e^{\gl}_j$ defined in~\eqref{rep-JTL-1}, is self-adjoint $H=H^{\dagger}$ with respect to this inner product (actually, each $e_j^{\gl}$ is a self-adjoint operator).
Its eigenvalues are real and the eigenvectors are computed in~\cite{GRS1} using a relation with XX
spin-chains. It was also shown that the Hamiltonian has non-trivial Jordan cells (of rank-two).

\subsection{Centralizers and bimodules}\label{subsec:imp-bimod}

As was mentioned in the introduction, an important step in our approach is to find a decomposition of the spin-chain over the $\JTL{N}$ for any finite $N$. The representation $\repgl$ of $\JTL{N}$ is
non-faithful and there are thus no direct evident ways of getting the
decomposition of the spin-chain, unlike in the open case where
one deals with  a faithful representation of $\TL{N}$. For example, the general theory~\cite{GL0} of projective  modules over
a cellular algebra (which includes $\TL{N}(m)$ and $\JTL{N}(m)$ algebras) could
be applied in a faithful
representation. In our non-faithful case, we need an indirect strategy,
which uses the symmetry algebra as  discussed  below. In turn, the use of this indirect strategy is made  complicated by the fact that we deal with the non semi-simple representation of an associative algebra. Our problem is thus rather complicated.

\medskip
In general, an important concept in  lattice models is the full symmetry
algebra which is technically the centralizer of  a ``hamiltonian
densities" algebra of the model.
  By 
the latter algebra we generally mean any (representation of a) Hecke-type algebra -- mostly 
  $\TL{N}(m)$ for  open spin-chains or $\JTL{N}(m)$ for closed
ones.
 We recall that, 
for  an associative algebra $A$ and its representation space~$\Hilb$, \textit{the~centralizer} of $A$ is an algebra
  $\cent_{A}$ of the maximum dimension such that $[\cent_{A},A]=0$,
  {\it i.e.}, the centralizer is defined as $\cent_{A}\cong\Endo_{A}(\Hilb)$
  -- the algebra of all endomorphisms on $A$-module $\Hilb$.

The representation theory of the centralizer $\cent_A$ is usually much
easier to study than  the representation theory of the  ``hamiltonian
densities" algebra $A$. It is thus more reasonable to start with a
decomposition of spin-chains over $\cent_A$ into indecomposable direct summands, which are in general what are technically called  tilting modules~\cite{Donkin}.
Note here that strictly speaking the tilting modules are defined for a quasi-hereditary algebra, which is the case for the centralizers of the TL algebras, and for JTL representations in closed alternating $\gl(n|m)$ spin chains only if $n+m>2$. So for the $\gl(1|1)$ closed case the ``spin-chain modules" are not tilting, though we will sometimes still call them ``tilting", abusing notations (the concept of tilting modules is very powerful and will be used for studying $\gl(2|1)$ spin chains in forthcoming papers, see also~\cite{GJSV} for a short review in the  context of boundary spin chains.) 

The next step is to  study all
homomorphisms between the direct summands in the decomposition 
 to obtain  the module structure over the ``hamiltonian densities''
algebra $A$. In particular, multiplicities in front of tilting
$\cent_A$-modules give the  dimensions of simple $A$-modules, and the
subquotient structure of projective $A$-modules can be deduced from
the one of the tilting $\cent_A$-modules, see~\cite{AF-book}.
As a result, one  gets a sequence of bi-modules $\Hilb_N$ over the two commuting algebras
parametrized by the number $N$ of sites/tensorands in the
spin-chain. 

This approach however requires the double-centralizing property -- that is, that  the algebra $A$ is the centralizer
of its centralizer $\cent_A$.  This property, which holds in the semi-simple case, is  not obvious for non-semisimple representations of an associative algebra $A$, and in fact does not hold in our case.
Our problem is  thus more complicated.  While we first  follow the general  strategy by studying the decomposition of the spin chain over  $\cent_A$   in  section \ref{sec:rep-th-centJTL}, the analysis of the  decomposition of the spin chain over $\JTL{N}$ (see section \ref{sec:sp-ch-decomp-JTL}) requires some extra steps where we start by proposing a subquotient structure for the spin chain $\JTL{N}$ modules, and then check consistency and uniqueness based on our previous analysis of the centralizer.

\medskip
As a simple example, the open $\gl(1|1)$ spin-chain exhibits a large symmetry algebra
dubbed ${\cal A}_{1|1}$ in~\cite{ReadSaleur07-1}. This algebra is the
centralizer $\centTL$ of $\TL{N}(0)$ and is
generated by the identity and the five generators
\begin{gather*}
F_{(1)}=\sum_{1\leq j\leq N} f_j,\qquad
F^\dagger_{(1)}=\sum_{1\leq j\leq N} f_j^\dagger,\\
F_{(2)}=\sum_{1\leq j<j'\leq N}f_jf_{j'},\qquad
F_{(2)}^\dagger=\sum_{1\leq j<j'\leq N} f_{j'}^\dagger f_j^\dagger,\qquad
{\N}=\sum_{1\leq j\leq N} (-1)^jf_j^\dagger f_j.
\end{gather*}
We note that these formulas give just a representation of the quantum
group $\LQG$ for $\q=i$. The fermionic generators, with the index
$(1)$, are from the nilpotent part and the bosonic ones form the
$s\ell(2)$ subalgebra in $\LQG$ (see a precise correspondence below
in~\eqref{QG-ferm-2} and~\eqref{QG-ferm-1}.)

The decomposition of the open spin-chain as a bimodule over the pair
$(\TL{N},{\cal A}_{1|1})$ of mutual centralizers 
is shown on Fig.~\ref{openbimodule-fin} for  $N=8$ case (borrowed
from~\cite{ReadSaleur07-2}).
 \begin{figure}\centering
 \leavevmode
 \epsfysize=80mm{\epsffile{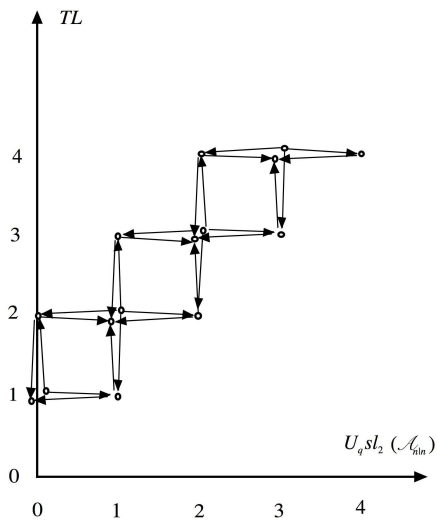}}
  \caption{The structure of the open $\gl(1|1)$ spin-chain for  $N=8$
  sites, as a representation of $\TL{N}\boxtimes\LQGi$. Each node with a Cartesian coordinate $(n,n')$ corresponds to the tensor product $(\TLX_{n'})\boxtimes\XX_{1,n+1}$, see notations in Sec~\ref{sec:rep-th-centJTL}. Some nodes
  with Cartesian coordinates $(n,n+1)$ occur twice and those nodes
  have been separated slightly for clarity. 
}\label{openbimodule-fin}
 \end{figure}
Each node with a Cartesian coordinate $(n,n')$ in the bimodule diagram
corresponds to a simple subquotient $(\TLX_{n'})\boxtimes\XX_{1,n+1}$ over the tensor product $\TL{N}\boxtimes\LQGi$ of associative algebras and
arrows show the action of both algebras -- the Temperley--Lieb $\TL{N}$ acts
in the vertical direction (preserving the coordinate $n$), while
$\LQGi$ acts in the horizontal way. Indecomposable projective $\TL{N}$-modules
$\PrTL{n'}$ (which are discussed below in
Sec.~\ref{sec:TL-decomp}) can be recovered by ignoring all the
horizontal arrows, while tilting $\LQGi$-modules $\PP_{1,n+1}$ are obtained by ignoring all the vertical arrows
of the bimodule diagram (these are also projective and given
in~\eqref{schem-proj}.) Having the decomposition over $\LQGi$, we see that the subquotient structure of direct summands
over $\TL{N}$ is obtained by drawing arrows corresponding to all possible homomorphisms between the tilting modules.

\medskip
In the closed case, while the $\gl(1|1)$ symmetry 
generated by $F_{(1)}$, $F^\dagger_{(1)}$, and ${\N}$ remains, the
``bosonic" $s\ell(2)$ generators $F_{(2)}$ and $F^\dagger_{(2)}$ do not
commute with the action of $\JTL{N}(0)$. Instead, we have essentially
only a ``fermionic" subalgebra of ${\cal A}_{1|1}$ that generates the
centralizer of $\JTL{N}(0)$.  We next describe in detail the
centralizer of (the representation $\repgl$ of) $\JTL{N}$ obtained
first in our previous paper~\cite{GRS1} where it is realized as a
subalgebra of the quantum group $\LQG$.

\subsection{The centralizer of $\JTL{N}(0)$}
 Recall first that
the \textit{full} quantum group $\LQG$ with $\q = e^{i\pi/p}$ and
an integer $p\geq2$ is generated by $\E$, $\F$, $\K^{\pm1}$, and $\e$, $\f$,
$\h$. The first three generators satisfy the standard quantum-group relations
\begin{equation*}
  \K\E\K^{-1}=\q^2\E,\quad
  \K\F\K^{-1}=\q^{-2}\F,\quad
  [\E,\F]=\ffrac{\K-\K^{-1}}{\q-\q^{-1}},
\end{equation*}
with additional relations
\begin{equation*}
  \E^{p}=\F^{p}=0,\quad \K^{2p}=\one,
\end{equation*}
and  the divided powers $\f\sim \F^p/[p]!$ and $\e\sim \E^p/[p]!$ satisfy the usual $s\ell(2)$-relations:
\begin{equation*}
  [\h,\e]=\e,\qquad[\h,\f]=-\f,\qquad[\e,\f]=2\h.
\end{equation*}
The full list of relations with comultiplication formulae are borrowed
from~\cite{BFGT} and listed in App.~A where we also give relations for
the related  quantum group generators
$S^{\pm}$, $S^z$ and $\q^{S^z}$ more common in spin chain literature. We will use in the text only the
notation $S^z$ which is proportional to $\h$ as $2\h=S^z$.

 For applications to $\gl(1|1)$
spin-chains, we consider only the case $p=2$ and set  in what follows
$\q\equiv i$. 
As a module over $\LQG$, the spin
chain $\chVv$ is a tensor product of $N$ copies of
two-dimensional irreducibe representations defined as $\E\mapsto\sigma^+=\sigmap$,
$\F\mapsto\sigma^-=\sigmam$, $\K\mapsto \q\sigma^z=\sigmaq$, and $\e\mapsto0$, $\f\mapsto0$.
Using the $(N-1)$-folded
comultiplications~\eqref{N-fold-comult-cap},
\eqref{N-fold-comult-ren-e}, and~\eqref{N-fold-comult-ren-f}, 
we obtain
the representation $\repQG:\LQG\to\Endo_{\oC}(\chVv)$ in terms of
the  operators $f_j$ and $f^{\dagger}_j$ defined in Sec.~\ref{sec:super-spin-ch-def},
\begin{align}
\repQG(\h)&=\half\sum_{j=1}^{N} (-1)^jf_j^\dagger f_j-\ffrac{L}{2},\;
&\repQG(\e) &= \q^{-1}\sum_{1\leq j_1<j_2\leq N}f_{j_1}^\dagger f_{j_2}^\dagger,\;
&\repQG(\f) &= \q\sum_{1\leq j_1<j_2\leq N}f_{j_1}f_{j_2},\label{QG-ferm-2}\\
\repQG(\K)&=(-1)^{2\repQG(\h)},\;
&\repQG(\E)& = \sum_{j=1}^N f_j^\dagger\,\repQG(\K),\;
&\repQG(\F) &= \q^{-1}\sum_{j=1}^N f_j.\label{QG-ferm-1}
\end{align}

\begin{dfn}\label{dfn:Uqodd}
We now introduce an associative algebra $\LQGodd$, with $\q=i$. The algebra $\LQGodd$ is
generated by $\FF n$, $\EE m$ ($n,m\in\oN\cup\{0\}$), $\K^{\pm1}$, and
$\h$ with the following defining relations
\begin{gather}
\K \EE m \K^{-1} = \q^2 \EE m,\qquad \K \FF n \K^{-1} = \q^{-2} \FF n,
\qquad \K^4=\one,\label{Uqodd-dfn-1}\\
[\EE m,\FF n] = \sum_{r=1}^{\text{min}(n,m)}P_r(\h)\FF{n-r}\EE{m-r},\\
\EE{m}\EE{n} = \EE{n}\EE{m} = 0,\quad \FF{m}\FF{n}= \FF{n}\FF{m} =
0,\quad [\K,\h]=0,\label{Uqodd-dfn-3}\\
[\h,\EE m] = (m + \ffrac{1}{2})\EE{m}, \qquad [\h,\FF n] = -(n + \ffrac{1}{2})\FF{n},\label{Uqodd-dfn-4}
\end{gather}
where $P_r(\h)$ are polynomials in $\h$ obtained from the usual $s\ell(2)$
relation
$[\e^m,\f^{n}]=\sum_{r=1}^{\text{min}(n,m)}P_r(\h)\f^{n-r}\e^{m-r}$,
and we assume that $\sum_{r=1}^{0}f(r)=0$.

The algebra $\LQGodd$ has the PBW basis $\EE n \FF m \h^k\K^l$, with
$n,m,k\geq0$ and $0 \leq l \leq 3$.
The positive Borel
subalgebra is generated by $\h$, $\K$ and $\EE n$ while the negative subalgebra -- by $\h$, $\K$ and $\FF n$, for $n\geq0$.
\end{dfn}

\begin{rem}\label{rem:inj-hom-LQGodd}
There is an injective homomorphism $\LQGodd\to\LQG$ of associative
algebras defined as
\begin{equation}\label{LQGodd-LQG-hom}
\EE m \mapsto \e^m\E\,\ffrac{\K^2+\one}{2}, \qquad \FF n \mapsto
\f^n\F\,\ffrac{\K^2+\one}{2},\qquad m,n\geq0.
\end{equation}
This homomorphism together with expressions~\eqref{QG-ferm-2}
and~\eqref{QG-ferm-1} defines  by restriction a representation of
$\LQGodd$ on the space $\chVv$ which we also denote by $\repQG$. The
representation $\repQG$ of
$\LQGodd$ is given in~\cite{GRS1} in terms of the fermionic operators $f_j$ and $f^{\dagger}_j$.
\end{rem}

We next recall the result~\cite{GRS1} about the
centralizer of the $\JTL{2L}(0)$.
\begin{thm}\label{Thm:centr-JTL-main}\cite{GRS1}
Fix $\q=i$ and let 
$\cent$ be the subalgebra in $\repQG\bigl(\LQG\bigr)$
generated by $\LQGodd$ and $\f^L$, $\e^L$. On the alternating
periodic $\gl(1|1)$ spin chain $\Hilb_{2L}$,
 the centralizer $\centJTL$ of the image of the
Jones--Temperley--Lieb algebra $\repgl\bigl(\JTL{2L}(0)\bigr)$ is the associative algebra $\cent$, where
$\repgl$ is defined in~\eqref{rep-JTL-1}.
\end{thm}

We rely below on  the representation theory of the
$\JTL{N}$-centralizer $\centJTL$ to study the  decomposition of the
periodic spin-chain into indecomposable $\JTL{N}$-modules. The question of what replaces the appealing bi-module
structure known to exist in the open case when one turns to periodic
systems is the subject of the following three sections.

\section{Representation theory of the centralizer $\centJTL$}\label{sec:rep-th-centJTL}
We now briefly describe  the representation theory of 
$\centJTL$ (which coincides up to trivial details  due to  the extra $\f^L$, $\e^L$
with the  representation theory of $\LQGodd$). We begin with recalling  the decomposition of the spin-chain
$\Hilb_{2L}$ over $\LQG$ and then we describe all simple subquotients
over $\centJTL$ occuring in the decomposition.  We then
use this in studying particular indecomposable modules constituting
blocks in a spin-chain decomposition over the centralizer $\centJTL$.
 We give a decomposition
over $\LQGodd$ in Sec.~\ref{subsec:PPodd-def} and describe spaces of
intertwining operators among indecomposable direct summands in the
decomposition in Sec.~\ref{sec:Tn-homo}, where we also give
 important facts about 
 extensions (``glueings") among simple $\LQGodd$-modules.

\subsection{Spin-chain decomposition over $\LQG$}\label{subsec:LQG-decomp-def}

We first recall the decomposition of $\Hilb_{2L}$ over the full quantum
group $\LQG$ (which is relevant to the open case~\cite{ReadSaleur07-2}), in the
representation $\repQG$ defined in~\eqref{QG-ferm-2}
and~\eqref{QG-ferm-1} (we suppress usually the notation $\repQG$ in
the text below and write simply $\E$ instead of $\repQG(\E)$, {\it etc.})
\begin{equation}\label{decomp-LQG}
\Hilb_{N}|_{\rule{0pt}{7.5pt}%
\LQG} =  \bigoplus_{j=1}^{L}\IrrTL{j}\boxtimes \PP_{1,j},\qquad \quad
 N=2L,
\end{equation}
 where multiplicities $d^0_j=\sum_{i=j}^L(-1)^{j-i}\left(\binom{N}{L+i}
 - \binom{N}{L+i+1}\right)$ are dimensions of irreducibles over
 $\TL{2L}(0)$. The indecomposable direct summands $\PP_{1,j}$ in the decomposition are projective
 covers of simple modules $\XX_{1,j}$ which are introduced in
 App.~\Approjmodbase~with the $\LQG$-action given
 in~\eqref{eq:sl-irrep} (a module $\XX_{1,j}$ has a trivial action of
 $\E$, $\F$, and $\K$, while it is a $j$-dimensional simple $s\ell(2)$-module.)
  We recall the subquotient
structure of $\PP_{1,n}$ is then
\begin{equation}\label{schem-proj}
\includegraphics[scale=0.9]{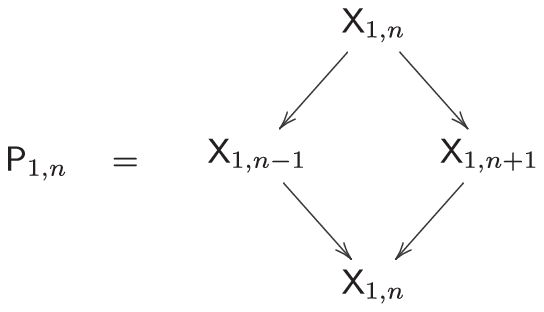}
\end{equation}
with the $\LQG$-action explicitly described in
App.~\Approjmodbase~which is the particular case $\q=i$
of~\cite{BFGT}. In the diagram~\eqref{schem-proj}, we assume
$\XX_{1,0}\equiv0$.

\subsection{Simple modules over $\LQGodd$}
We now describe simple modules over $\LQGodd$ occurring in the spin-chain
decomposition. Using Rem.~\ref{rem:inj-hom-LQGodd},
we consider the restriction to
the subalgebra $\LQGodd$ in a simple  $\LQG$-module $\XX_{1,r}$.
The action~\eqref{eq:sl-irrep} on  $\XX_{1,r}$ where the generators $\E$ and $\F$ act
trivially, and thus $\EE{n}$ and $\FF{m}$ do the same, proves that the restriction  decomposes onto one-dimensional subspaces 
\begin{equation*}
\XX_{1,r}|_{\rule{0pt}{7.5pt}%
\LQGodd} = \bigoplus_{n=1-r}^{r-1}\Xodd{n},
\end{equation*}
where we introduced the  notation $\Xodd{n}$  for simple modules 
over $\LQGodd$. These one-dimensional
modules  are parametrized by the weight $n$
with respect to the Cartan generator $2\h=S^z$.

With the use of the decomposition~\eqref{decomp-LQG} and~\eqref{schem-proj}, we conclude that all the simple
modules over $\LQGodd$ that occur as subquotients in the spin-chain
$\Hilb_{2L}$ are the one-dimensional modules $\Xodd{n}$ parametrized by the weight $n$,
 where $n$ is an integer number in the interval
$-L\leq n\leq L$.

The only difference in the
representation theory of $\centJTL$ when compared to $\LQGodd$ is due to the two additional generators
$\f^L$ and $\e^L$ which map the two $\JTL{N}$-invariants 
 (at $S^z=\pm L$) of the spin-chain $\Hilb_{2L}$ onto each
other.
Simple modules over $\centJTL$ are the same $\Xodd{n}$ for $-L+1\leq
n\leq L-1$ (since $\e^L$ and $\f^L$ act then trivially) and we use the same
notation for them. The $\LQGodd$-modules $\Xodd{\pm L}$ are combined
by the action of $\e^L$ and $\f^L$ into a two-dimensional simple
module over $\centJTL$ which we also denote as $\Xodd{L}$ (to avoiding 
confusion we explicitly indicate the corresponding algebra in our
decompositions).  In what follows, we will contend ourselves by studying  modules over $\LQGodd$. Modules over $\centJTL$ are easily recovered, and the distinction is not relevant for our purposes.

\begin{rem}\label{rem:others}
There are also simple $\LQGodd$-modules of dimension $2r$ with
the action given by the restriction on  the $2r$-dimensional
$\LQG$-modules $\XX_{2,r}$, with $r\geq 1$, described
in~\cite{BFGT}. We do not give details because these modules do not
appear in our spin-chains.
\end{rem}

\subsection{Spin-chain decomposition over $\LQGodd$}\label{subsec:PPodd-def}
We now introduce  indecomposable $\LQGodd$-modules
$\PPodd{n}$ which are used then in the decomposition of
$\Hilb_N$. With the use of the algebra homomorphism~\eqref{LQGodd-LQG-hom}, we define the modules $\PPodd{n}$ as the restriction of the projective $\LQG$-modules
$\PP_{1,n}$ described above in Sec.~\ref{subsec:LQG-decomp-def}. 
Using the
homomorphism~\eqref{LQGodd-LQG-hom} together with the action in $\PP_{1,n}$
from App.~\Approjmodbase,
one easily shows  that all $\PPodd{n}$, with $1\leq n\leq L$, are indecomposable
$\LQGodd$-modules with dimension  $4n$.

As an example, for the restriction of the projective module
$\PP_{1,2}$ covering the doublet representation, we have the following diagrams
of subquotient structure
\begin{equation*}
\includegraphics[scale=0.9]{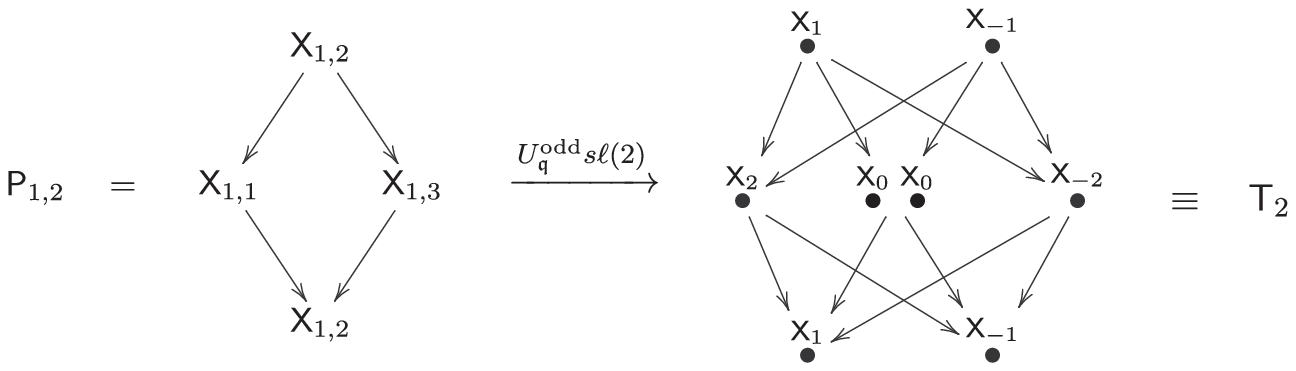}\\
\end{equation*}
where the horizontal arrow means the restriction to the subalgebra
$\LQGodd$, and the diagram on the right depicts the subquotient
structure for $\PPodd{2}$. The two-dimensional top subquotient $\XX_{1,2}$ in
$\PP_{1,2}$ is split into two one-dimensional top subquotients $\Xodd{\pm1}$ in
$\PPodd{2}$, and the arrows are split in such a way that short
south-\textit{west} arrows, say mapping from $\Xodd{1}$ to $\Xodd{2}$, and south-\textit{east}
ones denote the action of $\E\equiv\EE{0}$ and $\F\equiv\FF{0}$, respectively, while
the long south-west, say mapping from  $\Xodd{-1}$ to $\Xodd{2}$, and south-east
arrows denote the action of $\EE{1}$ and $\FF{1}$, respectively. Due
to~\eqref{Uqodd-dfn-3} and the fermionic relations
$\EE{0}^2=\FF{0}^{2}=0$, it follows that a node in the middle of
the diagram, say the left $\Xodd{0}$, has ingoing arrows of either
south-west or south-east direction and  outgoing arrows of opposite direction.

We next study the decomposition of the representation $\repQG$ of $\LQGodd$ in $\Hilb_N$. To help the reader, we begin with an
example for $N=8$ (or $L=4$).
The decomposition is given in Fig.~\ref{Uq_odd_T-N8},
\begin{figure}\centering
    \includegraphics[scale=1]{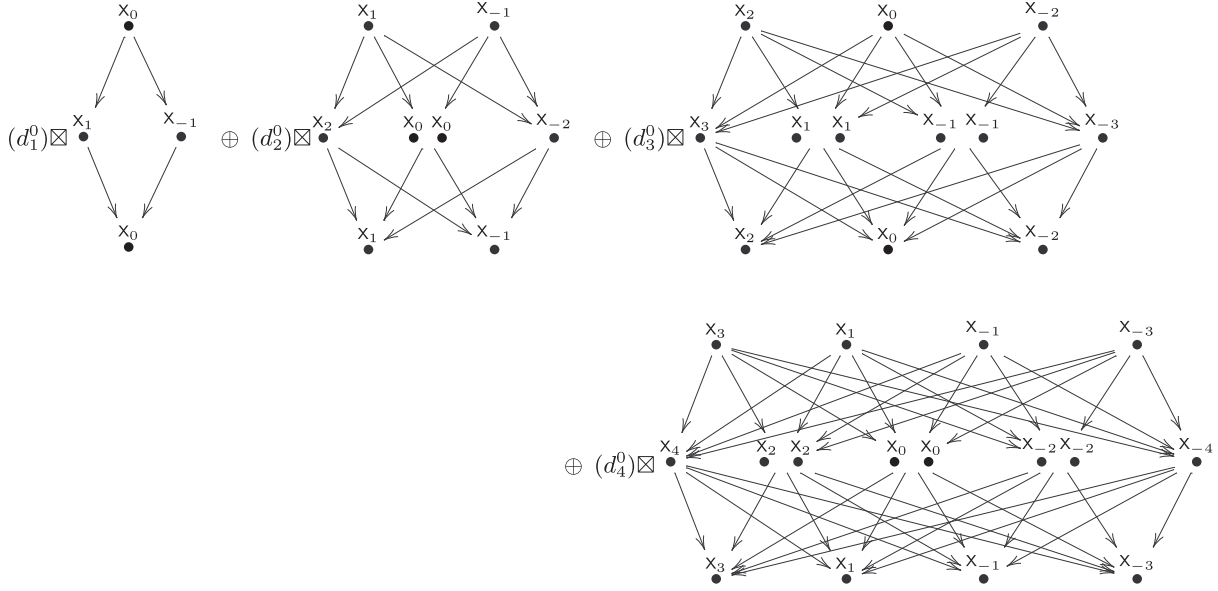}
  \caption{The decomposition of the spin-chain ($N=8$) over
$\LQGodd$ into four indecomposable modules $\PPodd{n}$ with the
    multiplicities $(d^0_n)$, $1\leq n\leq 4$.  Each node in the
 middle level of $\PPodd{n}$ has ingoing arrows only of one type
 (either south-west or south-east) and outgoing ones  of the opposite type.}
    \label{Uq_odd_T-N8}
    \end{figure}
where the multiplicities $d^0_n$ are given by the same expression as
the one after~\eqref{decomp-LQG} for the open case, because the
restriction in each $\PP_{1,n}$ is an indecomposable module over
$\LQGodd$ as we noted before. For $N=8$ one finds $(d^0_1)=(14)'$,
$(d^0_2)=(14)$, $(d^0_3)=(6)$, and $(d^0_4)=(1)$. These numbers   must  be
dimensions of simple modules over $\JTL{N}$, a fact we will discuss more below
 (these dimensions turn out to coincide with those of the simples
in the open case, a peculiarity of this value of $\q=i$).  We note that
a south-west arrow mapping from a subquotient $\Xodd{m}$ to
$\Xodd{n}$, {\it i.e.}, $n>m$, represents an action of the raising generator
$\EE{(n-m-1)/2}$ while a south-east arrow mapping from a subquotient
$\Xodd{m}$ to $\Xodd{n}$, {\it i.e.}, $n<m$, represents an action of the
lowering generator $\FF{(m-n-1)/2}$. We also note that all
subquotients of $\PPodd{n}$ in the middle level (those $\Xodd{k}$ that satisfy $k-n=0\mod2$) are divided into two classes -- one
having only south-west ingoing  and south-east outgoing arrows,
and the other  has having only south-east ingoing
and south-west outgoing arrows. In Fig.~\ref{Uq_odd_T-N8}, we thus have for  $\PPodd{4}$ that the left-most subquotient $\Xodd{2}$ which is in the image of $\FF 0$ is mapped by $\EE 0$ to the $\Xodd{3}$ in the bottom,   while all generators $\FF{n}$ represented by south-east arrows act as zero on it, and, in contrast, the right-most node $\Xodd{2}$ is sent to zero by $\EE 0$ while it is mapped to three subquotients corresponding to the targets of the three south-east arrows.

\medskip

\begin{figure}\centering
   \begin{equation*}
    {\footnotesize
 \xymatrix@C=3pt@R=70pt@M=2pt@W=2pt{
 	&\stackrel{\Xodd{n-1}}{{\color{col-p-3}\bullet}}
 		\ar@[col-p-1][dl]  \ar@[col-m-2]@{->}[drrr]  \ar@[col-m-4]@{->}[drrrrrrr]\ar@[col-m-5]@{->}[drrrrrrrrrrr]
           && 
           &\dots&\stackrel{\Xodd{n-2k'-1}}{{\color{col-p-k} \bullet}}
 	   \ar@[col-p-1][dl] \ar@[col-p-3][dlllll] \ar@[col-m-2][drrr] \ar@[col-m-4][drrrrrrr]
 	   &\dots&\stackrel{\Xodd{2l'-n+1}}{{\color{col-m-l} \bullet}}
 	    \ar@[col-p-2][dlll] \ar@[col-p-4][dlllllll] \ar@[col-m-1][dr] \ar@[col-m-3][drrrrr] 
 	    &\dots&
 	    &&\stackrel{\Xodd{-n+1}}{{\color{col-m-2} \bullet}}
 		 \ar@[col-p-2][dlll] \ar@[col-p-4][dlllllll]\ar@[col-p-5]@{->}[dlllllllllll] \ar@[col-m-1][dr]&\\
 	\stackrel{\Xodd{n}}{{\color{col-p-4}\bullet}}
 		\ar@[col-m-1][dr]\ar@[col-m-3][drrrrr] \ar@[col-m-4]@{->}[drrrrrrr]\ar@[col-m-5]@{->}[drrrrrrrrrrr]
 	  &&\dots
 		&&\stackrel{\Xodd{n-2k}}{{\color{col-p-k}\bullet}}\stackrel{\Xodd{n-2k}}{{\color{col-p-k}\bullet}} 
 		    \ar@[col-p-2][dlll] \ar@[col-m-1][dr] \ar@[col-m-2][drrr]\ar@[col-m-4][drrrrrrr]
 	    &&\dots
	     		&&{\stackrel{\Xodd{2l-n}}{{\color{col-p-k}\bullet}}\stackrel{\Xodd{2l-n}}{{\color{col-p-k}\bullet}}} 
			\ar[dl]\ar[dlll] \ar[dlllllll]\ar[drrr]
 	     &&\dots
 	      &&\stackrel{\Xodd{-n}}{{\color{col-m-4}\bullet}}\ar@[col-p-1][dl]\ar@[col-p-3][dlllll] \ar@[col-p-4]@{->}[dlllllll]\ar@[col-p-5]@{->}[dlllllllllll]\\
       &\stackrel{\Xodd{n-1}}{{\color{col-p-3} \bullet}}
        &&
        &\dots&\stackrel{\Xodd{n-2k'-1}}{{\color{col-p-k} \bullet}}
        &\quad\dots\quad&\stackrel{\Xodd{2l'-n+1}}{{\color{col-m-l} \bullet}}
 	&\dots&
 	 &&\stackrel{\Xodd{-n+1}}{{\color{col-m-3} \bullet}}&
    }  
 }
  \end{equation*}
  \caption{Subquotient structure of the $\LQGodd$-modules $\PPodd{n}$,
  where $n\geq1$, $1\leq k,l\leq n-1$,  $1\leq k',l'\leq n-2$. Each simple subquotient $\Xodd{k}$ appears once in the top and bottom parts of the diagram, and each $\Xodd{k}$, with $-n<k<n$ and $k-n=0\modd2$, appears twice in the middle.
 A south-west arrow from $\Xodd{m}$ to
  $\Xodd{n}$, {\it i.e.}, when $n>m$, represents the generator $\EE{(n-m-1)/2}$ while a
  south-east arrow with $n<m$ corresponds to the action of $\FF{(m-n-1)/2}$.}
    \label{Uq_odd_Tn}
    \end{figure}

In general, restricting the open chain decomposition~\eqref{decomp-LQG} on $\LQGodd$ and because the
restriction in each $\PP_{1,n}$ is an indecomposable module over
$\LQGodd$ as we noted before, we thus have the following decomposition over $\LQGodd$ 
\begin{equation}\label{eq:Hilb-decomp-Uqodd}
\Hilb_N|_{\rule{0pt}{8.5pt}%
\LQGodd} = \bigoplus_{n=1}^{L}\IrrTL{n}\boxtimes\PPodd{n}
\end{equation}
with the subquotient structure for $\PPodd{n}$, with $n\geq2$, given in Fig.~\ref{Uq_odd_Tn}.
We note that each node in the
 middle level of each $\PPodd{n}$ has ingoing arrows only of one type
 (either south-west or south-east) and outgoing ones  of the opposite type. This
 trivially follows from the relation~\eqref{Uqodd-dfn-3} and the restriction on the subalgebra $\LQGodd$
 using formulas in App.~\Approjmodbase. With the use of the
 homomorphism~\eqref{LQGodd-LQG-hom}, the formulas give an explicit
 action of $\EE{n}$ and $\FF{m}$, with $n,m\geq0$, in the basis used
 in App.~\Approjmodbase. We only note again that a south-west arrow mapping from a subquotient $\Xodd{m}$ to
$\Xodd{n}$, {\it i.e.}, when $n>m$, represents an action of the raising generator
$\EE{(n-m-1)/2}$ while a south-east arrow with $n<m$ corresponds to $\FF{(m-n-1)/2}$.

The space $\Hilb_{2L}$ being considered as a module over the centralizer $\centJTL$ has the same decomposition~\eqref{eq:Hilb-decomp-Uqodd} with the only difference in the subquotient structure for $\PPodd{\pm L}$. The two nodes $\Xodd{\pm L}$ in Fig.~\ref{Uq_odd_T-N8}
(for $L=4$) and Fig.~\ref{Uq_odd_Tn} are mixed by the action of $\f^L$ and $\e^L$ into one simple subquotient over $\centJTL$.

Finally, we note that the full dimension of the $\gl(1|1)$ spin chain is recovered via
\begin{equation}
\sum_{n=1}^L 4n ~d_n^0=4\times 2^{2L-2}=2^{2L}
\end{equation}
in agreement with the dimension of the $\PPodd{n}$ being equal to $4n$. 
The same formula would represent the dimension of the
$\gl(1|1)$ Hilbert space in the open case, $4n$ now being the dimension
of projective modules of the centralizer given by the full
quantum group $\LQG$.
These
are replaced here
by modules over $\LQGodd$.

In the rest of this section and in the next section, we describe our rather technical results (homomorphisms between direct summands $\PPodd{n}$ and the structure of JTL standard modules) which are used then in an analysis of the spin-chain decomposition over the JTL algebra.
The reader can skip the rest of this section and the next section in  first reading and go over directly to Sec.~\ref{sec:sp-ch-decomp-JTL}
where 
the decomposition over $\JTL{N}$ is described. 

\subsection{Spaces of intertwining operators}\label{sec:Tn-homo}
We  now describe  all intertwining operators respecting the $\LQGodd$ action on the
spin-chain by studying homomorphisms among the indecomposable
direct summands $\PPodd{n}$ in the
 decomposition~\eqref{eq:Hilb-decomp-Uqodd} for each even $N$. We
 begin with  basic information about first extension groups for a
 pair of simple modules. Then, we introduce Weyl-type modules that
 allow us to describe images and kernels of all the homomorphisms between $\PPodd{n}$.

\subsubsection{Extensions for $\LQGodd$}

We  study possible extensions between simple
$\LQGodd$-modules in order to construct indecomposable modules in what
follows, and  begin our description of the extensions by introducing some
standard notations and definitions. 

Let $A$ and $C$ be left
$\LQGodd$-modules. We call a
short exact sequence 
 $0\to A\to B\to C\to 0$
 an \textit{extension} of $C$ by $A$, and we let $\Ext(C,A)$ denote
the set of equivalence classes (see, {\it e.g.},~\cite{[M]}) of such extensions. Qualitatively, the extension group $\mathrm{Ext}^1(C,A)$ is the 
vector space of possible  glueings between modules $A$ and $C$ 
into an indecomposable module $B$ containing a submodule isomorphic to $A$ and
having at the top the subquotient $C$. 

First extensions  can be  analyzed in principle by simply using  the defining relations from Def.~\ref{dfn:Uqodd}.
While this task is difficult in general, our problem is rather easy because the modules $\Xodd{n}$ are one dimensional. 
 Therefore, when  asking  how the action of  $\LQGodd$  on $\Xodd{n}$ can be modified by glueing $\Xodd{m}$ to it, the relations~\eqref{Uqodd-dfn-4} in particular  (recall that $\h$ has the eigenvalue $n/2$ on $\Xodd{n}$) show that there is only possible one generator that can do it -- it is $\EE{(m-n-1)/2}$ if $m>n$, and $\FF{(n-m-1)/2}$ if $n>m$. 
 
The following result then easily follows.
\begin{prop}\label{prop:exts}
For $-L\leq n,m\leq L$,
  there are vector-space isomorphisms
  \begin{equation*}
    \Ext(\Xodd{n},\Xodd{m})\cong
    \begin{cases}
      \oC,\quad n+m = 1\mod2,\\
      0,\quad \text{otherwise}.
    \end{cases}
  \end{equation*}
All other 
first extensions between simple modules in the category
of finite-dimensional $\LQGodd$-modules vanish (see Rem.~\ref{rem:others} and results on extension groups in~\cite{BFGT}.)
\end{prop}

Using the relations, see also~\cite{GRS1},
\begin{equation*}
[\FF{0},\e^L] = L\K^{-1}\EE{L-1}, \qquad [\EE{0},\f^L] = L\K^{-1}\FF{L-1}
\end{equation*}
in the centralizer $\centJTL$ and similar ones for $\EE{n}$ and $\FF{n}$,
we obtain the same result as in Prop.~\ref{prop:exts} on the first extension groups for simple modules over  $\centJTL$. 
The only difference is in the range $-L+1\leq n,m\leq L$, and the module $\Xodd{L}$ is two-dimensional, see the comment above  Rem.~\ref{rem:others}. 

\subsubsection{Indecomposable and Weyl modules}\label{sec:ind-Weyl}

Let now $\Noddp$ denote the positive subalgebra in $\LQGodd$ generated by $\EE n$, for $n\geq0$, and $\Noddm$ denote the negative  subalgebra generated  by $\FF n$, with $n\geq0$;
let also $\Boddp$ denote the positive Borel subalgebra generated by $\h$, $\K$ and $\EE n$, for $n\geq0$, and $\Boddm$ denote the negative Borel subalgebra generated  by $\h$, $\K$ and $\FF n$, with $n\geq0$. 

 Using information about first-extension groups in
 Prop.~\ref{prop:exts}, we construct two series of indecomposable
 $\LQGodd$-modules as extensions of two semi-simple modules. The
 first series consists of modules denoted by $\Bmod{m}{n}$, where $n-m=0\mod2$ and
 $n\geq0$ and $m\geq2$, with trivial action of the positive subalgebra
 $\Noddp$ and  the subquotient structure
   \begin{equation}\label{Bmn-mod}
   \includegraphics[scale=0.9]{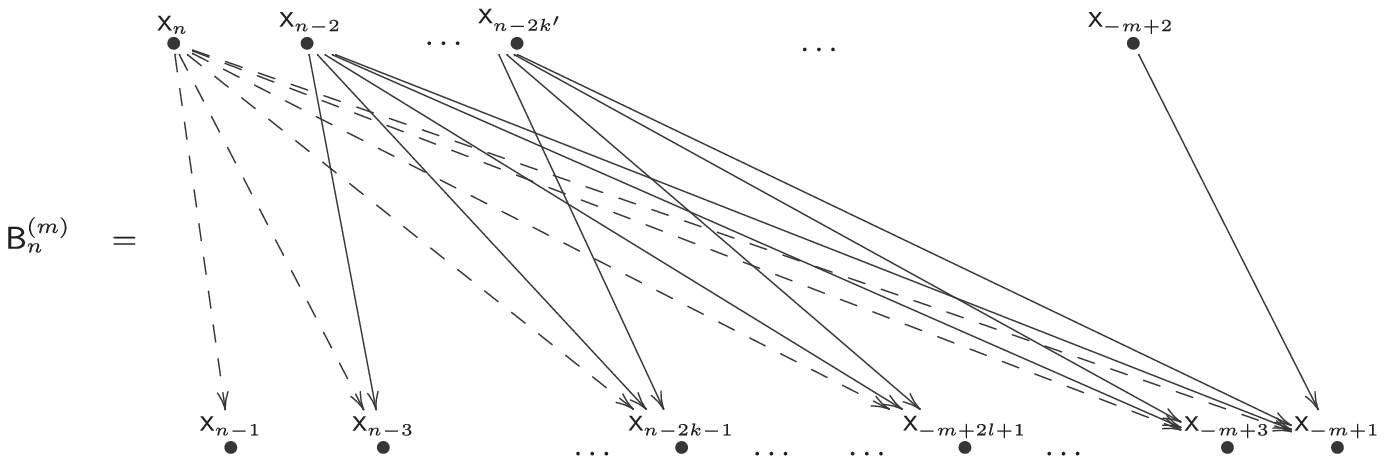}
  \end{equation}
  The
second series consists of modules $\Rmod{m}{n}$ with  trivial action of the
 negative subalgebra $\Noddm$ and the subquotient structure
   \begin{equation}\label{Rmn-mod}
   \includegraphics[scale=0.9]{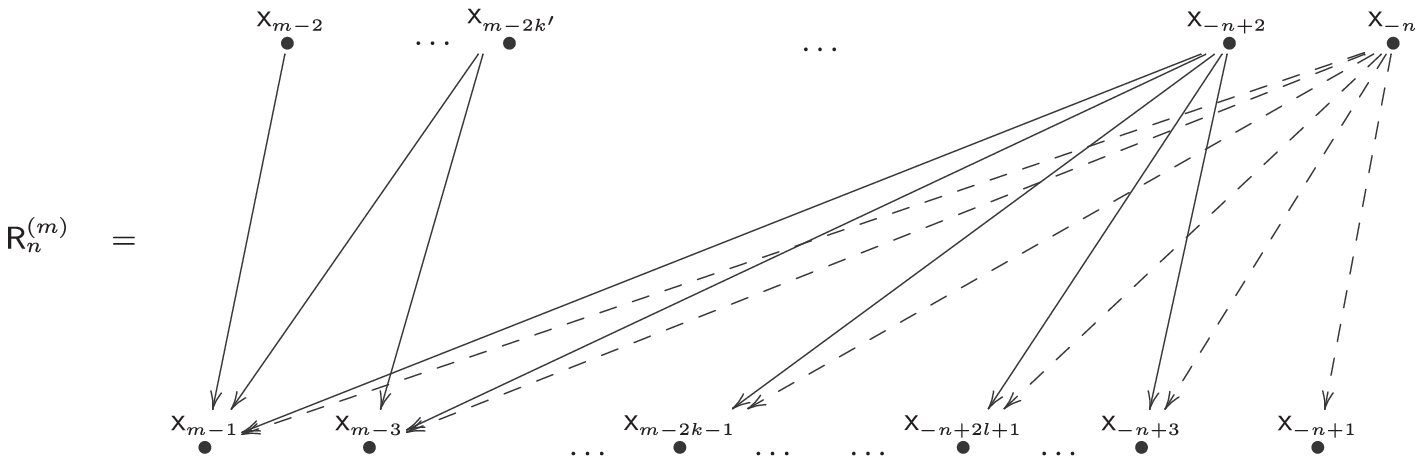}
  \end{equation}
Here,  representatives from $\Ext(\Xodd{l},\Xodd{l'})$ are 
depicted by a south-west arrow if $l'>l$ and a south-east arrow if
$l'<l$ (dash lines are used just for clarity). 
We note that the source and the target
of an arrow uniquely define the generator represented by the arrow. The generators $\EE k$ of the positive
subalgebra are represented in~\eqref{Rmn-mod} by south-west arrows mapping from a node
$\Xodd{l}$ to $\Xodd{l'}$ whenever $(l'-l-1)/2=k$,   and the action of
$\FF k$ from the negative subalgebra is given in~\eqref{Bmn-mod} by south-east arrows mapping from a node
$\Xodd{l}$ to $\Xodd{l'}$ whenever $(l'-l+1)/2=-k$.

The modules ($\Bmod{n+2}{n}$) $\Bmod{n}{n}$ and ($\Rmod{n+2}{n}$) $\Rmod{n}{n}$ play the role
of the (contragredient) Weyl modules over $\LQG$. Recall that Weyl modules over $\LQG$ are obtained as limits of irreducible modules from generic value  of $\q$ to a root of unity value~\cite{ChPr}.  In more details, the
modules $\Bmod{n}{n}\oplus\XX_{-n}$ and $\Rmod{n}{n}\oplus\XX_{n}$ are
restrictions of the Weyl module of dimension $2n+1$ to the
negative and positive Borel subalgebras $\Boddm$ and $\Boddp$ of
$\LQGodd$, respectively. It is straightforward to check with the use
of the defining relations~\eqref{Uqodd-dfn-1}-\eqref{Uqodd-dfn-4} that these
restrictions are $\LQGodd$-modules as well. Similarly, $\Bmod{n+2}{n}\oplus\XX_{n+1}$ and
$\Rmod{n+2}{n}\oplus\XX_{-n-1}$ are  restrictions of the
contragredient Weyl $\LQG$-module of dimension $2n+3$. We show below
 that these ``Weyl" $\LQGodd$-modules are building
blocks of the spin-chain -- indecomposable direct summands are
glueings of a pair of these modules -- like the Weyl modules over $\LQG$ do in the open case.

Using~\eqref{Bmn-mod} and~\eqref{Rmn-mod}, we find the following filtrations of the  $\Bmod{m}{n}$ and $\Rmod{m}{n}$ modules.
\begin{eqnarray}
0=\Bmod{m}{-m+4}\subset\Bmod{m}{-m+2}\subset\dots\subset\Bmod{m}{n-4}\subset\Bmod{m}{n-2}\subset\Bmod{m}{n},\label{Bmod-filtr}\\
0=\Rmod{m}{-m+4}\subset\Rmod{m}{-m+2}\subset\dots\subset\Rmod{m}{n-4}\subset\Rmod{m}{n-2}\subset\Rmod{m}{n},\label{Rmod-filtr}
\end{eqnarray}
where for each pair of neighbour terms  $\Bmod{m}{k}/\Bmod{m}{k-2}$ is isomorphic to an indecomposable module with the subquotient structure
   $\Xodd{k}\rightarrow\Xodd{k-1}$, and $\Rmod{m}{k}/\Rmod{m}{k-2}$ is isomorphic to $\Xodd{-k}\rightarrow\Xodd{-k+1}$.
   
An obvious and important property of the filtrations~\eqref{Bmod-filtr} and~\eqref{Rmod-filtr} is the constant value of the superscript index $(m)$ in their terms. The $\Bmod{m}{n}$ and $\Rmod{m}{n}$ modules do not contain submodules isomorphic to $\Bmod{m'}{n'}$ and $\Rmod{m'}{n'}$, respectively, for any $n'$ and $m'<m$, as well as any of submodules in $\Bmod{m'}{n'}$ and $\Rmod{m'}{n'}$. We call this property of the $\Bmod{m}{n}$ and $\Rmod{m}{n}$ modules \textit{the constant-$m$} property, which will be used below.

\medskip
 
 As was observed above in Sec.~\ref{subsec:PPodd-def}, 
all subquotients of $\PPodd{n}$ in the middle level (those $\Xodd{k}$ satisfying $k-n=0\mod2$) are divided into two classes -- one
having only south-west ingoing and south-east outgoing arrows,
and the other having only south-east ingoing
and south-west outgoing arrows. Following this division, we therefore can construct a $\PPodd{n}$ module as an extension of the
modules 
introduced in~\eqref{Bmn-mod} 
and~\eqref{Rmn-mod} in the following two ways
 \begin{equation}\label{Podd-Weyl}
     \xymatrix@C=2pt@R=13pt@M=1pt@W=2pt{&&\\&\PPodd{n}\quad\cong\quad&\\&&}
\xymatrix@C=4pt@R=22pt{
           &\Bmod{n+1}{n-1} \ar@[col-p-3][dl] &\\
 	  \Bmod{n}{n} 
 	    &&      } 
  \xymatrix@C=2pt@R=13pt@M=1pt@W=2pt{&&\\&\quad\cong\quad&\\&&}
\xymatrix@C=4pt@R=22pt{
           &\Rmod{n+1}{n-1} \ar@[col-m-3][dr] &\\ 
 	    &&  \Rmod{n}{n}    } 
\end{equation}
where the south-west and south-east arrows depict the action of the positive and negative subalgebras $\Noddp$ and $\Noddm$, respectively.
This construction of direct summands in the spin-chain decomposition is similar to what happens in  the open case, where a $\PP_{1,n}$ module  is an extension of a pair of Weyl modules over $\LQG$.
 
 \medskip
 
In order to study the decomposition over $\JTL{N}$  we now describe
all intertwining operators respecting the $\LQGodd$ action on the
spin-chain. Using the decomposition~\eqref{eq:Hilb-decomp-Uqodd}, it
is enough to describe all homomorphisms among the indecomposable
direct summands $\PPodd{n}$.
\begin{thm}\label{lem:endo-PPodd} For $n,m\in\oN$,
  we have the equalities
\begin{equation}
\dim \Hom_{\LQGodd}(\PPodd{n},\PPodd{m})=
    \begin{cases}
      2,&\quad m = n\pm1,\\
      \min(n,m)+\delta_{n,m},&\quad m - n=0 \mod 2,\\
      0,&\quad\text{otherwise}.
    \end{cases}
\end{equation}
The  two-dimensional space in the case $m =n+1$ is spanned
by homomorphisms $f^{\pm}_{n,n+1}$ with  images
\begin{equation*}
\mathrm{im}(f^+_{n,n+1}) \cong \Rmod{n+1}{n-1},\quad 
\mathrm{im}(f^-_{n,n+1}) \cong \Bmod{n+1}{n-1},
\end{equation*}
while  the case $m =n-1$ corresponds to maps $f^{\pm}_{n,n-1}\in\Hom(\PPodd{n},\PPodd{n-1})$ with  images
\begin{equation*}
\mathrm{im}(f^+_{n,n-1}) \cong \Rmod{n-1}{n-1},\quad 
\mathrm{im}(f^-_{n,n-1}) \cong \Bmod{n-1}{n-1}.
\end{equation*}
In the case $m - n = 0\mod 2$, the $\Hom$-space is spanned by homomorphisms with
semisimple images.
\end{thm}
\begin{proof}
We first describe the space $\HomQodd(\PPodd{n},\PPodd{m})$ when
$n-m=0 \mod 2$. The subquotient structure of $\PPodd{n}$ in
Fig.~\ref{Uq_odd_Tn} makes evident that the only non-trivial
intertwining operators from $\PPodd{n}$ to $\PPodd{m}$, with
$n-m=0\mod2$ and $n\ne m$, are homomorphisms with images isomorphic to
semi-simple submodules in $\PPodd{m}$. The corresponding $\Hom$ space
is spanned by homomorphisms with images isomorphic to $\Xodd{k}$, with
$k-n=0\mod2$ and $1-\min(n,m)\leq k\leq \min(n,m)-1$. In the case
$n=m$ we have one more homomorphism given by identity.

Second, it is crucial to note that, for $n\ne m$,
non-trivial
homomorphisms with images being an indecomposable but reducible submodule are only
between $\PPodd{n}$ and $\PPodd{n\pm1}$. Indeed, all homomorphisms between 
$\PPodd{n}$ and $\PPodd{n\pm(2k+1)}$, for $k>0$, are trivial. To show this, assume that 
there exists a non-trivial homomorphism  $\PPodd{n}\to\PPodd{n\pm(2k+1)}$. Then, at least one of 
the top subquotients $\Xodd{k}$, with $-n+1\leq k\leq n-1$, should cover one of the subquotients $\Xodd{k}$ in the middle
level of $\PPodd{n\pm(2k+1)}$ but the latter subquotient has an outgoing arrow to $\Xodd{n+2k}$ or $\Xodd{-n-2k}$, see Fig.~\ref{Uq_odd_Tn} and Fig.~\ref{Uq_odd_T-N8} in particular. Meanwhile, the two last subquotients are not present in $\PPodd{n}$. Therefore, any homomorphism from $\PPodd{n}$ to $\PPodd{n\pm(2k+1)}$ if $k>0$ is trivial. Similar type of arguments shows that any homomorphism from $\PPodd{n\pm(2k+1)}$ to $\PPodd{n}$ if $k>0$ is trivial as well.

We next describe explicitly homomorphisms between $\PPodd{n}$ and $\PPodd{n\pm1}$.
As  follows from~\eqref{Podd-Weyl}
there are at least  two independent homomorphisms beween $\PPodd{n}$ and
$\PPodd{n\pm1}$ -- of ``positive/south-west" and
``negative/south-east" types. A homomorphism $\PPodd{n}\to\PPodd{n+1}$
of the positive type has its kernel isomorphic to $\Rmod{n}{n}$ and
its image is the submodule $\Rmod{n+1}{n-1}\subset\PPodd{n+1}$, where
we use~\eqref{Podd-Weyl} and the filtration~\eqref{Rmod-filtr};
the negative-type homomorphism $\PPodd{n}\to\PPodd{n+1}$ has its image
isomorphic to $\Bmod{n+1}{n-1}\subset\PPodd{n+1}$, where we
use~\eqref{Bmod-filtr}.  To describe the two homomorphisms
$\PPodd{n}\to\PPodd{n-1}$, we only note that their kernels are
generated by $\Rmod{n}{n}$ and the subquotient $\Xodd{n}$ -- for the
positive-type homomorphisms, -- and by $\Bmod{n}{n}$ together with the
subquotient $\Xodd{-n}$ -- for the negative-type. The images of the
last two homomorphisms are isomorphic to $\Rmod{n-1}{n-1}$ and
$\Bmod{n-1}{n-1}$, respectively.

 Finally, assuming that there exists one more
 homomorphism from $\PPodd{n}$ to $\PPodd{n+1}$ linearly independent
 with the two ones just constructed we should necessarily consider
 one of the top subquotients of $\PPodd{n}$ in the kernel of the
 assumed homomorphism. Then, this implies that an image of the 
 homomorphism should be a submodule in $\Bmod{m}{n-1}$ or $\Rmod{m}{n-1}$, with
 $m<n+1$, and at the same time this image should be a submodule in $\Bmod{n+1}{n+1}$ or $\Rmod{n+1}{n+1}$ from
 $\PPodd{n+1}$, see~\eqref{Podd-Weyl}. This property contradicts  the constant-$m$ property of the 
 $\Bmod{n+1}{n+1}$ and $\Rmod{n+1}{n+1}$ modules introduced after~\eqref{Rmod-filtr}.
 Similarly, one can show that there are only two linearly independent
 homomorphisms from $\PPodd{n}$ to $\PPodd{n-1}$.
 These statements finish the proof.
\end{proof}

\section{The standard modules over $\JTL{N}$}\label{sec:JTL-stand}

\subsection{Generalities}\label{sec:JTL-decomp-gen}
 
 We now go back for a little while to the case of the full affine
 Temperley--Lieb algebra $\ATL{N}(m)$. Set $m=\q+\q^{-1}$. For generic
 $\q$ (not a root of unity), the irreducible representations we shall need are parametrized by two
 numbers. In terms of diagrams, the first is the number of
 through-lines, which we denote by $2j$, $j=0,1,\ldots, L$, connecting
 the inner boundary of the annulus with $2j$ sites and the outer boundary with $2L$ sites; the $2j$ sites on the inner boundary we call free or non-contractible. For example, the diagrams
  \begin{tikzpicture}
 	\draw[thick, dotted] (-0.05,0.1) arc (0:10:0 and -3.5);
 	\draw[thick, dotted] (-0.05,0.1) -- (1.05,0.1);
 	\draw[thick, dotted] (1.05,0.1) arc (0:10:0 and -3.5);
	\draw[thick, dotted] (-0.05,-0.55) -- (1.05,-0.55);
	\draw[thick] (0,-0.2) arc (-90:0:0.25 and 0.25);
	\draw[thick] (0.4,0.05) arc (0:10:0 and -2.6);
	\draw[thick] (0.6,0.05) arc (0:10:0 and -2.6);
	\draw[thick] (1.0,-0.2) arc (-90:0:-0.25 and 0.25);
	\end{tikzpicture}
  \, and \, 
  \begin{tikzpicture}
 	\draw[thick, dotted] (-0.095,0.1) arc (0:10:0 and -3.5);
 	\draw[thick, dotted] (-0.095,0.1) -- (1.025,0.1);
 	\draw[thick, dotted] (1.0,0.1) arc (0:10:0 and -3.5);
	\draw[thick, dotted] (-0.095,-0.55) -- (1.025,-0.55);
  \draw[thick] (0,0.05) arc (0:10:0 and -2.6);
	\draw[thick] (0.2,0.05) arc (-180:0:0.25 and 0.25);
	\draw[thick] (0.9,0.05) arc (0:10:0 and -2.6);
	\end{tikzpicture}
 correspond to $L=2$ and $j=1$, where as usual we identify the left and right sides of the framing rectangles, so the diagrams live on the annulus. The action of the algebra $\ATL{N}(m)$ is defined
 in a natural way on these diagrams, by joining their outer boundary
 to an inner boundary of a diagram from $\ATL{N}(m)$, and removing the
 interior sites. As usual, a closed contractible loop is replaced by
 $m$. Whenever the affine diagram thus obtained has a number of
 through lines less than $2j$, the action is zero. For a given
 non-zero value of $j$, it is possible in this action to cyclically
 permute the free sites: this gives rise to the introduction of a
  pseudomomentum $K$ (not to be confused with the quantum group
 generator). Whenever $2j$ through-lines wind counterclockwise around
 the annulus $l$ times, we unwind them at the price of a factor
 $e^{2ijlK}$; similarly, for clockwise winding, the phase 
is $e^{-i
 2jlK}$~\cite{MartinSaleur,MartinSaleur1}\footnote{A more pedantic
 definition due to~\cite{GL} is the relation
\begin{equation*}
\mu=\mu'\circ u_j^n \equiv e^{iKn}\mu',
\end{equation*} 
where $\mu$ is an affine diagram with $2j$ through lines, $u_j$ is the
translational operator acted on through lines by shifting a free site
by one, and $\mu'$ is so-called standard diagram which has no through
lines winding the annulus.}. This action gives rise to a generically
irreducible module, which we denote by
$\AStTL{j}{e^{2iK}}$. Note that we used a parametrization such that
different pairs $(j,e^{2iK})$ correspond to non-isomorphic modules
over the even-rank subalgebra $\oATL{N}(m)\subset\ATL{N}(m)$
introduced in Sec.~\ref{sec:TL-alg-def}. In the parametrization $(t,z)$ chosen
in~\cite{GL}, this corresponds to $t=2j$ and the twist parameter $z^2=e^{2iK}$.

 The dimensions of these modules $\AStTL{j}{e^{2iK}}$ over $\ATL{2L}(m)$  are then given by 
 \begin{equation}\label{dim-dj}
 \hat{d}_{j}=
 \binom{2L}{L+j},\qquad j>0.
 \end{equation}
Note that the numbers do not depend on $K$ (but representations with
different $e^{iK}$ are not isomorphic).
These generically irreducible modules
$\AStTL{j}{e^{2iK}}$ are known also as
 standard (or cell) $\ATL{N}(m)$-modules~\cite{GL}.
 
 Keeping $\q$ generic, degeneracies in the standard modules appear whenever 
 \begin{eqnarray}\label{deg-st-mod}
 e^{2iK}&=&\q^{2j+2k},\qquad
k\hbox{ is a strictly positive integer.}
 \end{eqnarray}
 The representation $\AStTL{j}{\q^{2j+2k}}$ then becomes reducible, and contains a submodule isomorphic to 
 $\AStTL{j+k}{\q^{2j}}$ that we set to zero whenever $j+k>L$. The quotient is generically irreducible, with
 dimension $\hat{d}_j-\hat{d}_{j+k}$. The degeneracy
~\eqref{deg-st-mod} is well-known \cite{MartinSaleur1,GL}~\footnote{Note that the twist term in~\cite{PasquierSaleur}, which was denoted there
   $q^{2t}$, reads in these notations as $e^{2iK}$. It  corresponds to $z^2$ in
   the Graham--Lehrer work~\cite{GL}, and to the parameter $x$ in the work of  Martin--Saleur \cite{MartinSaleur1}. The case where $k=1$ is special, and related with braid translation of the blob algebra theory. 
We note that in the $\JTL{N}$ case, $2j$ through-lines going around the cylinder pick up a phase $e^{i 2jK}=1$. In \cite{MartinSaleur1}, this corresponds to $\alpha_h=x^h=1$.}.  When $\q$ is a root of unity, there are infinitely many solutions to the equation \eqref{deg-st-mod}, leading to a complex pattern of degeneracies to which we turn below.


The case $j=0$ is a bit special. There is no pseudomomentum, but representations are still characterized by another parameter, related with  the weight given to  non contractible loops. Parametrizing this weight as $z+z^{-1}$, the corresponding standard module of  $\ATL{2L}(m)$ is denoted  $\AStTL{0}{z^2}$ and it has dimension given by~\eqref{dim-dj} for $j=0$.

\bigskip

We now specialize to the Jones--Temperley--Lieb algebra $\JTL{N}(m)$ defined
in Sec.~\ref{sec:TL-alg-def}.
 In this case, the rule that winding through-lines can simply be unwound means that the pseudomomentum must satisfy 
$jK\equiv 0~\hbox{mod}~\pi$ \cite{Jones}.
 All possible values of the parameter $z^2=e^{2iK}$ are thus $j$-th
 roots of unity ($z^{2j}=1$,~\cite{Green}). The kernel of the
 homomorphism $\psi$ in~\eqref{diag-alg} (and the ideal  in $\ATL{N}(m)$ generated by $u^N-1$, in
 particular) acts trivially on these modules if $j>0$. In what
 follows, we will thus  use
 the same notation $\AStTL{j}{z^2}$, with $j>0$, for the standard $\JTL{N}(m)$-modules.
 We note that two standard $\JTL{N}$-modules having only different signs in the $z$ parameter are isomorphic.
 
\medskip

If $j=0$, requiring the weight of the non contractible loops to be $m$ as well
leads to the $\ATL{N}(m)$-module $\AStTL{0}{\q^2}$ which is reducible
even for generic $\q$ -- it contains a submodule isomorphic to
$\AStTL{1}{1}$. Meanwhile, on the standard module $\AStTL{0}{\q^2}$ the kernel of the homomorphism 
$\psi$ is non-trivial: the standard module over $\JTL{N}(m)$ for $j=0$ is obtained
precisely by taking the quotient $\AStTL{0}{\q^2}/\AStTL{1}{1}$ as
in~\cite{GL}. This module is now simple for generic $\q$, has the
dimension $\binom{2L}{L}-\binom{2L}{L-1}$ and is denoted by~$\bAStTL{0}{\q^2}$.

\medskip

In what follows we use the representation theory~\cite{GL} of $\ATL{N}$
 in order to describe the subquotient structure of
$\JTL{N}$-standard modules. For this it is convenient to use a variant   of $\JTL{N}$, which is also 
 embedded in $\ATL{N}(m)$. This variant (dubbed here ``augmented") is the
finite-dimensional algebra  $\aJTL{N}(m)$,  isomorphic to $\JTL{N}(m)$
except for  the ideal without through-lines. In this ideal, the algebra $\aJTL{N}(m)$
differs from  $\JTL{N}(m)$ in that  connections within the points on the inner  or outer annulus,
which are topologically different are treated as different.  Recall that in $\JTL{N}(m)$, diagrams in the ideal with no through lines can be chosen to be planar  (they can be drawn in a box without crossings), and  are in bijection with ordinary  $\TL{N}$-diagrams. 
This distinction leads to the   standard $\aJTL{N}$-module
$\AStTL{0}{\q^2}$ of  dimension $\binom{2L}{L}$.

Several results can easily be established following~\cite{GL}
when $\q=i$, to which we restrict for now. We note that the dimension
of the sector of value $S^z=j$ or $S^z=-j$ (including $j = 0$) in the spin
chain coincides with the dimension $\hat{d}_j$ of the standard module
$\AStTL{j}{e^{2iK}}$ over the augmented algebra $\aJTL{N}$. For $\q=i$, these spin-chain sectors provide highly
reducible representations of the Jones--Temperley--Lieb algebra $\JTL{N}$ closely
related (but non-isomorphic) to the standard modules. By the
discussion of the correspondence~\cite{GRS1} between the XX and the $\gl(1|1)$
spin-chains, we see that these representations occur at pseudomomentum satisfying
$e^{2iK}=(-1)^{j+1}$. Before describing indecomposables appearing in the
spin-chain we first discuss more the standard ones with this value of the pseudomomentum.

\subsection{The standard modules at $\q=i$}\label{sec:stand-q-i}
We first describe modules over the
 algebra $\ATL{N}$, containing the generator $u$.
The structure of the standard $\ATL{N}$-modules at $\q=i$ can be inferred
from~\cite{GL}. 
For a standard module $\AStTL{j}{(-1)^{j+1}}$ with $2j>0$ through lines,
we deduce the subquotient structure using two Graham--Lehrer's theorems, Thm.~3.4 and
 proof of Thm.~5.1 in~\cite{GL}. A crucial fact is that the space of
 homomorphisms 
 \begin{equation}\label{HomATL-iso}
 \HomATL(\AStTL{j}{(-1)^{j+1}},\AStTL{j-1}{(-1)^{j}}) \cong
 \oC,\qquad 1\leq j\leq L.
  \end{equation}
between the standard $\ATL{N}$-modules is one-dimensional and the
 homomorphisms are injective.  The dimensions of simple modules
 $\AIrrTL{j}{(-1)^{j+1}}$ happen to be the same as those in the open
 case, and given by
\begin{equation*}
\widehat{d}_{j,(-1)^{j+1}}^0=d_j^0=\sum_{j'\geq j} (-1)^{j'-j}d_{j'}\qquad
\text{with}\qquad d_j = \binom{2L}{L+j} -  \binom{2L}{L+j+1}.
\end{equation*}
One can show the equivalent formula
\begin{equation*}
\widehat{d}_{j,(-1)^{j+1}}^0=\left(\begin{array}{c} 
2L-2\\
L-j\end{array}\right)-\left(\begin{array}{c} 
2L-2\\
L-j-2\end{array}\right).
\end{equation*}

Our final result for the standard $\ATL{N}$-modules is given on the
left side of Figs.~\ref{GLfig2} and~\ref{GLfig}
 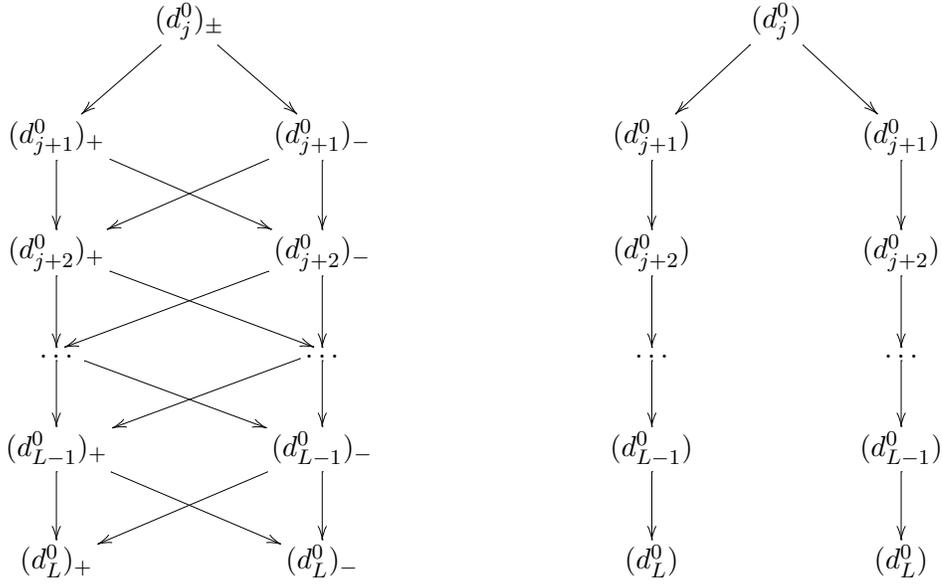
\begin{figure}\centering
 \begin{equation*}
        \xymatrix@R=26pt@C=14pt@W=2pt@M=2pt
   {{}&(\TLX_j)_{\pm}\ar[dr]\ar[dl]&&\\
     {(\TLX_{j+1})_{+}}\ar[d]\ar[drr]
     &&{(\TLX_{j+1})_{-}}\ar[d]\ar[dll]\\
     {(\TLX_{j+2})_{+}}\ar[d]\ar[drr]
     &&{(\TLX_{j+2})_{-}}\ar[d]\ar[dll]\\
     {\quad\dots\quad}\ar[d]\ar[drr]
     &&{\dots}\ar[d]\ar[dll]\\
     {(\TLX_{L-1})_{+}}\ar[d]\ar[drr]
     &&{(\TLX_{L-1})_{-}}\ar[d]\ar[dll]\\
     {(\TLX_{L})_{+}}
     &&{(\TLX_{L})_{-}}
     }
\qquad\qquad\qquad
   \xymatrix@R=26pt@C=18pt@W=2pt@M=2pt
   {{}&(\TLX_{j})\ar[dr]\ar[dl]&&\\
     (\TLX_{j+1})\ar[d]
     &&(\TLX_{j+1})\ar[d]\\
     (\TLX_{j+2})\ar[d]
     &&(\TLX_{j+2})\ar[d]\\
     {\dots}\ar[d]
     &&{\dots}\ar[d]\\
     (\TLX_{L-1})\ar[d]
     &&(\TLX_{L-1})\ar[d]\\
     (\TLX_{L})
     &&(\TLX_{L})
     }
 \end{equation*}
      \caption{The structure of the standard modules
      $\AStTL{j}{(-1)^{j+1}}$ with $2j>0$ through lines at $\q=i$. We
      set $\AIrrTL{j}{(-1)^{j+1}}\equiv(\TLX_{j})$.  The
      module on the left is over $\ATL{N}$ and on the right is the
      restriction to the subalgebra
      $\aJTL{N}$. The twist parameter $z=\pm\sqrt{(-1)^{k+1}}$ for
      each node $(\TLX_k)_{\pm}$ is assumed.}
    \label{GLfig2}
    \end{figure}
where each node corresponds to a simple subquotient. In the case $j=0$, we have no top subquotient because
$\TLX_0=0$.
 We denote the dimension of a simple
subquotient $\AIrrTL{k}{(-1)^{k+1}}$  in the round brackets (with the twist parameter
$z=\pm\sqrt{(-1)^{k+1}}$ for each node $(\TLX_k)_{\pm}$, with $j\leq k\leq L$, to be
assumed).  For simplicity, we use in what follows the round-brackets
notation for simple subquotients. We will also denote the Graham--Lehrer's parameter
$z=\pm\sqrt{z^2}$
 by the subscript $\pm$
distinguishing non-isomorphic simple
$\ATL{N}$-subquotients. 
Restricting to the subalgebra
$\aJTL{N}$,
 subquotients $(\TLX_k)_{\pm}$ are isomorphic and we discard the subscripts.

We now turn to the description of standard modules over the subalgebra $\aJTL{N}$.
\begin{prop}\label{prop:stmod-subqstr}
The subquotient structures for the standard $\JTL{N}$-modules
$\AStTL{j}{(-1)^{j+1}}$, with $j>0$, and for the standard $\aJTL{N}$-module
$\AStTL{0}{-1}$ are given on the right in Figs.~\ref{GLfig2}
  and~\ref{GLfig}, respectively.
\end{prop}
\begin{proof}
The proof consists of two parts \textbf{1.} and \textbf{2}. The first one considers the case
$j=0$ and it is then used in \textbf{2.} to deduce the structure for $j>0$.

\textbf{1.} For the standard $\aJTL{N}$-module $\AStTL{0}{-1}$ without through lines, the
subquotient structure degenerates into a direct sum of two
non-isomorphic indecomposable modules each consisting of affine diagrams
of even or odd rank~\cite{GL}. These two summands are of chain type
and presented on the right diagram of Fig.~\ref{GLfig}.
 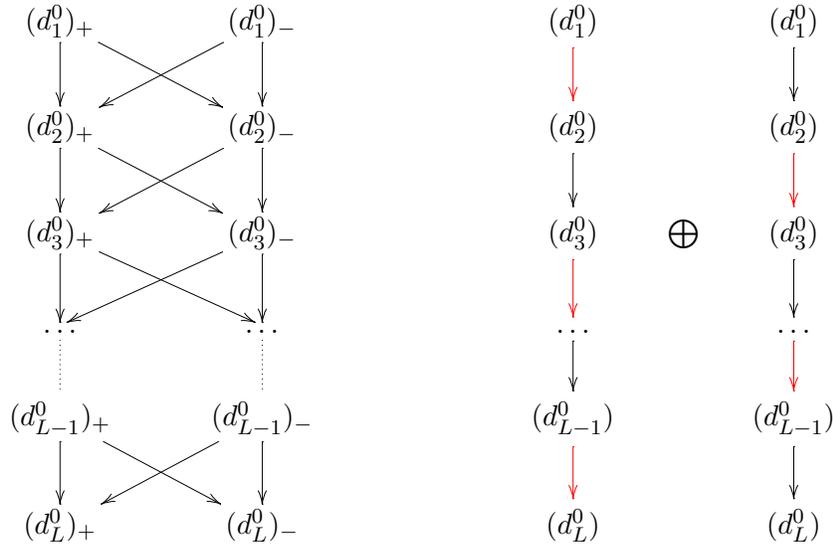
\begin{figure}\centering
 \begin{equation*}
   \xymatrix@R=24pt@C=14pt@W=2pt@M=2pt{
     {(\TLX_{1})_+}\ar[d]\ar[drr]
     &&{(\TLX_{1})_-}\ar[d]\ar[dll]\\
     {(\TLX_{2})_+}\ar[d]\ar[drr]
     &&{(\TLX_{2})_-}\ar[d]\ar[dll]\\
     {(\TLX_{3})_+}\ar[d]\ar[drr]
     &&{(\TLX_{3})_-}\ar[d]\ar[dll]\\
     \dots\ar@{.}[]-<0pt,24pt>&&\dots\ar@{.}[]-<0pt,24pt>\\
     (\TLX_{L-1})_+\ar[d]\ar[drr]
     &&(\TLX_{L-1})_-\ar[d]\ar[dll]\\
      {(\TLX_{L})_+}
     && (\TLX_{L})_-}
    \qquad\qquad\qquad\quad
   \xymatrix@R=20pt@C=12pt@W=4pt@M=4pt
   {{(\TLX_1)}\ar@[red][d]
     &&{(\TLX_1)}\ar[d]\\
     {(\TLX_2)}\ar[d]
     &&{(\TLX_2)}\ar@[red][d]\\
     {(\TLX_3)}\ar@[red][d]
     &\bigoplus&{(\TLX_3)}\ar[d]\\
     \dots\ar[d]
     &&\dots\ar@[red][d]\\
 {(\TLX_{L-1})}\ar@[red][d]
     &&{(\TLX_{L-1})}\ar[d]\\
     {(\TLX_L)}
     &&{(\TLX_L)}
   }
 \end{equation*}
      \caption{The structure of the standard $\ATL{N}$-module $\AStTL{0}{-1}$ at $\q=i$  is on the left side, while the corresponding standard $\aJTL{N}$-module
      $\AStTL{0}{-1}$ (but not over $\JTL{N}$) is given on
      the right side. In the latter case, we show the decomposition on standard
      $\TL{N}$-modules using black and red  arrows. Nodes connected by
      black arrows constitute a standard module over $\TL{N}$ while
      red arrows indicate the  action of the last generator $e_N$ that mixes the
direct summands over $\TL{N}$.}
    \label{GLfig}
    \end{figure}
 Black arrows represent the action of the subalgebra $\TL{N}$
-- open Temperley--Lieb algebra -- generated by $e_j$, with $1\leq j
\leq N-1$, and red arrows indicate  the action of the last generator $e_N$ that mixes the
direct summands over $\TL{N}$. The left direct summand over $\JTL{N}$ is spanned by affine
diagrams $\mu$ of even rank $0\leq|\mu|\leq L$, the right summand --
by odd-rank diagrams. This picture easily follows from the
filtration~\cite{GL} of $\AStTL{0}{-1}$ by the standard
$\TL{N}$-modules.

We note next that the translation operator $u$ ($u e_k u^{-1} = e_{k+1}$) mixes even affine diagrams with
odd ones. The corresponding standard module for $j=0$ with respect to the
bigger algebra $\ATL{N}$ containing the element~$u$ has the
subquotient structure given in Fig.~\ref{GLfig} on the left side.
By selecting a node further down in the ladder, and truncating all
that is at its level or above, one can obtain as well the structure of
all the other standard modules over $\ATL{N}$ presented on the left
side in Fig.~\ref{GLfig2}, using~\eqref{HomATL-iso} and injectivity
of the homomorphisms. We recall  that  the subscript $\pm$
distinguishes non-isomorphic $\ATL{N}$-irreducibles,  and that there are
actually two standard $\ATL{N}$-modules, with the top $(\TLX_j)_{\pm}$, corresponding to the notation $\AStTL{j}{(-1)^{j+1}}$.

\textbf{2.} Restricting to $\aJTL{N}$, the simple modules $(\TLX_{j})_+$
and $(\TLX_{j})_-$ as well as their standard modules are isomorphic as modules over
$\aJTL{N}$ and we thus have the isomorphism of vector spaces
 \begin{equation}\label{HomTL-iso}
 \HomaJTL(\AStTL{j}{(-1)^{j+1}},\AStTL{j-1}{(-1)^{j}}) \cong
 \oC^2,\qquad 1\leq j\leq L.
  \end{equation}
Using this isomorphism, we now show that the ``diagonal" arrows
connecting the left and right strands in $\ATL{N}$-modules are absent in
the corresponding $\aJTL{N}$-modules, {\it i.e.}, they represent actually
 the action of the ideal in $\ATL{N}$ generated by the element $u$. We begin with studying
homomorphisms from $\AStTL{1}{1}$ to $\AStTL{0}{-1}$. We recall that the last
module is a direct sum of two indecomposables each consisting of affine diagrams
of even or odd rank as in Fig.~\ref{GLfig} on the right side, and each having the same top $(\TLX_{1})$ as the $\AStTL{1}{1}$. Therefore, a
basis in~\eqref{HomTL-iso} 
for $j=1$ can be chosen as two homomorphisms with the image isomorphic
to the left 
 or right
 direct summand in $\AStTL{0}{-1}$ in Fig.~\ref{GLfig}. This means
 the kernel of any of these homomorphisms contains either the submodule
 $(\TLX_{2})\to(\TLX_{3})\to\dots$ -- the chain starting with the red
 arrow -- or the one
 starting with the black arrow.
The kernels are submodules over $\aJTL{N}$ and we thus can choose a
 basis in $\AStTL{1}{1}$ such as there are no arrows (with respect to
 the action of $\aJTL{N}$) 
 mixing these submodules. We proceed in the same way for $j>1$.
 This finally gives the diagrams for the modules 
 $\AStTL{j}{(-1)^{j+1}}$ over $\aJTL{N}$ in Fig.~\ref{GLfig2}
 on the right. In these diagrams, we could also indicate the action of $e_N$ by red
 arrows connecting standard $\TL{N}$-modules in a decomposition over
 the subalgebra generated by $e_j$, with $1\leq j\leq N-1$, as in
 Fig.~\ref{GLfig}: the diagrams in such a basis would 
 contain  some ``diagonal" arrows
connecting the left and right strands in  $\AStTL{j}{(-1)^{j+1}}$.

By the definition of $\aJTL{N}$ algebra given above in
Sec.~\ref{sec:JTL-decomp-gen}, the $\JTL{N}$-modules
$\AStTL{j}{(-1)^{j+1}}$ for $j>0$ have the same subquotient structure
as in the right diagram in Fig.~\ref{GLfig2}.
This finishes the proof.
\end{proof}

We finally give some explicit examples.

\begin{example}\label{ex:stand-mod}
For $L=3$ or $N=6$, we have the following diagrams for the subquotient structure of the standard $\ATL{N}$-modules $\AStTL{j}{(-1)^{j+1}}$:
 \begin{equation*}
\xymatrix@R=25pt@C=6pt@W=2pt@M=2pt
{&j=2&&\\&&&\\{}&(4)_{\pm}\ar[dr]\ar@[][dl]&&
\quad\xrightarrow{\;\dim=1\;}\qquad\\
     {(1)_+}
     &&{(1)_-}}
\xymatrix@R=25pt@C=10pt@W=2pt@M=2pt
{&j=1&&\\{}&(5)_{\pm}\ar[dr]\ar@[][dl]&&\quad\xrightarrow{\;\dim=1\;}\quad&\\
     {(4)_+}\ar[d]\ar[drr]
     &&{(4)_-}\ar@[][d]\ar[dll]&\\
     {(1)_+}   &&{(1)_-}&}
   \xymatrix@R=25pt@C=10pt@W=2pt@M=2pt
    {&j=0&&\\{(5)_+}\ar@[][d]\ar[drr] &&{(5)_-}\ar[d]\ar[dll]\\
     {(4)_+}\ar[d]\ar[drr]
     &&{(4)_-}\ar@[][d]\ar[dll]\\
     {(1)_+}
     &&{(1)_-}}
 \end{equation*}
 where we also indicated the injective homomorphisms.
 We also show the dimension of
 the spaces of homomorphisms between the standard $\ATL{N}$-modules in
 the figure. 

The diagrams for  subquotient structure of the modules 
 $\AStTL{j}{(-1)^{j+1}}$ over $\aJTL{N}$ are
 \begin{equation*}
\xymatrix@R=25pt@C=6pt@W=2pt@M=2pt
{&j=2&&\\&&&\\{}&(4)\ar[dr]\ar@[red][dl]&&\quad\xrightarrow{\;\dim=2\;}\qquad\\
     {(1)}
     &&{(1)}}
\xymatrix@R=25pt@C=10pt@W=2pt@M=2pt
{&j=1&&\\{}&(5)\ar[dr]\ar@[red][dl]&&\quad\xrightarrow{\;\dim=2\;}\qquad\\
     {(4)}\ar[d]
     &&{(4)}\ar@[red][d]\\
     {(1)}   &&{(1)}}
   \xymatrix@R=25pt@C=10pt@W=2pt@M=2pt
    {&j=0&&\\{(5)}\ar@[red][d] &&{(5)}\ar[d]\\
     {(4)}\ar[d]
     &\bigoplus&{(4)}\ar@[red][d]\\
     {(1)}
     &&{(1)}}
 \end{equation*}
where we also show the filtration~\cite{GL} of the standatd $\aJTL{N}$-modules by the standard
$\TL{N}$-modules and the red arrows represent  the action of
the  generator $e_6$.
They are the same diagrams as in Fig.~\ref{GLfig2} and
Fig.~\ref{GLfig} but truncated for $L=3$. 
The leftmost  diagram
is for the sector with $2j=4$ through lines ($\hat{d}_2=6$), the
central one is spanned by affine diagrams with $2$ through lines
($\hat{d}_1=15$), and the right most diagram has no through lines
($j=0$, $\hat{d}_0=20$).
The two invariants in $\AStTL{0}{-1}$ are given explicitly by
\begin{equation*}
\inv_1 = \sum_{j=1}^{3}u^{2j}
\Bigl( {\; \begin{tikzpicture}
	\draw[thick] (0,0) arc (-180:0:0.1 and 0.25);
	\draw[thick] (0.4,0) arc (-180:0:0.3 and 0.3);
	\draw[thick] (0.6,0) arc (-180:0:0.1 and 0.18);
	\end{tikzpicture}} - {\begin{tikzpicture}
	\draw[thick] (0,0) arc (-180:0:0.5 and 0.35);
	\draw[thick] (0.2,0) arc (-180:0:0.3 and 0.25);
	\draw[thick] (0.4,0) arc (-180:0:0.1 and 0.15);
	\end{tikzpicture}} \;\Bigr)\qquad
	\inv_2 = u \bigl(\inv_1\bigr),
\end{equation*}
where we use the notation for diagrams on an annulus introduced at  the beginning of Sec.~\ref{sec:JTL-decomp-gen}.
The two $\ATL{N}$-invariants $(1)_{\pm}$ on the diagram above are spanned by $\inv_1\pm\inv_2$, respectively.
\end{example}

\section{The spin-chain decomposition over $\JTL{N}$}\label{sec:sp-ch-decomp-JTL}
It turns out that the structure of the modules present in the
 $\gl(1|1)$ spin chain is closely related to the standard modules
 discussed above. First, we give some results about extensions between (``glueings" of)
 simple modules and  give explicit examples.
Then, we construct ``zig-zag'' indecomposable
 $\JTL{N}$-modules that play the role of the standard modules for $\TL{N}$ in the
 spin-chain decomposition, {\it i.e.}, indecomposable direct summands over
 $\JTL{N}$ in the
 spin-chain are gluings of two such zig-zag modules. Finally, we use
 these modules to describe the 
 subquotient structure of spin-chain modules over $\JTL{N}$ and obtain finally the 
  bimodule structure over the pair $\bigl(\JTL{N},\LQGodd\bigr)$.

\subsection{Extensions between simple $\JTL{N}$-modules}\label{prop:exts-PTL}
We formulate now an important lemma which will be used in what follows. 
\begin{lemma}\label{thm:exts-PTL}
 The dimension of the group of first
extensions between simple $\JTL{N}(0)$-modules $\AIrrTL{n}{(-1)^{n+1}}$
and $\AIrrTL{m}{(-1)^{m+1}}$, for $ n= m\pm1$, is not less than
$2$. 
\end{lemma}
\begin{proof}
Assume that the dimension is less than $2$, {\it i.e.}, $\mathrm{dim} \ExtJTL\bigl(\AIrrTL{n}{(-1)^{n+1}},\AIrrTL{n\pm1}{(-1)^{n}}\bigr) = 1$ (it is obviously not zero). We take then the  standard module $\AStTL{j}{(-1)^{j+1}}$, which is reducible but indecomposable as it was shown in Prop.~\ref{prop:stmod-subqstr}, and consider its quotient by a submodule generated from both subquotients $(\TLX_{j+2})$, or $\AIrrTL{j+2}{(-1)^{j+1}}$, see the right part of Fig.~\ref{GLfig2}. This quotient is still a reducible but indecomposable module with the subquotient structure
\begin{equation*}
   \xymatrix@R=18pt@C=12pt@W=2pt@M=2pt
   {{}&(\TLX_{j})\ar[dr]\ar[dl]&&\\
     (\TLX_{j+1})
     &&(\TLX_{j+1})
     }
\end{equation*}
On the other hand, using our assumption about the first extension groups we can choose a basis in the direct sum 
$(\TLX_{j+1})\oplus(\TLX_{j+1})$ such that the resulting module has a decomposition onto a direct sum $(\TLX_{j+1})\oplus 
(\TLX_{j})\to(\TLX_{j+1})$ of an irreducible module and an indecomposable one. This property contradicts  the fact the the module
is indecomposable. We therefore obtain that dimension of $\ExtJTL\bigl(\AIrrTL{n}{(-1)^{n+1}},\AIrrTL{m}{(-1)^{m+1}}\bigr)$, for 
$m=n+1$, is not less than $2$. Similarly, we can prove the statement for $m=n-1$. We take a conjugate module  $\AStTL{j}{(-1)^{j+1}}^*$, 
which is the space of linear maps  $\AStTL{j}{(-1)^{j+1}}\to\oC$ with the JTL action given by $a f(\cdot) = f(a^{*}\cdot)$, where the anti-involution $\cdot^*$ on the JTL algebra corresponds to reflecting the diagram (for an element $a$) in a horizontal line. The conjugate module has all arrows inverted when compared with  the diagram for the original module. We then consider a submodule with two subquotients $(\TLX_{j+1})$ and one $(\TLX_{j})$ in $\AStTL{j}{(-1)^{j+1}}^*$ and repeat the previous steps using the assumption on the one-dimensionality of the first extensions. This last step finishes our proof.
\end{proof}


In what follows, we use a notation for basis elements denoted by
$\ext{x}_{\pm}$ and $\ext{y}_{\pm}$ that span a two-dimensional subspace in the first extension groups
$\ExtJTL\bigl(\AIrrTL{n}{(-1)^{n+1}},\AIrrTL{n\pm1}{(-1)^{n}}\bigr)$ from Lem.~\ref{thm:exts-PTL}. 
The basis element $\ext{x}_{\pm}$ is chosen to
represent an extension corresponding to the action of the open Temperley--Lieb
subalgebra $\TL{N}$ generated by $e_j$, with $1\leq j\leq N-1$, and it is
depicted by an arrow connecting two simple subquotients 
$\AIrrTL{n}{(-1)^{n+1}}$ and $\AIrrTL{n\pm1}{(-1)^{n}}$. The second extension $\ext{y}_{\pm}$ corresponds to 
the action of the subalgebra
$u\TL{N}u^{-1}\subset{\JTL{N}}$ isomorphic to $\TL{N}$ and  containing the generator $e_N$ and it is depicted by
 a second
arrow connecting the same pair of subquotients as in the diagram
\begin{equation}\label{mod-first-ext}
\xymatrix@R=27pt@C=3pt@W=2pt@M=2pt
{
&&&\AIrrTL{n}{(-1)^{n+1}}\ar[dl]^(.5){\beta \ext{x}_-}\ar@[][dll]_(.5){\alpha \ext{y}_-}
      \ar[dr]_(.5){\gamma \ext{x}_+}\ar@[][drr]^(.5){\delta \ext{y}_+}&&\\
     &\AIrrTL{n-1}{(-1)^{n}}&\AIrrTL{n-1}{(-1)^{n}}   &
     &\AIrrTL{n+1}{(-1)^{n}}&\AIrrTL{n+1}{(-1)^{n}}
     }
\end{equation}
where the coefficients $\alpha,\beta,\gamma,\delta\in\oC$, and we set
$\AIrrTL{0}{-1}\equiv0$. We note that different elements in the
intersection of the two subalgebras $\TL{N}$ and $u\TL{N}u^{-1}$ can
actually map to different linear combinations of the simple
submodules; this is not shown explicitly on the diagram. The existence of two different extensions
of this type was actually announced in the previous section -- we 
refer the reader to our  discussion in the proof of
Prop.~\ref{prop:stmod-subqstr} where arrows of two different types/colors correspond to the action of the two different subalgebras on the right part of Fig.~\ref{GLfig}.  In the proof, we give a decomposition of standard
$\JTL{N}$-modules on standard modules over the subalgebra $\TL{N}$
generated by $e_j$, with $1\leq j\leq N-1$, and show the action of $e_N\in u\TL{N}u^{-1}$ connecting the
direct summands.

Taking all possible 
quotients of the module in~\eqref{mod-first-ext} by a submodule isomorphic to the direct sum
$\AIrrTL{n-1}{(-1)^{n}}\oplus\AIrrTL{n+1}{(-1)^{n}}$, we obtain a family of indecomposable $\JTL{N}$-modules with the subquotient structure
\begin{equation}\label{mod-first-quotient}
\xymatrix@R=30pt@C=30pt@W=2pt@M=2pt
{
\AIrrTL{n-1}{(-1)^{n}}
&{\;\AIrrTL{n}{(-1)^{n+1}}\;}\ar@<-0.5ex>@[][r]_(.5){(\delta:\gamma)\ext{y}_+}\ar@<+0.5ex>[r]^(.5){\ext{x}_+}
\ar@<+0.5ex>@[][l]^(.5){(\alpha:\beta)\ext{y}_-}\ar@<-0.5ex>[l]_(.5){\ext{x}_-}
 &\AIrrTL{n+1}{(-1)^{n}}&
     }
\end{equation}
and parametrized by two points $x_0=\alpha:\beta$ and $y_0=\delta:\gamma$ on a complex
projective line, $x_0,y_0\in\oC\oP^1$. These modules are denoted by $\MmodJTL{1}{n}{x_0,y_0}$
and they will appear below in
spin-chain decompositions. To simplify notations, we will use below
only single arrows with specified parameters on them. Before going to the decomposition, we first
give an example at $N=4$ where parameters on $\oC\oP^1$ appear.

\subsubsection{Example for $N=4$}\label{sec:spin-chain-decomp-four}
The decomposition of the full spin-chain for $N=4$ sites with respect
to the $\JTL{N}$ action is given by the direct sum, where we set for
simple subquotients
$\AIrrTL{1}{1}=(2)$ and $\AIrrTL{2}{-1}=(1)$,
 \begin{equation*}
   \xymatrix@C=10pt@R=18pt{
    &&(1)\ar[d]&
     &&(2)\ar@[][dl]_{1:0}\ar[dr]^{0:1}&&
		 &(1)\ar[d]&&\\
	(1)&
	 {\oplus}&(2)\ar[d]_{1:i}&
	 {\oplus}&(1)\ar@[][dr]_{1:0}
	  & &(1)\ar[dl]^{0:1}&
        {\oplus}&(2)\ar[d]^{1:(-i)}
			  &{\oplus}&(1)\\
	&&(1)&
	 &&(2)&&
      &(1)&&&
   }  
 \end{equation*}
where the left-most direct summand is at $S^z=2$, the second is at
$S^z=1$, {\it etc.}, see also a general decomposition in~\eqref{decomp-JTL}
below.  The only isomorphic modules are the two invariants depicted by
$(1)$ and mapped to each other by $\e^2$ and $\f^2$. The two modules
at $S^z=\pm1$ are non-isomorphic -- they differ by the points on
$\oC\oP^1$ indicated as $(1:\pm i)$ on the lower parts of their
diagrams; in other words, the arrow from $(2)$ to $(1)$ in the sector
$S^z=1$, on the left side, corresponds to the extension $\ext{x}_+ +
i\ext{y}_+$ while the submodule $(2)\to(1)$ at $S^z=-1$ corresponds to
the extension $\ext{x}_+ - i\ext{y}_+$. The basis extensions
$\ext{x}_+$ and $\ext{y}_+$ are introduced
before~\eqref{prop:exts-PTL} and here they simply mean that $e_1$ maps
from the two-dimensional subquotient $(2)$ to the one-dimensional $(1)$ with the coefficient $1$ in an
appropriate basis in the submodule $(1)$
while $e_3$ maps with the coefficient $\pm i$, in the same basis of course.

\subsection{Spin-chain modules over $\JTL{N}$}\label{sec:TL-decomp}
We now recall the decomposition of $\Hilb_{N}$ over the  $\TL{N}$, the open case~\cite{ReadSaleur07-2}.
\begin{equation}\label{decomp-TL}
\Hilb_{N}|_{\rule{0pt}{7.5pt}%
\TL{N}} =  \bigoplus_{j=1}^{L} \PrTL{j}\boxtimes \XX_{1,j} \oplus \StTL{L}\boxtimes \XX_{1,L+1}, 
\end{equation}
 with ``multiplicities'' $\XX_{1,j}$ in front of indecomposable direct
 summands $\PrTL{j}$ being simple $j$-dimensional modules over
 $\LQG$ defined in~\eqref{eq:sl-irrep}. We use the notations $\PrTL{j}$ and $\StTL{j}$ for
 projective and standard $\TL{N}$-modules, respectively. The standard
 module $\StTL{L}$ is the trivial representation $(1)$; the standard
 modules with $1 \leq j < L$ have structure of simple subquotients as $\StTL{j}:
 \IrrTL{j}\rightarrow\IrrTL{j+1}$, and $\StTL{0}$ is the simple module
 $\IrrTL{1}$. The projectives
 $\PrTL{j}$ are self-conjugate and
 described by the diagram $\StTL{j}\rightarrow\StTL{j-1}$.

In general in the periodic case, the $\JTL{N}$ action commutes with
$S^z=2\h$ and
we have thus a decomposition of the
full spin-chain over $\JTL{N}$ on $N=2L$ sites as
\begin{equation}\label{decomp-JTL}
\Hilb_N|_{\rule{0pt}{7.5pt}%
\JTL{N}} =  \bigoplus_{j=-L+1}^{L-1} \APrTL{j}\boxtimes\Xodd{j} \oplus
\AStTL{L}{(-1)^{L-1}}\boxtimes\Xodd{L}, 
\end{equation}
where $\APrTL{j}$ denotes a unique module in the sector
$S^z=j$  which we call \textit{the spin-chain}
   $\JTL{N}$-module, and $\Xodd{j}$ is the one-dimensional and  $\Xodd{L}$ is the two-dimensional simple
$\centJTL$-module (the representation theory of the centralizer $\centJTL$ is described
 in Sec.~\ref{sec:rep-th-centJTL}). We show below the the modue $\APrTL{j}$  is indecomposable. We also set
 $\APrTL{L}\equiv\AStTL{L}{(-1)^{L-1}}$ in what follows.

The subquotient structure of the spin-chain $\JTL{N}$-modules $\APrTL{j}$ 
will be obtained using the centralizing property with the $\centJTL$ algebra which
is essentially the representation $\repQG$ of $\LQGodd$.
 The decomposition of the
 spin-chain over $\LQGodd$ together with the subquotient structure of
 each indecomposable summand $\PPodd{n}$ is given
 in~\eqref{eq:Hilb-decomp-Uqodd}.

In a double-centralizing situation, it would be sufficient to use
Thm.~\ref{lem:endo-PPodd} describing all 
intertwining operators between indecomposable  $\PPodd{n}$ and $\PPodd{m}$ modules over the
$\JTL{}$-centralizer  to reconstruct all
 arrows for the subquotient structure
depicting the $\JTL{N}$ action. The double centralizing property is
obvious in the  semisimple case but it is not evident for the
representation we consider here. We  show below that the
centralizer of $\centJTL$ is actually a larger algebra containing the
representation $\repgl$ of $\JTL{N}$. Our strategy  thus requires a
technical modification  of the ideas described in Sec.~\ref{subsec:imp-bimod}, where the
double-centralizing property was assumed for simplicity of the general discussion, and consists
of the following steps. First, we propose a
subquotient structure for the spin-chain $\JTL{N}$-modules $\APrTL{j}$;
we  study  all homomorphisms between these modules and
identify the corresponding intertwining operators with PBW basis
elements of $\LQGodd$ given in Def.~\ref{dfn:Uqodd}  that are
represented faithfully. At this step, we will only have a \textit{sufficient}
condition on the module structure to have the centralizer
$\centJTL$. To show that the subquotient structure indeed corresponds
to the $\JTL{N}$ action, we then carry out  a subsequent analysis involving fermionic
expressions for generators of $\repgl\bigl(\JTL{N}\bigr)$ computed
in~\cite{GRS1} and Thm.~\ref{lem:endo-PPodd} giving a subquotient
structure for $\Hilb_{N}$ as a module over the centralizer of
$\centJTL$.

Before going to the decomposition for any even $N$, 
we first  discuss the case $N=8$, to give the reader 
more experience with modules and arrows.

\subsubsection{Decomposition for $N=8$}\label{subsec:decom-N8}
Following the strategy described above and using the decomposition
over $\LQGodd$ for $N=8$ sketched in Fig.~\ref{Uq_odd_T-N8}, we find 
the decomposition of the full spin-chain over $\JTL{8}$: it is the direct sum~\eqref{decomp-JTL} with each 
 indecomposable summand $\APrTL{j}$ given from left
 to right corresponding to the decreasing value of $-4\leq j\leq 4$  in the sum
 \begin{equation}\label{PTL_decomp_N8}
    \includegraphics[scale=0.93]{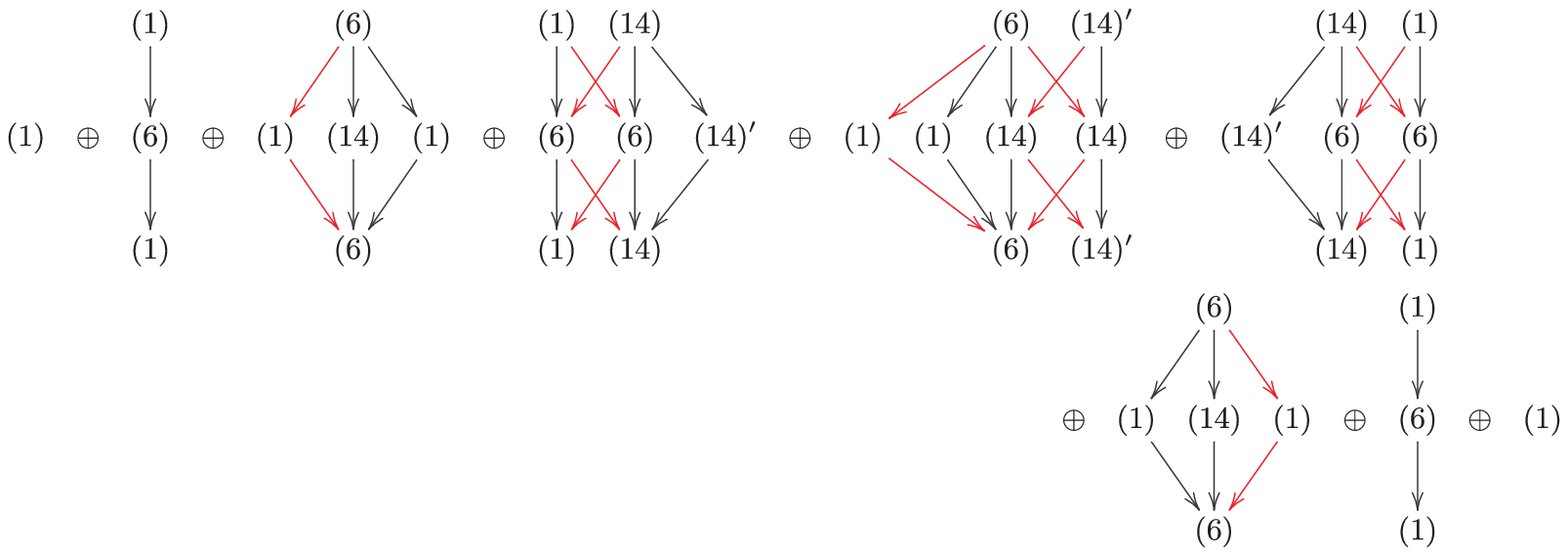}
\end{equation}
and we have set $\AIrrTL{1}{1}=(14)'$, $\AIrrTL{2}{-1}=(14)$,
$\AIrrTL{3}{1}=(6)$, and $\AIrrTL{4}{-1}=(1)$ indicating dimensions of
simple subquotents in the round brackets. In the diagrams, we
introduced arrows of two types as  in
Fig.~\ref{GLfig}
 --  black arrows show the action of the subalgebra
$\TL{N}\subset{\JTL{N}}$ generated by $e_j$, $1\leq j\leq N-1$ (this
fixes a basis in each  sector $S^z$, up to a basis in each $\TL{N}$-summand), and
red arrows
show the action of the last generator $e_N$ that mixes the
direct summands over $\TL{N}$ (the $e_N$ can act non-trivially also along
the black arrows in such a basis);
ignoring the red arrows gives a decomposition over $\TL{N}$ for each sector with $-3\leq
   S^z\leq3$ which is obtained from~\eqref{decomp-TL} by  restricting to a given  
 value of $S^z$. 
We depict the arrows
   without their projective-line parameters introduced in Sec.~\ref{prop:exts-PTL} for brevity.

 We describe now intertwining operators
respecting the subquotient structure proposed
in~\eqref{PTL_decomp_N8}. We note first that the only isomorphic
modules in~\eqref{PTL_decomp_N8} are the two invariants depicted by
$(1)$ and connected by the action of $\e^4$ and $\f^4$, otherwise we
would have an intertwining operator that does not belong $\centJTL$. From the
relations~\eqref{Uqodd-dfn-4} and the PBW basis given in
Def.~\ref{dfn:Uqodd}, we see that  each space
$\HomTL(\APrTL{j},\APrTL{k})$ should be one-dimensional whenever
$j-k=1\modd2$, and spanned by
$\FF{(j-k-1)/2}$ if $j>k$ (and $\EE{(k-j-1)/2}$, if $j<k$), times
appropriate polynomial in $\h$ projecting onto $\APrTL{j}$. We denote
the corresponding projector by $\proj_j(\h)$, and it is defined as
\begin{equation}\label{proj-def}
\proj_j(\h) = \prod_{n=-N/2;\,n\ne j}^{n=N/2}(2\h-n).
\end{equation}

In order to see the corresponding homomorphisms between the direct
summands in~\eqref{PTL_decomp_N8}, we consider  the case $j=0$ (the
right diagram) and $k=1$ (the left one)
\begin{equation*}\label{decomp-N8-Sz0}
    \xymatrix@R=28pt@C=16pt{
          &(1) \ar@[red][dr] \ar@[][d]
           &&(14)\ar@[red]@/_/[dll] \ar[dl] \ar[dr]&\\
        &(6)\ar@[red]@/_/[drr]_(.65){\ext{y}_-}\ar[d]_{\ext{a}}
             &(6)\ar@[red][dl]^(.65){\ext{y}_+}\ar[dr]^{\ext{b}}
              &&(14)' \ar[dl]_{\ext{c}}\\
        &(1)&&(14)&
}
    \xymatrix@R=31pt@C=4pt{
&&&\\
&\boldsymbol{\oplus}&&\\
&&&
}
    \xymatrix@R=28pt@C=16pt{ 
      &&(6) \ar@[red]@/_/[dll]_(.35){\ext{y}_+} \ar@[][dl]_{\ext{a'}} \ar@[][dr]_{\ext{b'}} \ar@[red]@/^/[drr]^(.35){\ext{y}_-}
 	  &&(14)'\ar@[red][dl]_(.33){\ext{y'}_+} \ar[d]^{\ext{c'}}&\\
    (1) \ar@[red]@/_/[drr]_(.33){}
     &(1) \ar@[][dr]_(.33){}
      &&(14)\ar[dl]_{y_0} \ar@[red][dr]^(.33){}
 	  &(14)\ar@[red]@/^/[dll]\ar[d]^{x_0}&\\
      &&(6)&&(14)'&
    }   
\end{equation*}
where we mark arrows by corresponding representatives from the
first-extensions groups for $\JTL{N}$, see
Lem.~\ref{thm:exts-PTL} and the discussion below the lemma (we do not
suppose that extensions marked by first latin letters, like $\ext{a}$,
are linear combinations of $\ext{x}_+$ and $\ext{y}_+$ introduced in Sec.~\ref{prop:exts-PTL}, and use only
the lower bound stated in Lem~\ref{thm:exts-PTL}.) The homomorphism
mapping the right diagram ($k=0$) to the left one ($j=1$) corresponds
to  $\E\,\proj_0(\h)$ and has the image isomorphic to a submodule with
the subquotient structure
$(1)\leftarrow(6)\rightarrow(14)\leftarrow(14)'$ where the subquotient
$(6)$ is in a linear combination of the pair of $(6)$'s in the left
diagram. The kernel of the homomorphism is generated from a linear
combination of the two $(1)$'s and a linear combination of the two
$(14)$'s in the right diagram. The linear combinations can in
principle be computed using the decomposition over $\TL{N}$ and
$\LQG$, and the action
in projective $\LQG$-modules from App.~\Approjmodbase~but we do not
need it. What we get are linear relations among the extensions
$\ext{a'}=\alpha\ext{a}+\gamma\ext{y}_+$ and
$\ext{b'}=\beta\ext{b}+\delta\ext{y}_-$, and similarly for $\ext{c'}$,
where $\alpha$, $\beta$, $\gamma$, $\delta$ are some complex
numbers. We see also that there are no more homomorphisms from the right
diagram to the left: an image isomorphic to $(1)\leftarrow(6)$ is not
possible because $\ext{y}_-$ and $\ext{b}$ in the left diagram are linearly independent,
an image isomorphic to $(1)\leftarrow(6)\rightarrow(14)$ is not
possible too because this would require the top $(14)'$ in the right
diagram to be in the kernel and  both $(14)$ to be  in the kernel
too. A similar analysis can be carried out  for all other pairs $(j,k)$, showing  
 that the decomposition~\eqref{PTL_decomp_N8} has an algebra
of intertwining operators isomorphic to $\centJTL$ indeed.

So far, we have only shown that a \textit{sufficient}
condition for the module structure to have the centralizer
$\centJTL$ holds. We cannot have more arrows in the diagrams, see a
general discussion after Thm.~\ref{Thm:homs-APrTL}. Next,  we show that
removing at least one red arrow in the  decomposition~\eqref{PTL_decomp_N8} results in
an enlarged endomorphism algebra. Indeed, let us suppose that the arrow
connecting the top $\IrrTL{1}=(14)'$ with $\IrrTL{2}=(14)$ and marked by $\ext{y'}_+$ in
the right diagram, for $S^z=0$, is absent. We note the self-conjugacy ($e_j^{\dagger}=e_j$) of the $\JTL{N}$-representation
$\repgl$ in~\eqref{rep-JTL-1} which implies that
$\APrTLs{0}\cong\APrTL{0}$.
Therefore, another arrow
mapping from the same subquotient $\IrrTL{2}$ to $\IrrTL{1}$ in the
bottom should be also absent. This means, there exists an extra homomorphism from $S^z=0$ to $S^z=1$  with the image
$(1)\leftarrow(6)\rightarrow(14)$ because in this case we can take the
top $\IrrTL{1}$ to be in the kernel and we  can still embed the top
$\IrrTL{3}=(6)$ into a linear combination of the two $\IrrTL{3}$'s in
the middle level of the left diagram due to the linear relation
between $\ext{b'}$, $\ext{b}$, and $\ext{y}_-$ stated above. The
extra homomorphism is not from $\centJTL$ and we thus get a contradiction with
Thm.~\ref{Thm:centr-JTL-main}. We could
similarly suppose that there is  the arrow marked by $\ext{y}_-$ in the left
diagram is absent, and the arrow from the top $\IrrTL{2}$ to $\IrrTL{3}$ should
be thus absent too. Then, the extra homomorphism from $S^z=0$ to
$S^z=1$ does not exist in general but we get an extra homomorphism
from $S^z=1$ to $S^z=2$ which is also not from $\centJTL$. The analysis can be repeated for any direct
summand in the decomposition~\eqref{PTL_decomp_N8}.

So far, we considered consequences of absence of arrows from a top
$\IrrTL{j}$ to $\IrrTL{j+1}$. To see what happens if we suppose the absence of an
arrow mapping a top $\IrrTL{j}$ to $\IrrTL{j-1}$ requires a still more
delicate analysis of the extensions. Let us suppose that the arrow
connecting the top $\IrrTL{3}=(6)$ with $\IrrTL{2}=(14)$ and marked by $\ext{y}_-$ in
the right diagram, for $S^z=0$, is absent. We also mark  the right-most
arrow from the top $\IrrTL{3}$ to the right node $\IrrTL{2}$ in the
diagram for $S^z=2$ (it is the third summand
in~\eqref{PTL_decomp_N8}) by a corresponding extension $\ext{d}$. Then, mapping by $\FF{0}=\F$ the module
$\APrTL{0}$ to $\APrTL{-1}$, and by $\FF{1}$ the 
$\APrTL{2}$ to $\APrTL{-1}$, we get that $\ext{b'}$ is proportional to
the $\ext{d}$ because the two operators map the two top $\IrrTL{3}$'s
to \textit{the same} linear combination of two $\IrrTL{3}$'s in the
middle level of $\APrTL{-1}$. On the other hand, mapping $\APrTL{2}$ to
$\APrTL{1}$ by $\F$ and $\APrTL{0}$ to $\APrTL{1}$ by $\E$, we get
that the two extensions $\ext{b'}$ and $\ext{d}$ should be linearly
independent, hence a contradiction.  This can  only be solved by the presence of the
arrow marked by $\ext{y}_-$ in the diagram for~$\APrTL{0}$.

We finally conclude that $\JTL{N}$
   action  mixes the direct summands over $\TL{N}$ in each sector into one indecomposable
   module in the way  described just above and in~\eqref{PTL_decomp_N8}.

 \subsection{The spin-chain decomposition  over $\JTL{N}$: the general case}\label{sec:JTL-decomp-full}
We now give  the spin-chain
 decomposition over $\JTL{N}$ for any even  number of sites $N$. 
Following the examples given above, we see
   that $\JTL{N}$ action  mixes all the projective modules over the
   subalgebra $\TL{N}$ in each subspace with $S^z=j$ into one
   indecomposable module $\APrTL{j}$.
Using the decomposition~\eqref{decomp-TL} over the $\TL{N}$
 subalgebra, we propose the subquotient
 structure for $\APrTL{0}$
\thispagestyle{empty}
\begin{figure}\centering
    \includegraphics[scale=0.9]{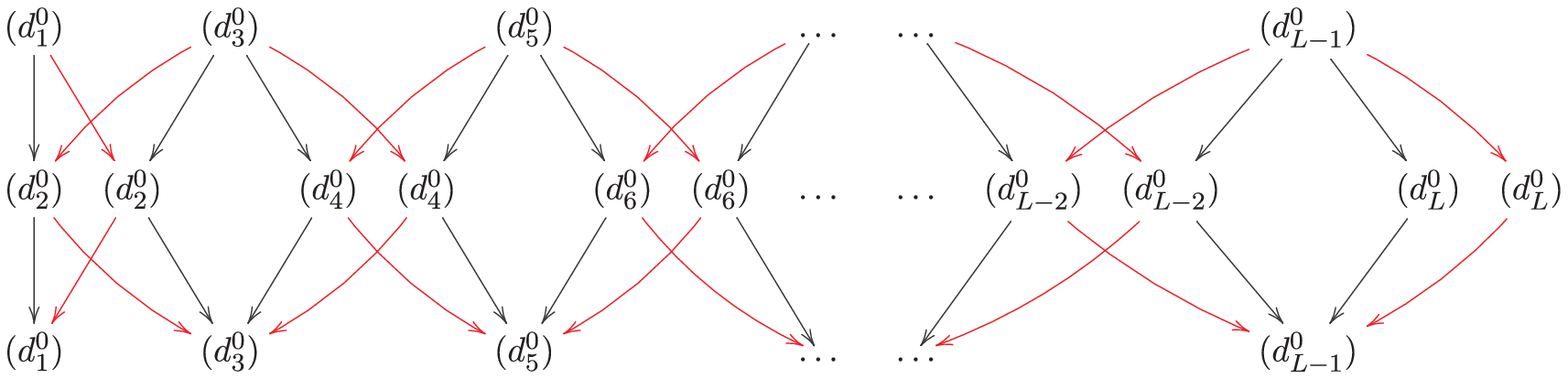}
      \caption{The subquotient structure for the $\JTL{N}$-module
 $\APrTL{0}$ corresponding to the sector   $S^z=0$ and  $L=0\modd2$.  For odd $L$, $\APrTL{0}$  is obtained by changing $L\to L+1$ and then removing nodes $(\TLX_{L+1})$.}
    \label{FF-JTL-mod-zero}
    \end{figure}
given in Fig.~\ref{FF-JTL-mod-zero} for $L=0\modd2$ ($2L=N$), and we set as usual $\AIrrTL{j}{(-1)^{j-1}}=(\TLX_{j})$.
Here, again as in the example for $N=8$, we see that removing all red arrows gives the decomposition over the open Temperley--Lieb $\TL{N}\subset\JTL{N}$ into a direct sum of its projective and trivial modules. 

We note that the diagram for $\APrTL{0}$ can be depicted in a more
familiar way as a module with the ``two-strands" subquotient
structure 
(of ``Feigin--Fuchs'' type)
 presented on the left of
Fig.~\ref{FF-JTL-mod}, where we do not use colors and it is supposed
that arrows connecting isomorphic subquotients correspond to linearly
independent extensions.

 For any sector with non-zero  $j=S^z$, we propose similarly
the subquotient structure for $\APrTL{j}$  given in
Fig.~\ref{FF-JTL-mod}, on the right. The tower for $\APrTL{0}$ has a
trivial top subquotient because $(\TLX_0)=0$ and is cut at the $L$-th
level, {\it i.e.}, it ends with  the pair of $(\TLX_L)$. Other towers have
a top, and also end with the pair of $(\TLX_L)$.  We note that the
two simple subquotients at each level of the ladders are
isomorphic. The Hamiltonian $H$ from~\eqref{hamil-def} acts by Jordan
blocks of rank $2$ on each pair of isomorphic simple subquotients
with one at the top (having only outgoing arrows) and the second
subquotient in the socle of
the module (having only ingoing arrows). The Jordan block structure is due to presence of zero
fermionic modes in the Hamiltonian as  observed in~\cite{GRS1}.
 \begin{figure}\centering
\begin{equation*}
  \xymatrix@R=24pt@C=30pt{
    &{\mbox{}}&&\\
    {(\TLX_1)}\ar@[][d]\ar[drr]
    &&{(\TLX_1)}\\
    {(\TLX_2)}\ar@[][urr]\ar[drr]
    &&{(\TLX_2)}\ar[u]\ar@[][d]\\
    {(\TLX_3)}\ar[u]\ar@[][d]\ar@[][urr]\ar[drr]
    &&{(\TLX_3)}\\
    {(\TLX_4)}\ar@[][urr] \ar@{.}[]-<0pt,18pt>
    &&{(\TLX_4)}\ar[u]
    \ar@{.}[]-<0pt,18pt>}
\qquad\qquad\qquad
  \xymatrix@=22pt
  {{}&(\TLX_{|j|})\ar[dr]&&\\
    {(\TLX_{|j|+1})}\ar[ur]\ar@[][d]\ar[drr]
    &&{(\TLX_{|j|+1})}\\
    {(\TLX_{|j|+2})}\ar@[][urr]\ar[drr]
    &&{(\TLX_{|j|+2})}\ar[u]\ar@[][d]\\
    {(\TLX_{|j|+3})}\ar[u]\ar@[][d]\ar@[][urr]\ar[drr]
    &&{(\TLX_{|j|+3})}\\
    {(\TLX_{|j|+4})}\ar@[][urr] \ar@{.}[]-<0pt,18pt>
    &&{(\TLX_{|j|+4})}\ar[u]
    \ar@{.}[]-<0pt,18pt>}
\end{equation*}
      \caption{The two-strands structure of the spin-chain
      $\JTL{N}$-modules $\APrTL{j}$ for $j=0$ on the left and $j\ne0$
      on the right side. The towers are ended by a pair of
      $(\TLX_L)$. We do not show red arrows (corresponding to $e_N$
      action) used in Fig.~\ref{FF-JTL-mod-zero} but they can be
      easily recovered using the decomposition over $\TL{N}$.}
    \label{FF-JTL-mod}
    \end{figure}
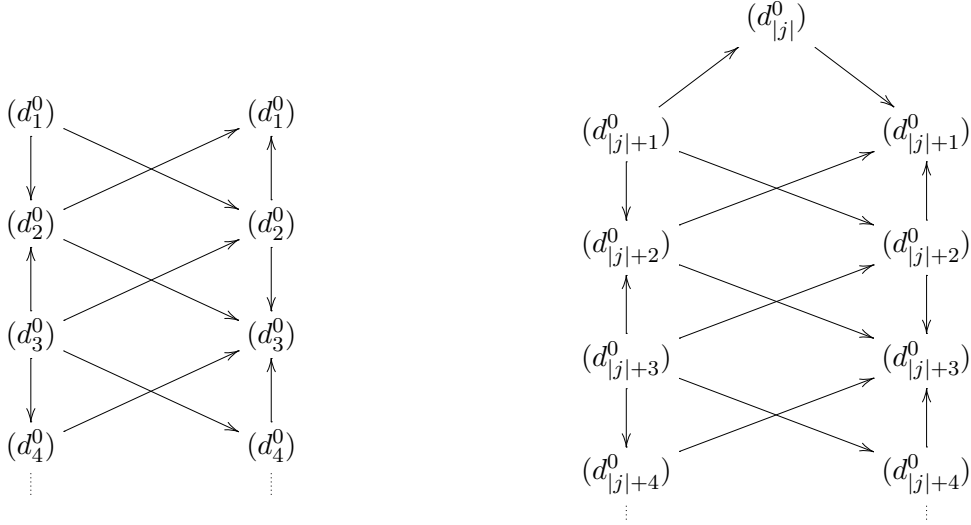

It is important to note that modules $\APrTL{j}$ and $\APrTL{-j}$ are
not isomorphic, otherwise we would have an intertwining operator not
from $\centJTL$, -- the only isomorphic ones are the trivials
$\APrTL{-L}\cong\APrTL{L}=(1)$ connected by the action of $\f^L$ and
$\e^L$ from $\centJTL$. For each $-L+1<j<L-1$, the module $\APrTL{j}$
is fully described by a sequence of parameters $\{x_i\}$ on a complex
projective line, $x_i\in\oC\oP^1$, which were introduced
after~\eqref{mod-first-quotient} and mentioned also
in~\ref{sec:spin-chain-decomp-four} for a particular case. We leave
this characterization for  future work.  We only note that in a faithful
representation, say, in a direct sum of all projective JTL modules or
its tilting modules 
such parameters would not appear (as the indecomposable projective and tilting modules are uniquely characterized by their subquotient structure, up to an isomorphism) but in a non-faithful representation we know such parameters might appear because of the  two (or higher) dimensionality of the
first extension groups, see the discussion in Sec.~\ref{prop:exts-PTL}.  The question
about these parameters  is actually related to  question of how  the
spin-chain modules are obtained by taking particular quotients of projective
$\JTL{N}$-modules. The indecomposable
modules we encounter are particular quotients of a direct sum
of projective covers over $\JTL{N}$. To cover a module $\APrTL{j}$,
one should take the direct sum $\oplus_{k=j}^{L-1}\mathrm{Proj}_k$ of
projective covers $\mathrm{Proj}_k$ for each simple $\JTL{N}$-module
$\AIrrTL{k}{(-1)^{k+1}}$.

\subsection{The indecomposable ``zig-zag'' modules}\label{sec:zig-zag-mod}
Before
describing all homomorphisms between $\APrTL{j}$ and $\APrTL{k}$, with
$-L\leq j,k\leq L$,  we introduce 
indecomposable modules of ``zig-zag'' shape. For simplicity, we
consider only the case of even $L$ (or $N=0\;\textrm{mod}\;4$); the odd $L$ case is quite similar. For a positive odd $n$, let $k=(L-n+1)/2$. Then, taking a
quotient of the direct sum
$\MmodJTL{1}{n}{x_1,y_1}\oplus\MmodJTL{1}{n+2}{x_2,y_2}\oplus\dots\oplus\MmodJTL{1}{L-1}{x_k,y_k}$ of the $\JTL{N}$-modules introduced
in~\eqref{mod-first-quotient} by a submodule
$\AIrrTL{n+1}{(-1)^n}\oplus\AIrrTL{n+3}{(-1)^n}\oplus\dots\oplus\AIrrTL{L-2}{(-1)^n}$
we get a family of $\JTL{N}$-modules $\MJTL{n-1}{x_1,y_1,\dots,x_k,y_k}$, where $x_i,y_i\in\oC\oP^1$.
This family is thus parametrized by the
set
$\{x_i,y_i\,|\, 1\leq i \leq k\}\in\oC\oP^1\times\dots\times\oC\oP^1$ and sketched at the top of   Fig.~\ref{JTL-Mmod}, where 
we set $(\TLX_{n})\equiv\AIrrTL{n}{(-1)^{n+1}}$ as usual and
$(\TLX_{L+1})\equiv0$. 
We define similarly a family of $\JTL{N}$-modules $\NJTL{n-1}{x_0,y_0,\dots,x_k}$ sketched  at
the bottom of Fig.~\ref{JTL-Mmod}, with even $n$ and $k=(L-n)/2$. In what follows, we also use the modules
$\MJTLs{n-1}{x_1,y_1,\dots,x_k,y_k}$ and
$\NJTLs{n-1}{x_0,y_0,\dots,x_k}$ with all arrows reversed.

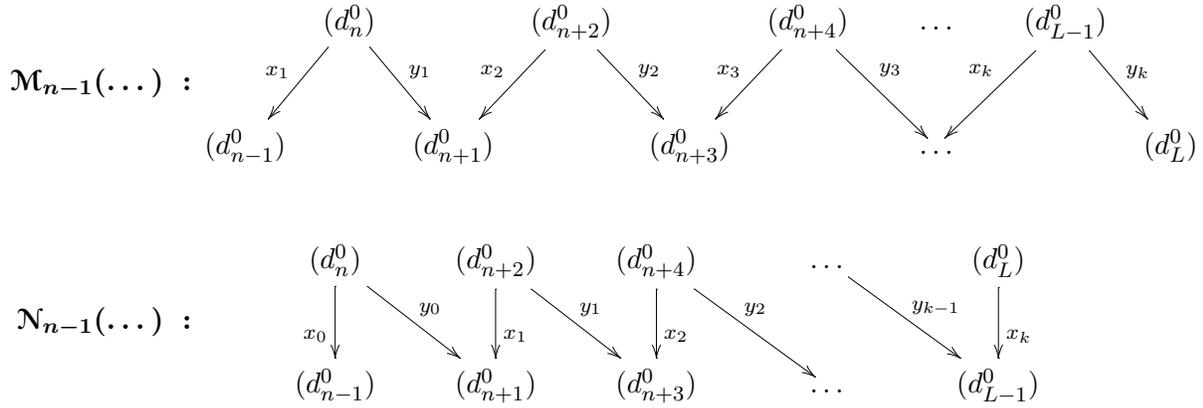
\begin{figure}\centering
\begin{align*}
\xymatrix@R=15pt@C=3pt@W=2pt@M=1pt
{
&\mbox{}\\
&\boldsymbol{\MJTL{n-1}{\dots}\;:}
}&
\xymatrix@R=26pt@C=6pt@W=4pt@M=4pt
{
&{(\TLX_n)}\ar[dr]^(.5){y_1}\ar[dl]_(.5){x_1}&&
{(\TLX_{n+2})}\ar[dr]^(.5){y_2}\ar[dl]_(.5){x_2}&&
{(\TLX_{n+4})}\ar[dr]^(.5){y_3}\ar[dl]_(.5){x_3}&
\dots&
{(\TLX_{L-1})}\ar[dr]^(.5){y_k}\ar[dl]_(.5){x_k}&
\\
 (\TLX_{n-1})&
 &(\TLX_{n+1})&
   &(\TLX_{n+3})&
   &\mbox{}\quad\dots\quad\mbox{}&
   &(\TLX_{L})&
     }\\
&\mbox{}\\
\xymatrix@R=15pt@C=3pt@W=2pt@M=1pt
{
&\mbox{}\\
&\boldsymbol{\NJTL{n-1}{\dots}\;:}
}&
\xymatrix@R=26pt@C=22pt@W=4pt@M=4pt
{
&{(\TLX_n)}\ar[dr]^(.5){y_0}\ar[d]_(.6){x_0}&
{(\TLX_{n+2})}\ar[dr]^(.5){y_1}\ar[d]^(.6){x_1}&
{(\TLX_{n+4})}\ar[dr]^(.5){y_2}\ar[d]^(.6){x_2}&
\dots\ar[dr]^(.5){y_{k-1}}&
{(\TLX_{L})}\ar[d]^(.6){x_k}&
\\
& (\TLX_{n-1})&
 (\TLX_{n+1})&
  (\TLX_{n+3})&
  \mbox{}\quad\dots\quad\mbox{}&
   (\TLX_{L-1})&
     }
\end{align*}
\caption{For even $L$, the indecomposable ``zig-zag'' $\JTL{N}$-modules
  $\MJTL{n-1}{x_1,y_1,\dots,x_k,y_k}$ at the top, with odd $n$, 
  and  $\NJTL{n-1}{x_0,y_0,\dots,x_k}$ at the bottom, with even $n$,
   $k=(L-\bar{n})/2$, and $\bar{n} = n - (n\,\textrm{mod}\,2)$. Each
  single arrow stands for a doubled arrow as introduced in~\eqref{mod-first-quotient}
  and is characterized by a point on a complex projective line, $x_i,y_i\in\oC\oP^1$.}
    \label{JTL-Mmod}
    \end{figure}

The spin-chain module $\APrTL{0}$ is an extension/glueing
$\NJTLs{1}{\dots}\to\NJTL{1}{\dots}$ of two modules of the
$\mathscr{N}$-type, where appropriate parameters stand in the round
brackets. As we said before, to determine the parameters is out of the
scope of the paper and we will thus use the notation without
specifying them explicitly. We only note that a canonical
way to specify the submodule $\NJTL{1}{\dots}$ in $\APrTL{0}$ is to take the kernel
of the quantum-group generator $\F$ in $\APrTL{0}$; we remind that
(the representation~\eqref{QG-ferm-1}
 of) $\F$ belongs to the $\JTL{N}$-centralizer
$\centJTL$ on the spin-chain and therefore its kernel is a $\JTL{N}$-module. That the kernel of $\F$
is isomorphic to a $\NNJTL{1}$ module (with appropriate parameters)
easily follows from the decompositions over $\LQG$ in~\eqref{decomp-LQG} and over
$\TL{N}$ in~\eqref{decomp-TL}
restricted to $S^z=0$ and the explicit action of $\LQG$ given in
App.~\Approjmodbase.

The $\JTL{N}$-modules $\APrTL{j}$, for $j\ne0$, proposed on the right
side of Fig.~\ref{FF-JTL-mod} can be also obtained as an extension of
two modules $\NJTL{|j|}{\dots}$ and $\MJTL{|j|+1}{\dots}$, for odd
$j$, with the first one being a subquotient and the second  a
submodule of $\APrTL{j}$, and similarly for the even-$j$ case. For any
$j$, using again the decompositions~\eqref{decomp-LQG} and~\eqref{decomp-TL} restricted
to the supbspace with $S^z=\pm j$ and the $\LQG$-action from
App.~\Approjmodbase, we obtain the following short exact sequences of
$\JTL{N}$-modules
\begin{align}
\mbox{}\qquad\qquad
0\;\to\;\NJTL{1}{\dots}\;\to&\;\APrTL{0}\;\to\;\NJTLs{1}{\dots}\;\to\;0,&\label{sh-ex-seq-JTL-1}\\
0\;\to\;\MJTL{|j|+1}{\dots}\;\to&\;\APrTL{j}\;\to\;\NJTL{|j|}{\dots}\;\to\;0,
\qquad  j - \,\text{odd},&\\
0\;\to\;\NJTL{|j|+1}{\dots}\;\to&\;\APrTL{j}\;\to\;\MJTL{|j|}{\dots}\;\to\;0,
\qquad  j - \,\text{even},&\label{sh-ex-seq-JTL-3}
\end{align}
where we define the submodules $\MJTL{|j|+1}{\dots}$ and
$\NJTL{|j|+1}{\dots}$ as the kernels of the
quantum-group generator $\F$ on $\APrTL{j}$, for $j>0$, and the kernels
of $\E$, for $j<0$,  for odd and even $j$, respectively. The projective-line parameters in the round brackets of all modules in~\eqref{sh-ex-seq-JTL-1}-\eqref{sh-ex-seq-JTL-3} are thus uniquely fixed.

Using the self-conjugacy ($e_j^{\dagger}=e_j$,
with $1\leq j\leq N$) of the $\JTL{N}$-representation $\repgl$
in~\eqref{rep-JTL-1}  which implies that $\APrTLs{j}\cong\APrTL{j}$, we obtain the
dual short exact sequences of $\JTL{N}$-modules
\begin{align}
\mbox{}\qquad\qquad
0\;\to\; \NJTLs{|j|}{\dots}\;\to&\;\APrTL{j}\;\to\;\MJTLs{|j|+1}{\dots}\;\to\;0,
\qquad j - \,\text{odd},&\label{sh-ex-seq-JTL-4}\\
0\;\to\; \MJTLs{|j|}{\dots}\;\to&\;\APrTL{j}\;\to\;\NJTLs{|j|+1}{\dots}\;\to\;0,
\qquad j - \,\text{even},&\label{sh-ex-seq-JTL-5}
\end{align}
where the submodules $\NJTLs{|j|}{\dots}$ and
$\MJTLs{|j|}{\dots}$ are now defined as the kernels of the
quantum-group generator $\E$ on $\APrTL{j}$, for $j>0$, and the kernels
of $\F$, for $j<0$, for odd and even $j$, respectively.  Then, parameters in the round brackets  in~\eqref{sh-ex-seq-JTL-4}-\eqref{sh-ex-seq-JTL-5} are also uniquely fixed.

We next use the zig-zag modules described above and the short exact
sequences in description of all
homomorphisms between the spin-chain $\JTL{N}$-modules proposed above.

\begin{Thm}\label{Thm:homs-APrTL}
For $-L\leq j,k\leq L$, the space of homomorphisms between $\APrTL{j}$ and $\APrTL{k}$ has
the dimension
\begin{equation}
\dim \Hom(\APrTL{j},\APrTL{k})=
    \begin{cases}
      1,&\quad j - k = 1\mod 2,\\
      \half\bigl(L-\max(|j|,|k|) +
      j\,\mathrm{mod}\,2\bigr)+\delta_{j,k},&\quad j - k=0\mod 2,
    \end{cases}
\end{equation}
and the one-dimensional space in the case $j - k = 1\mod 2$ is given
by the map $f_{j,k}\in\Hom(\APrTL{j},\APrTL{k})$ with its image
\begin{equation}\label{f-jk-image}
\mathrm{im}(f_{j,k}) \cong
\begin{cases}
\NJTL{|j|}{\dots},&\quad j-\text{odd} \;\;\text{and}\;\; |j|>|k|,\\
\MJTLs{|k|}{\dots},&\quad j-\text{odd} \;\;\text{and}\;\; |j|<|k|,\\
\MJTL{|j|}{\dots},&\quad j-\text{even}\;\;\text{and}\;\; |j|>|k|,\\
\NJTLs{|k|}{\dots},&\quad j-\text{even}\;\;\text{and}\;\; |j|<|k|,\\
\end{cases} 
\end{equation}
with appropriate parameters from $\oC\oP^1$  in the
round brackets. 

In the case $j - k = 0\mod 2$, the $\Hom$-space is spanned by homomorphisms with
semisimple images.
\end{Thm}

We give only an idea of the proof. The case  $j - k$ is even is obvious and follows
from the subquotient structure of $\APrTL{j}$ given in
Fig.~\ref{FF-JTL-mod}: a basis in the space
$\Hom(\APrTL{j},\APrTL{k})$ can be chosen as the homomorphisms having
the images isomorphic to $\AIrrTL{j}{(-1)^{j+1}}$.

The case  $j - k$  odd can be analyzed by taking a concatenation of
 the short exact
 sequences~\eqref{sh-ex-seq-JTL-1}-\eqref{sh-ex-seq-JTL-5} with the
mappings $\F$ and $\E$. The result of such concatenation are
 two cochain complexes with the differentials $\F$ and $\E$
(we recall that $\F^2=\E^2=0$)
\begin{gather}
0\to \APrTL{L} \xrightarrow{\;\F\;} \APrTL{L-1} \xrightarrow{\;\F\;}
\dots \xrightarrow{\;\F\;}
\APrTL{j+1} \xrightarrow{\;\F\;} \APrTL{j} \xrightarrow{\;\F\;}
\APrTL{j-1} \xrightarrow{\;\F\;} \dots \xrightarrow{\;\F\;} \APrTL{-L} \to 0,\label{ch-complex-1}\\
0\to \APrTL{-L} \xrightarrow{\;\E\;} \APrTL{-L+1} \xrightarrow{\;\E\;}
\dots \xrightarrow{\;\E\;}
\APrTL{j-1} \xrightarrow{\;\E\;} \APrTL{j} \xrightarrow{\;\E\;}
\APrTL{j+1} \xrightarrow{\;\E\;} \dots \xrightarrow{\;\E\;} \APrTL{L} \to 0,\label{ch-complex-2}
\end{gather}
which have trivial cohomologies, {\it i.e.}, they are long exact
sequences. The images (and therefore the kernels) of these
differentials are the zig-zag $\JTL{N}$-modules described
in~\eqref{sh-ex-seq-JTL-1}-\eqref{sh-ex-seq-JTL-5}. This proves
existence of the homomorphisms $f_{j,j\pm1}$ with the
properties~\eqref{f-jk-image} in the case $j - k=\pm1$. Existence for
all other cases is proven by taking into account the commutation of
$\JTL{N}$ with operators $\FF n$ and $\EE m$ from the representation
$\repQG$ of $\LQGodd$. To compute their images, we first recall 
 the decompositions~\eqref{decomp-LQG} and~\eqref{decomp-TL}  over $\LQG$ and
 $\TL{N}$, respectively, restricted to the subspaces with
$\repQG(\h)=j/2$ and $\repQG(\h)=k/2$. These subspaces are mapped to each other by
the generators $\FF{(j-k-1)/2}$ and $\EE{(j-k-1)/2}$. Then, the isomorphisms~\eqref{f-jk-image}
are obtained using the explicit $\LQG$-action on the direct summands (given in App.~\Approjmodbase)
and the homomorphism of algebras from Rem.~\ref{rem:inj-hom-LQGodd}. For example,
the image of $\FF{(j-k-1)/2}$ in the subspace with $\repQG(\h)=k/2$ is
isomorphic to the $\JTL{N}$-module
 $\NJTL{|j|}{\dots}$, when $j>k$ and $j$ is odd, {\it etc.}

Recall then that each module $\APrTL{j}$ is a glueing of two zig-zag modules~\eqref{sh-ex-seq-JTL-1}--\eqref{sh-ex-seq-JTL-5}, {\it e.g.}, for $j$ odd it has $\NJTL{|j|}{\dots}$ as its subquotient and $\MJTL{|j|+1}{\dots}$ as its submodule. To prove that there are no other homomorphisms (up to an overall
 rescaling) between $\APrTL{j}$ and $\APrTL{k}$, with $j-k$  an odd
 number, it is sufficient to consider filtrations of the zig-zag
 submodules/subquotients in $\APrTL{j}$ and $\APrTL{k}$ by their (smaller) zig-zag submodules/subquotients, {\it i.e.}, by those with higher sub-index, see
 Fig.~\ref{JTL-Mmod} where the smaller zig-zag submodules are easily
 identified. For example, we have a filtration $\dots\subset\MJTL{j+1}{\dots}\subset\MJTL{j-1}{\dots}$.  Any homomorphism
 should obviously respect the filtrations. Then, using the  subquotient structure for
 $\APrTL{j}$ proposed in Fig.~\ref{FF-JTL-mod} and assuming existence of a homomorphism with a kernel non-isomorphic to the kernel of $f_{j,k}$ constructed just above we readily see that such a homomorphism does not respect the filtrations, thus a contradiction. Care should be taken for a pair of arrows connecting isomorphic pair of subquotients --
 these arrows correspond to
 linearly independent elements from the first extension groups in
 Lem.~\ref{thm:exts-PTL}. 

\subsection{Intertwiners and the PBW basis}
We now identify all the homomorphisms from the space
$\EndJTL(\Hilb_{N})=\bigoplus_{j,k=-L}^L\Hom(\APrTL{j},\APrTL{k})$ with the PBW basis
elements in
$\LQGodd$ that are represented faithfully on the
spin-chain. For $j-k$ an odd number and for each
$f_{j,k}\in\Hom(\APrTL{j},\APrTL{k})$ described in
Thm.~\ref{Thm:homs-APrTL}, we have  the equalities (by the construction of $f_{j,k}$ in the proof of Thm.~\ref{Thm:homs-APrTL})
\begin{gather}
f_{j,k} = \repQG\bigl(\FF{(j-k-1)/2}\,\proj_j(\h)\bigr), \qquad \text{for} \; j>k,\label{homs-Uqoddgen-1}\\
f_{j,k} =  \repQG\bigl(\EE{(k-j-1)/2}\,\proj_j(\h)\bigr), \qquad \text{for} \; j<k,\label{homs-Uqoddgen-2}
\end{gather}
where projectors $\proj_j$ onto $\APrTL{j}$ are polynomials in $\h$
introduced in~\eqref{proj-def}.  The homomorphisms $f_{j,k}$ in the
case $j-k$  an even number are identified with $\repQG(\proj_j(\h))$
if $j=k$ and otherwise with composites of the generators $\FF n$ and
$\EE m$, times the projector $\proj_j(\h)$. Proceeding then by counting basis elements (a simple
calculation) in the image of $\LQGodd$ on the spin-chain
one would obtain that the operators constructed exhaust the PBW basis in the image of $\LQGodd$.

We thus have shown  that a sufficient
condition on the module structure in Fig.~\ref{FF-JTL-mod} to have the centralizer
$\centJTL$ holds. To show that the subquotient structure indeed corresponds
to the $\JTL{N}$ action, we do a further and final analysis.

\subsection{Final analysis}
 To finish
our exposition of the proof that  the proposed subquotient structure for
$\APrTL{j}$ is correct, we describe next the subquotient structure for $\APrTL{j}$
considered as a module over the centralizer of $\centJTL$ which is
isomorphic by the definition to the algebra $\EndcJTL(\Hilb_N)$. The
centralizer obviously contains $\repgl\bigl(\JTL{N}\bigr)$ as a
subalgebra. The opposite inclusion is not true, as we show now.

 The subquotient
structure can be obtained using intertwining operators respecting
$\centJTL$ action. These are described in Thm.~\ref{lem:endo-PPodd}. The only
difference from the diagrams for $\JTL{N}$ in Fig.~\ref{FF-JTL-mod} is
that there are additional (``long') arrows mapping a top
subquotient~$\IrrTL{j}$ (having only outgoing arrows) to~$\IrrTL{k}$
in the socle (having only ingoing arrows) whenver $|j-k|\geq4$ is an
even number. We note that these long arrows are not composites of any short
arrows mapping from the top to the middle level, and from the middle
to the socle; this distinguishing property appears only at
$N\geq10$. It turns out that $\JTL{N}$ generators correspond only to
these short arrows and not to the long ones, and therefore there is no
 element from $\JTL{N}$ represented by a long arrow. This can be
shown using a direct calculation with  fermionic expressions for
$e_j$ and $u^2$ (see~\eqrefej and~\eqrefuu in our first paper~\cite{GRS1}) in a basis
of root vectors of the Hamiltonian $H$ from~\eqref{hamil-def}. 
Indeed, the expression for $e_j$ is a bilinear combination of $2(N-1)$ generators
of a Clifford algebra. Half of them ($N-2$ creation modes
$\chi^{\dagger}_{p>0}$ and $\eta_{p>0}$ in notations of~\cite{GRS1}, Sec.~\secrefcr) generates
the bottom level -- the intersection of the kernels of $\F$ and $\E$
in $\Hilb_N$ -- from the vacuum state $\vac$, and also the top level
from one cyclic vector $\lvac$ which is involved with $\vac$ into a
Jordan cell for $H$. Among the Clifford algebra generators, there are
two -- zero modes  $\eta_0$ and $\chi^{\dagger}_0$ -- proportional to
$\F$ and $\E\K^{-1}$, respectively. These are the only
generators mapping vectors from the top level to the middle level, and
from the middle to the bottom level. We see from the
expression~\eqrefuuu in~\cite{GRS1} that  $u^2$, is
also a sum of monomials in the Clifford algebra, none of these monomials  
containing the product of the two zero modes. The product maps the top
to the bottom and a monomial containing it could thus correspond to a
long arrow. The $e_j$'s have such a monomial but it is quadratic,
{\it i.e.}, proportional to the product of the zero modes, and thus commutes
with the $\JTL{N}$ action
and  maps a top subquotient $\IrrTL{j}$ only to the bottom $\IrrTL{j}$. The
fermionic expression for $e_j$ has also other terms/monomials containing only
one of the zero modes and they thus map only by one level down.
We conclude that the action of $e_j$, with $1\leq j\leq N$, and $u^2$ cannot correspond to those long arrows connecting the top and the bottom
and which are not composites of  short arrows. This  proves
that there are no such arrows in diagrams for the subquotient
structure of $\JTL{N}$-modules~$\APrTL{j}$. We can thus conclude that  
the algebra $\repgl\bigl(\JTL{N}\bigr)$ does not contain the double centralizer $\EndcJTL(\Hilb_N)$.

Finally, we observe that removing at least one red arrow from the
diagrams for $\APrTL{0}$ in Fig.~\ref{FF-JTL-mod-zero} or for
$\APrTL{j}$ in Fig.~\ref{FF-JTL-mod} results in an enlarged
endomorphism algebra (black arrows should be present due to the action
of the subalgebra $\TL{N}$.) Indeed, removing a red arrow mapping from
a top subquotient $\IrrTL{j}$ to $\IrrTL{j+1}$ in the middle we should
remove also the red arrow mapping from the same subquotient
$\IrrTL{j+1}$ to $\IrrTL{j}$ in the bottom because of the
self-conjugacy ($e_j^{\dagger}=e_j$) of the $\JTL{N}$-representation
$\repgl$ in~\eqref{rep-JTL-1} which implies that
$\APrTLs{j}\cong\APrTL{j}$. Then, we can repeat the same analysis as
in Sec.~\ref{subsec:decom-N8} for $N=8$ and get an additional
intertwining operator not from $\centJTL$ but this  contradicts to
Thm.~\ref{Thm:centr-JTL-main}. Removing a red arrow connecting
$\IrrTL{j}$ and $\IrrTL{j-1}$ results eventually in a contradiction to
a statement related to Lem.~\ref{thm:exts-PTL} in a way very similar to what was stated also
in the example for $N=8$ in Sec.~\ref{subsec:decom-N8}. We do not give
a proper generalization of the results for $N=8$ because of their
simplicity.  This analysis finishes our exposition of the proof for  the  subquotient structure of
$\APrTL{j}$ modules over $\JTL{N}$ proposed in Fig.~\ref{FF-JTL-mod}.

\subsection{Comparison with the standard modules}\label{subsec:St-vs-SpCh}

{\footnotesize
\newcommand{\myar}{\ar@[|(2.5)]@[]}
 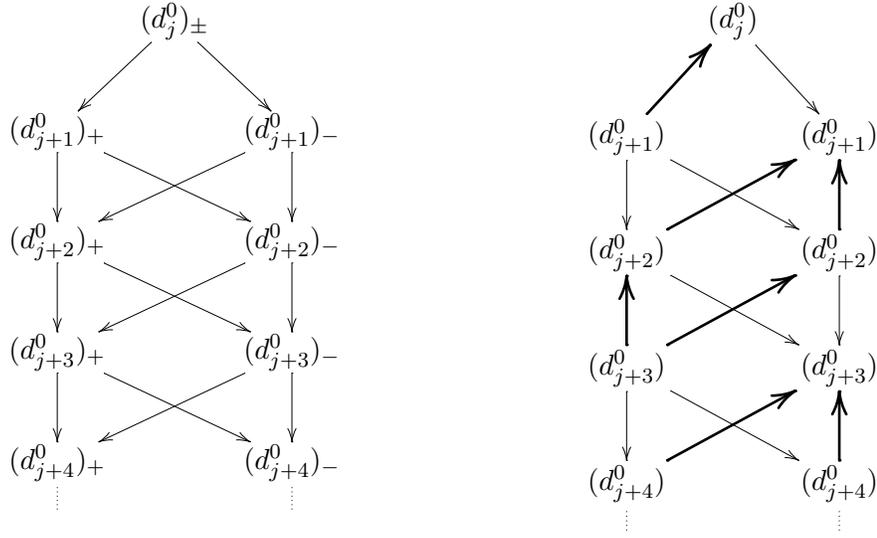
\begin{figure}\centering
 \begin{equation*}
   \xymatrix@R=26pt@C=11pt@W=1pt@M=1pt
   {{}&(\TLX_j)_{\pm}\ar[dr]\ar[dl]&&\\
     {(\TLX_{j+1})_{+}}\ar[d]\ar[drr]
     &&{(\TLX_{j+1})_{-}}\ar[d]\ar[dll]\\
     {(\TLX_{j+2})_{+}}\ar[d]\ar[drr]
     &&{(\TLX_{j+2})_{-}}\ar[d]\ar[dll]\\
     {(\TLX_{j+3})_{+}}\ar[d]\ar[drr]
     &&{(\TLX_{j+3})_{-}}\ar[d]\ar[dll]\\
     {(\TLX_{j+4})_{+}} \ar@{.}[]-<0pt,18pt>
     &&{(\TLX_{j+4})_{-}}
     \ar@{.}[]-<0pt,18pt>}
 \qquad\qquad\qquad\quad
  \xymatrix@R=26pt@C=12pt@W=2pt@M=2pt
  {{}&(\TLX_j)\ar[dr]&&\\
    {(\TLX_{j+1})}\myar[ur]\ar[d]\ar[drr]
    &&{(\TLX_{j+1})}\\
    {(\TLX_{j+2})}\myar[urr]\ar[drr]
    &&{(\TLX_{j+2})}\myar[u]\ar[d]\\
    {(\TLX_{j+3})}\myar[u]\ar[d]\myar[urr]\ar[drr]
    &&{(\TLX_{j+3})}\\
    {(\TLX_{j+4})}\myar[urr] \ar@{.}[]-<0pt,18pt>
    &&{(\TLX_{j+4})}\myar[u]
    \ar@{.}[]-<0pt,18pt>}
 \end{equation*}
      \caption{The structure of the spin chain modules $\APrTL{j}$ at $\q=i$ (the
        right one). The thick   arrows  have been flipped with
        respect to the structure of the standard modules on the left side.}
    \label{GLfigflip1}
    \end{figure}
}

We finally give  a qualitative  characterization of the spin-chain modules
$\APrTL{j}$ in the context of the standard modules in
Fig.~\ref{GLfig2} discussed in Sec.~\ref{sec:stand-q-i}. The
subquotient structure of the $\JTL{N}$-modules in the spin chain
is obtained by flipping half the arrows in the standard modules of $\ATL{N}$ and ignoring the subscript $\pm$ (distinguishing only non-isomorphic simple $\ATL{N}$-subquotients but not the ones over $\JTL{N}$),
as illustrated on Fig.~\ref{GLfigflip1}. This is similar
to what happens when comparing   Verma and Feigin--Fuchs modules over a
Virasoro algebra.
Note that we do not use here  the standard modules for $\JTL{N}$
which turn out to have no arrows inside the tower on
Fig.~\ref{GLfig2}; we believe this latter feature is a peculiarity of
the case $\q=i$.

\section{Bimodule structure in the closed $\gl(1|1)$ spin-chains}\label{ind-chain-bimod-sec}
In this short section, we find the subquotient structure of  the bimodules $\Hilb_{2L}$  over the pair
of the two commuting algebras centralizing each other both in the periodic and antiperiodic $\gl(1|1)$ spin-chains.

\subsection{Bimodule over $\JTL{N}$ and $\centJTL$}\label{ind-chain-bimod-subsec}
 We use the spin-chain decomposition~\eqref{decomp-JTL} over $\JTL{N}$ described in Secs.~\ref{sec:TL-decomp} and~\ref{sec:JTL-decomp-full} and the intertwining operators from Thm.~\ref{Thm:homs-APrTL} to study the structure of the bimodule $\Hilb_N$ over the two  algebras $\JTL{N}$ and $\centJTL$.

 \begin{figure}\centering
\mbox{}\bigskip\\
\mbox{}\medskip
  \def\svgwidth{470pt}
        \includegraphics[scale=1]{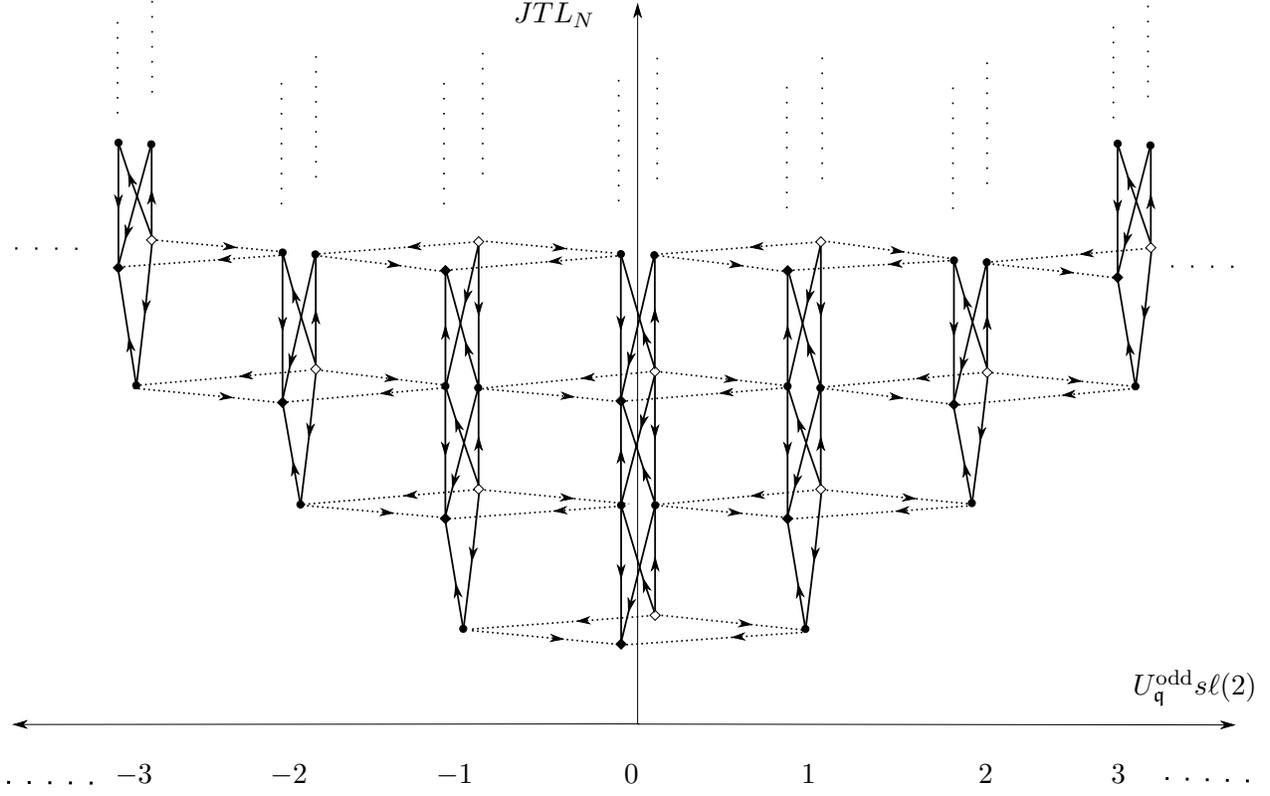}
      \caption{Bimodule over the pair $(\JTL{N},\LQGodd)$ of commuting
      algebras.  The action of $\JTL{N}$ is depicted by vertical
      arrows while the action of $\LQGodd$ is shown by dotted
      horizontal lines. Each label $j$ in the horizontal axis corresponds to the sector for $S^z$ and the label runs from $-L$ on the left to $j=L$ on 	
      the right. Each vertical tower above a label $j$ is the diagram for $\APrTL{j}$.
      The first horizontal layer at the bottom contains four
      nodes $(\TLX_1)$ and dotted arrows mixing them compose the
      $\LQGodd$-module $\PPodd{1}$. The second layer contains eight
      nodes $(\TLX_2)$ and the dotted arrows depict the action in the
      indecomposable module $\PPodd{2}$ presented on
      Fig.~\ref{Uq_odd_T-N8} in the front of $(\TLX_2)$, {\it etc.} We
      suppress long-range arrows representing action of the generators
      $\FF{>0}$ and $\EE{>0}$ in order to simplify diagrams. For
      example, the second layer of the bimodule contains in addition
      four long arrows going from the node $\diamond$ at $j=\mp1$ to the
      node $\bullet$ at $j=\pm2$, and from the node $\bullet$ at
      $j=\pm2$ to the node {\scriptsize$\vardiamond$} at $j=\mp1$.}
    \label{JTL-Uqodd-bimod}
    \end{figure}

One way to describe the bimodule $\Hilb_{2L}$ is to consider
the cochain complexes~\eqref{ch-complex-1} and~\eqref{ch-complex-2}
with the differentials $\F$ and $\E$.  The images (and the kernels) of
these differentials and of the generators $\FF n$ and $\EE m$ are the
zig-zag $\JTL{N}$-modules described in Sec.~\ref{sec:zig-zag-mod} and Thm.~\ref{Thm:homs-APrTL},
with the use of the
identifications~\eqref{homs-Uqoddgen-1} and~\eqref{homs-Uqoddgen-2}. The
centralizer $\centJTL$ then acts on each of these complexes
in a ``long-range" way
mapping terms with $S^z=j$ to ones with $S^z=k$ and with the only
condition that $|j-k|=1\modd2$.

We finally give 
the  diagram describing the subquotient structure of the bimodule $\Hilb_{2L}$ over the pair
$(\JTL{N},\centJTL)$.  The two commuting actions
are presented in Fig.~\ref{JTL-Uqodd-bimod}
 where we show a direct sum
of the spin-chain modules $\APrTL{j}$ over $\JTL{N}$. The direct sum
is depicted as a (horizontal) sequence of diagrams for $\APrTL{j}$
from $j=-L$ on the left to $j=L$ on the right. Each node in the
diagram is a simple subquotient over the product
$\JTL{N}\boxtimes\LQGodd$. The action of $\JTL{N}$ is depicted by
vertical arrows while the action of $\LQGodd$ is shown by dotted
horizontal lines connecting different $\JTL{N}$-modules.  We  note that the $\JTL{N}$-modules
$\APrTL{j}$ in Fig.~\ref{JTL-Uqodd-bimod} are drawn in opposite
direction ``from bottom to top" comparing to diagrams in
Fig.~\ref{FF-JTL-mod}.

In the diagram, the first (horizontal) layer at the bottom contains four
nodes, which are simple $\JTL{N}$-modules $(\TLX_1)$, and dotted
arrows mixing them describe the indecomposable $\LQGodd$-module
$\PPodd{1}$. The second layer contains eight nodes of type $(\TLX_2)$
and the dotted arrows contribute to the indecomposable module
$\PPodd{2}$ presented on Fig.~\ref{Uq_odd_T-N8} in the front of
$(\TLX_2)$, {\it etc.} We emphasize that we do not draw long-range arrows
representing action of the generators $\FF{>0}$ and $\EE{>0}$ in
modules $\PPodd{n>1}$ in order to simplify diagrams but the arrows can
be easily recovered using either the homomorphisms of
$\JTL{N}$-modules described above in Thm.~\ref{Thm:homs-APrTL} or 
the subquotient structure of $\PPodd{n}$ described
in~\ref{subsec:PPodd-def} -- for example, the second layer of the
bimodule contains in addition four long arrows going from the node
$\diamond$ at $j=\mp1$ to the node $\bullet$ at $j=\pm2$, and from the
node $\bullet$ at $j=\pm2$ to the node {\scriptsize$\vardiamond$} at $j=\mp1$. With this
comment about arrows in mind, the reader can compare he complexity of
this bimodule with the open-case bimodule in
Fig.~\ref{openbimodule-fin}.

\subsection{The bimodule in the twisted case}

We recall~\cite{GRS1} that the twisted or antiperiodic version
of the $\gl(1|1)$ spin chain is obtained by setting
$f^{(\dagger)}_{2L+1}=-f^{(\dagger)}_1$ (compare with the conditions~\eqref{f_j-rel} for the periodic case). 
We then obtain from~\eqref{rep-JTL-1} a different expression for
 $e_{2L}$,
\begin{equation*}
e_{2L}=(f_{2L}-f_{1})(f_{2L}^\dagger-f_{1}^\dagger),
\end{equation*}
 which does not provide  a
representation of the $\JTL{N}$
algebra any longer. It still   does provide of the  even  affine
  Temperley--Lieb algebra 
  $\oATL{N}$ introduced in Sec.~\ref{sec:TL-alg-def}, see also~\eqref{diag-alg}.
We recall that in the diagrammatic language the $\JTL{N}$
algebra  corresponds to a quotient of $\oATL{N}$, where a
non-contractible loop on a cylinder is replaced by the numerical
factor $m=0$, while the antiperiodic boundary conditions now require
 a quotient of $\oATL{N}$, where non-contractible loops are given the weight~$2$ (the
 dimension of the fundamental or its dual, instead of the
 superdimension).
We also have the relation $u^{N}=(-1)^j$ which is satisfied in the sector
with $2j$ through-lines and which
means that we impose the condition $z^{2j}=(-1)^j$ on the
$z^2$-parameter in this sector.
 We will call the
corresponding finite-dimensional algebra $\JTL{N}^{tw}$. This algebra is related with the twisted or
deformed version of the Jones algebra studied in~\cite{Green}.

We next recall the result~\cite{GRS1} about the
centralizer of the representation of $\JTL{2L}^{tw}$.
The choice of an ``even" subalgebra in $U_{\q}s\ell(2)$ at generic $\q$,
{\it i.e.}, generated by the renormalized \textit{even}-powers of the $\E$
and $\F$ gives in the limit $\q\to i$ the centralizer for (the representation of) $\JTL{N}^{tw}$
on the antiperiodic spin-chain --- the usual $U(s\ell(2))$
generated by the $\e$ and $\f$.
\begin{thm}\cite{GRS1}
On the alternating antiperiodic $\gl(1|1)$ spin chain,
 the centralizer of the image
of the representation of the algebra $\JTL{N}^{tw}$ is the associative algebra $\repQG(U s\ell(2))$.
\end{thm}

We then describe the decomposition of the spin-chain over  $U s\ell(2)$ and then use it to obtain the  decomposition over $\JTL{N}^{tw}$.
Recall first the decomposition~\eqref{decomp-LQG} of $\Hilb_N$ over $\LQG$
where each indecomposable direct summand $\PP_{1,j}$ given in~\eqref{schem-proj}
is decomposed over the $U s\ell(2)$ subalgebra onto the direct sum
$2\XX_{1,j}\oplus \XX_{1,j-1}\oplus \XX_{1,j+1}$. We recall that each 
module $\XX_{1,j}$ has a trivial action of
 $\E$, $\F$, and $\K$ while it is the $j$-dimensional module under $U s\ell(2)$.
We  can thus easily write a decomposition with respect to the action of
the renormalized powers $\e$ and $\f$:
\begin{equation}
\Hilb_N|_{\rule{0pt}{7.5pt}%
U s\ell(2)} = \bigoplus_{j=1}^{L+1} \bigl(2d^0_j + d^0_{j-1} +
d^0_{j+1}\bigr) \XX_{1,j}
\end{equation}
where we set $d^0_0=0$ and $d^0_j=0$ for all $j>L$. The multiplicities
in front of $\XX_{1,j+1}$, with $0\leq j\leq L$, give dimensions of simple modules over
$\JTL{N}^{tw}$ which we denote as $\AIrrTL{j}{(-1)^j}$ (half of
them, those corresponding to even $j$,
are also modules over $\JTL{N}$ and we use the same notation which
should not be a source of confusion). Therefore, we obtain the structure of the bimodule
, which is semisimple:
\begin{equation}
\Hilb_N|_{\rule{0pt}{7.5pt}%
\JTL{N}^{tw}\boxtimes U s\ell(2)} = \bigoplus_{j=0}^{L}\AIrrTL{j}{(-1)^j} \boxtimes \XX_{1,j+1}
\end{equation}
The dimension of $\AIrrTL{j}{(-1)^j}$ is $2d^0_{j+1} + d^0_{j} +
d^0_{j+2}$ and is  computed using the binomial expression
$d^0_j=\sum_{i=j}^L(-1)^{j-i}\left(\binom{N}{L+i}
 - \binom{N}{L+i+1}\right)$ with the result
\begin{equation}\label{dimIrrATL-anti}
\dim \AIrrTL{j}{(-1)^j} = \binom{N}{L+j} - \binom{N}{L+j+2},\qquad
0\leq j\leq L.
\end{equation}
This  agrees with the structure of the standard modules
$\AStTL{j}{(-1)^j}$ over $\JTL{N}^{tw}$ that can be deduced from~\cite{GL}. 
The subquotient structure is now of chain type, and thus simpler than for
$\AStTL{j}{(-1)^{j+1}}$ $\JTL{N}$-modules, which are of the two-strands type
described in Sec.~\ref{sec:stand-q-i}:
\begin{equation}\label{StATL-anti}
\AStTL{j}{(-1)^{j}}:\qquad
\AIrrTL{j}{(-1)^j} \longrightarrow \AIrrTL{j+2}{(-1)^j}
\longrightarrow \AIrrTL{j+4}{(-1)^j} \longrightarrow \dots
\end{equation}
Recal that $\dim \AStTL{j}{z^2} = \binom{N}{L+j}$. Then, the
dimensions~\eqref{dimIrrATL-anti} correspond to single subtractions in
accordance with the subquotient structure~\eqref{StATL-anti}.

\bigskip

\section{Conclusion}\label{sec:concl}

At the end of this technical paper we have thus reached our goal of
obtaining the bimodule structure for the $\gl(1|1)$ spin chain. While
the results are somewhat more complicated than in the open case, they
nevertheless bear a strong similarity with it. This corresponds
closely with the fact that bulk and boundary symplectic fermions
theories are deeply related as well. We emphasize that this is a
feature particular to the $\gl(1|1)$ case, which provides a non
faithful representation of the Jones--Temperley--Lieb algbera. Cases
such as $\gl(2|2)$ would be faithful, and in a certain sense even more
complicated, even though faithfulness would make many technical
aspects in fact simpler.

Our next and crucial task is to compare the bimodule over $\JTL{N}$ and $\LQGodd$  with the known information about the bulk symplectic fermion theory, and see to what extent the  algebraic properties of the finite spin chain could have been used to infer those of the continuum limit. This will be discussed in the third paper of this series \cite{GRS3}. 
 
\section*{Acknowledgements} 
We are grateful to C.~Candu, I.B.~Frenkel, M.R.~Gaberdiel, J.L.~Jacobsen, 
G.I.~Lehrer, V.~Schomerus, I.Yu.~Tipunin and R.~Vasseur for valuable discussions. We also thank the anonymous referee for many helpful comments and criticisms. The
work of A.M.G. was supported in part by Marie Curie IIF fellowship, the RFBR grant 10-01-00408,
and the RFBR--CNRS grant 09-01-93105.
A.M.G is  grateful to V.~Schomerus for his
kind hospitality in DESY, Hamburg, during  2010, and  also grateful to N.~Read for
kind hospitality in Yale University during  2011.  The work of
N.R. was supported by the NSF grants DMR-0706195 and DMR-1005895.   The work of
H.S. was supported by the ANR Projet 2010 Blanc SIMI 4 : DIME. The authors
are also grateful to the organizers of the ACFTA program at the Institut Henri Poincar\'e in Paris, where this work was finalized. 

\section*{Appendix A: The full quantum group $\LQG$ at roots of unity}
\renewcommand\thesection{A}
\renewcommand{\theequation}{A\arabic{equation}}
\setcounter{equation}{0}

We collect here the standard expressions for the quantum group $\LQG$ that we
use in the analysis of symmetries of $\gl(1|1)$ spin-chains. This
appendix is identical to appendix A in our first paper~\cite{GRS1},
and reproduced here only for the reader's convenience.  We introduce
standard notation for $\q$-numbers
$[n]=\ffrac{\q^n-\q^{-n}}{\q-\q^{-1}}$ and $[n]!=[1][2]\dots[n]$.

\subsection{Defining relations}
The \textit{full} (or Lusztig) quantum group $\LQG$ with $\q =
e^{i\pi/p}$, for integer $p\geq2$, is generated by $\E$, $\F$, and
$\K$ satisfying the standard relations for the quantum $s\ell(2)$,
\begin{equation}\label{Uq-com-relations}
  \K\E\K^{-1}=\q^2\E,\quad
  \K\F\K^{-1}=\q^{-2}\F,\quad
  [\E,\F]=\ffrac{\K-\K^{-1}}{\q-\q^{-1}},
\end{equation}
with the constraints,
\begin{equation}
  \E^{p}=\F^{p}=0,\quad \K^{2p}=\one,
\end{equation}
and 
additionally by the divided powers $\f\sim \F^p/[p]!$ and $\e\sim \E^p/[p]!$, which turn out to satisfy the usual $s\ell(2)$-relations:
\begin{equation}
  [\h,\e]=\e,\qquad[\h,\f]=-\f,\qquad[\e,\f]=2\h.
\end{equation}
There are also  ``mixed" relations
\begin{gather}
  [\h,\K]=0,\qquad[\E,\e]=0,\qquad[\K,\e]=0,\qquad[\F,\f]=0,\qquad[\K,\f]=0,\label{zero-rel}\\
  [\F,\e]= \ffrac{1}{[p-1]!}\K^p\ffrac{\q \K-\q^{-1} \K^{-1}}{\q-\q^{-1}}\E^{p-1},\qquad
  [\E,\f]=\ffrac{(-1)^{p+1}}{[p-1]!} \F^{p-1}\ffrac{\q \K-\q^{-1} \K^{-1}}{\q-\q^{-1}},
    \label{Ef-rel}\\
  [\h,\E]=\ffrac{1}{2}\E A,\quad[\h,\F]=- \ffrac{1}{2}A\F,\label{hE-hF-rel}
\end{gather}
where 
\begin{equation}\label{A-element}
  A=\,\sum_{s=1}^{p-1}\ffrac{(u_s(\q^{-s-1})-u_s(\q^{s-1}))\K
        +\q^{s-1}u_s(\q^{s-1})-\q^{-s-1}u_s(\q^{-s-1})}{(\q^{s-1}
         -\q^{-s-1})u_s(\q^{-s-1})u_s(\q^{s-1})}\,
        u_s(\K)\idem_s
\end{equation}
with the polynomials $u_s(\K)=\prod_{n=1,\;n\neq s}^{p-1}(\K-\q^{s-1-2n})$, and
the $\idem_s$ are some central primitive idempotents~\cite{BFGT}.
The relations~\eqref{Uq-com-relations}-\eqref{A-element} are the defining
relations of the associative algebra $\LQG$.

\medskip

The quantum group $\LQG$ has a Hopf-algebra structure
with the comultiplication 
\begin{gather}
  \Delta(\E)=\one\otimes \E + \E\otimes \K,\quad
  \Delta(\F)=\K^{-1}\otimes \F + \F\otimes\one,\quad
  \Delta(\K)=\K\otimes \K,\label{Uq-comult-relations}\\
  \Delta(\e)=\e\tensor\one +\K^p\tensor \e
  +\ffrac{1}{[p-1]!} \sum_{r=1}^{p-1}\ffrac{\q^{r(p-r)}}{[r]}\K^p\E^{p-r}\tensor \E^r
  \K^{-r},\label{e-comult}\\
 \Delta(\f)= \f\tensor \one + \K^p\tensor \f+\ffrac{(-1)^p}{[p-1]!} 
  \sum_{s=1}^{p-1}\ffrac{\q^{-s(p-s)}}{[s]}\K^{p+s}\F^s\tensor \F^{p-s}.\label{f-comult}
\end{gather}
The antipode and counit are not used in the paper but the reader can
find them, for example, in~\cite{BFGT}.

We can easily write the $(N-1)$-folded coproduct for the capital
generators $\E$ and $\F$,
\begin{equation}\label{N-fold-comult-cap}
\Delta^{N-1}\E =
\sum_{j=1}^{N}\underbrace{\one\tensor\dots\tensor\one}_{j-1}\tensor
\E\tensor \K \tensor \dots \tensor \K,\qquad
\Delta^{N-1}\F =
\sum_{j=1}^{N}\underbrace{\K^{-1}\tensor\dots\tensor \K^{-1}}_{j-1}\tensor
\F\tensor \one \tensor \dots \tensor \one.
\end{equation}

\subsection{Standard spin-chain notations}\label{gl-XX}
We introduced the more usual (in the spin-chain
literature~\cite{PasquierSaleur,DFMC}\footnote{We note that our
  convention for the spin-chain  representation differs from the one
  in~\cite{PasquierSaleur} by the change $\q\to\q^{-1}$.}) quantum group generators
 \begin{gather*}
 S^{\pm}=\sum_{1\leq j\leq N} \q^{-\sigma_1^z/2}\otimes\ldots
 \q^{-\sigma_{j-1}^z/2}\otimes\sigma_{j}^{\pm}\otimes\q^{\sigma_{j+1}^z/2}\otimes
 \ldots \otimes \q^{\sigma_{N}^z/2},\\
  k=\q^{S^z},\qquad \text{with}\quad S^z=\half\sum_{j=1}^{2L}\sigma_j^z,
 \end{gather*}
 where  $\sigma^{\pm}_j$ and $\sigma^z_j$ are  $2\times2$-matrices  acting on the $j$th tensorand as ,
\begin{equation}\label{Pauli}
\sigma^+ = 
\begin{pmatrix}
0 & 1\\
0 & 0
\end{pmatrix}, \quad
\sigma^- = 
\begin{pmatrix}
0 & 0\\
1 & 0
\end{pmatrix}, \quad
\sigma^z = 
\begin{pmatrix}
1 & 0\\
0 & -1
\end{pmatrix}.
\end{equation}
The defining relations are then (for $\q=e^{i\pi/p}$ and integer $p\geq2$)
\begin{gather*}
    k S^{\pm}k^{-1}=\q^{\pm 1}S^{\pm},\qquad
   \left[S^{+},S^{-}\right]=\ffrac{k^2-k^{-2}}{\q-\q^{-1}},\\
   (S^{\pm})^p = 0,\qquad k^{4p}=\one,
\end{gather*}
and the comultiplication is
\begin{equation*}
    \Delta(S^{\pm})=k^{-1}\otimes S^{\pm}+S^{\pm}\otimes k, \qquad \Delta(k^{\pm1})=k^{\pm1}\tensor k^{\pm1}.
\end{equation*}
We then have the Hopf-algebra homomorphism
\begin{equation*}
\E\mapsto S^+ k, \qquad \F\mapsto k^{-1} S^-,\qquad \K\mapsto k^2
\end{equation*}
relating the two choices. The antipode and counit formulas can be easily obtained in the spin-chain notations as well but we do not need them in this paper.

\subsubsection{The case of XX spin-chains}
For $p=2$ or ``XX spin-chain'' case, the  $(N-1)$-folded coproduct of the renormalized powers $\e$ and $\f$ reads 
\begin{multline}\label{N-fold-comult-ren-e}
\Delta^{N-1}\e =
\sum_{j=1}^{N}\underbrace{\one\tensor\dots\tensor\one}_{j-1}\tensor
\e\tensor \K^2 \tensor \dots \tensor \K^2 +\\
+ \q\sum_{t=0}^{N-2}\sum_{j=1}^{N-1-t} \underbrace{\one\tensor\dots\tensor\one}_{j-1}\tensor
\E\tensor \underbrace{\K \tensor \dots \tensor \K}_{t}\tensor \E\K\tensor \K^2 \tensor \dots \tensor \K^2
\end{multline}
and
\begin{multline}\label{N-fold-comult-ren-f}
\Delta^{N-1}\f =
\sum_{j=1}^{N}\underbrace{\K^2\tensor\dots\tensor \K^2}_{j-1}\tensor
\f\tensor \one \tensor \dots \tensor \one +\\
+ \q^{-1}\sum_{t=0}^{N-2}\sum_{j=1}^{N-1-t}
\underbrace{\K^2\tensor\dots\tensor \K^2}_{j=1}\tensor
\K^{-1}\F\tensor \underbrace{\K^{-1} \tensor \dots \tensor \K^{-1}}_{t}\tensor \F\tensor \one \tensor \dots \tensor \one.
\end{multline}
These renormalized powers can also be expressed in terms of the more usual spin-chain operators, and one finds at $p=2$
\begin{equation*}
\Delta^{N-1}(\e)=\q S^{+(2)}k^{2},\qquad \Delta^{N-1}(\f)=\q^{-1}
k^{-2}S^{-(2)},
\end{equation*}
where $\q=i$ and
 \begin{equation}
 S^{\pm (2)}=\sum_{1\leq j<k\leq N-1} \q^{-\sigma_1^z}\otimes\ldots
 \otimes \q^{-\sigma_{j-1}^z}\otimes\sigma_{j}^{\pm}
 \otimes1\otimes\ldots\otimes 1\otimes \sigma_k^{\pm} \otimes
 \q^{\sigma_{k+1}^z}\otimes \ldots \otimes \q^{\sigma_{N}^z}.
 \end{equation}

\medskip
We also note that the $\gl(1|1)$
spin-chain representation $\repgl$ is equivalent~\cite{GRS1} to a twisted XX spin
chain representation $\repXX$ of $\JTL{2L}$. The expression of the Temperley--Lieb generators in this case is well known for the open chain \cite{PasquierSaleur},
\begin{equation}
\repXX(e_j) \equiv e^{XX}_j=-\half\left[\sigma_j^x\sigma_{j+1}^x+\sigma_j^y\sigma_{j+1}^y - \q(\sigma_j^z-\sigma_{j+1}^z)\right],
\end{equation}
where $\sigma^{x,y,z}$ are Pauli matrices introduced in~\eqref{Pauli} and we used $\sigma^{\pm}=\half\bigl(\sigma^x\pm i\sigma^y\bigr)$. To get an equivalence  in the closed case we need to  set in the expression of $e_{2L}$ the following condition:
\begin{equation}
\sigma_{2L+1}^{\pm}=-(-1)^{S^z}\sigma_1^\pm.
\end{equation}
This means that a periodic $\gl(1|1)$ (alternating) spin chain
corresponds to a periodic XX spin chain for odd values of  $S^z$ and to an antiperiodic  XX spin chain 
for even values of  $S^z$.

\section*{Appendix B: Projective $\LQG$-modules $\PP_{1,r}$}
\renewcommand\thesection{B}
\renewcommand{\theequation}{B\arabic{equation}}
\setcounter{equation}{0}

We recall~\cite{BFGT} the action of $\LQG$ (for $\q=i$) in projective
covers $\PP_{1,r}$ of simple modules $\XX_{1,r}$, where $r$ is an
integer and $r\geq1$.  A module $\XX_{1,r}$ is $r$-dimensional and spanned by
$\stprp_{m}$, $0\leq m\leq r{-}1$, with \footnote{We
simplify a notation used in~\cite{BFGT} assuming $\XX_{1,r} \equiv
\XX^{\alpha(r)}_{1,r}$ with $\alpha(r)=(-1)^{r-1}$, and the same for $\PP_{1,r}$.  }
\begin{equation}\label{eq:sl-irrep}
\begin{split}
&\E\, \stprp_{m} = \F\, \stprp_{m} = 0,\qquad \K\, \stprp_{m} = (-1)^{r-1} \stprp_{m},\\
  \h\, \stprp_{m} &=  \half(r-1-2m)\stprp_{m},\quad
  \e\, \stprp_{m} =  m(r-m)\stprp_{m-1},\quad
  \f\, \stprp_{m} =  \stprp_{m+1},
\end{split}
\end{equation}
where we set $\stprp_{-1}
=\stprp_{r}=0$.
For $r=0$, we also set $\XX_{1,0}\equiv 0$. 
The subquotient
structure of $\PP_{1,r}$ is then given as 
\begin{equation}\label{schem-proj-app}
{\small
     \xymatrix@C=2pt@R=24pt@M=1pt@W=2pt{&&\\&\PP_{1,r}\quad=\quad&\\&&}
\xymatrix@C=4pt@R=18pt@M=2pt@W=2pt{
           &\XX_{1,r} \ar[dl] \ar[dr]&\\
 	  \XX_{1,r-1} \ar[dr]
 	    &&\XX_{1,r+1}\ar[dl]\\
        &\XX_{1,r}&
    } 
}
\end{equation}

For $r > 1$, the projective module $\PP_{1,r}$ has the basis
\begin{equation}\label{left-proj-basis-plus}
  \{\toppr_{m},\botpr_{m}\}_{0\le m\le r-1}
  \cup\{\leftpr_{l}\}_{1\le l\le
  r-1}
\cup\{\rightpr_{l}\}_{0\le l\le
  r},
\end{equation}
where $\{\toppr_{m}\}_{0\le m\le
    r-1}$ is the basis
corresponding to the top module in~\eqref{schem-proj-app},
$\{\botpr_{m}\}_{0\le m\le r-1}$
to the bottom, $\{\leftpr_{l}\}_{1\le l\le r-1}$ to the left, and
$\{\rightpr_l\}_{0\le l\le r}$ to
the right module. 
For $r=1$, the basis does not contain
$\{\leftpr_{l}\}_{1\le l\le
r-1}$ terms and we imply $\leftpr_{l}~\equiv~0$ in the
action.

We set $\alpha(r)=(-1)^{r-1}$.   The action of $\LQG$ on $\PP_{1,r}$ is then given by
\begin{align}
  \K\toppr_{m}&=\alpha(r)\toppr_{m}, \qquad  \K\botpr_{m}=\alpha(r)\botpr_{m}, &\quad 0\le m\le r-1,&\notag\\
  \K\leftpr_{l}&=-\alpha(r)\leftpr_{l}, &\quad 1\le l\le r-1,&\notag\\
  \K\rightpr_{l}&=-\alpha(r)\rightpr_{l}, &\quad 0\le l\le r,&\notag\\
  \E\toppr_{m}&=
    \alpha(r) \ffrac{r-m}{r}\rightpr_{m} + \alpha(r) \ffrac{m}{r}\leftpr_{m},   \qquad \E\botpr_{m}=  0,  &\quad 0\le m\le r-1,&\notag\\
  \E\leftpr_{l}&=
    \alpha(r) (l-r)\botpr_{l-1},  &\quad 1\le l\le r-1,&\label{proj-mod-QGbase}\\
  \E\rightpr_{l}&=
    \alpha(r) l\botpr_{l-1},
&  \quad 0\le l\le r,&\notag\\
  \F\toppr_{m}&=
    \ffrac{1}{r}\rightpr_{m+1}-\ffrac{1}{r}\leftpr_{m+1}, \qquad \F\botpr_{m}= 0 &\quad 0\le m\le r-1, 
    \quad(\leftpr_{r}\equiv0),&\notag\\
  \F\leftpr_{l}&=
    \botpr_{l}, &
  \quad 1\le l\le r-1,&\notag\\
  \F\rightpr_{l}&=
    \botpr_{l}, &
  \quad 0\le l\le r.&\notag
\end{align}

In the  basis thus introduced , the $s\ell(2)$-generators $\e$, $\f$ and $\h$
act in $\PP_{1,r}$ as in the direct sum
$\XX_{1,r}\oplus\XX_{1,r-1}\oplus\XX_{1,r+1} \oplus \XX_{1,r}$ with
the action defined in~\eqref{eq:sl-irrep}.


\begin{thebibliography}{99}

\bibitem{GRS1} A.M.~Gainutdinov, N.~Read and H.~Saleur,
  \textit{Continuum limit and symmetries  of  the periodic $\gl(1|1)$
    spin chain}, arXiv:1112.3403.

\bibitem{ReadSaleur07-1} N. Read and H. Saleur,  \textit{Enlarged symmetry algebras of spin chains, 
  loop models, and S-matrices}, Nucl. Phys. B777 (2007) 263.

\bibitem{ReadSaleur07-2} N. Read and H. Saleur,
  \textit{Associative-algebraic approach to logarithmic conformal
    field theories},  Nucl. Phys. B777 (2007) 316.

\bibitem{GRS3} A. Gainutidnov, N. Read and H. Saleur,  \textit{Associative algebraic approach to  logarithmic CFT in the bulk: 
the continuum limit of the $\gl(1|1)$ spin chain and the interchiral algebra}, arXiv:1207.6334. 

\bibitem{GL} J.J.~Graham and G.I.~Lehrer, \textit{The representation
theory of affine Temperley-Lieb algebras}, L'Ens. Math. 44 (1998) 173.

\bibitem{GL1} J.J. Graham and G. I. Lehrer,  \textit{The two-step
nilpotent representations of the extended Affine Hecke algebra of type
A}, Compositio Mathematica 133 (2002) 173.

\bibitem{MartinSaleur} P.P. Martin and H. Saleur,  \textit{The blob algebra
and the periodic Temperley-Lieb algebra}, Lett. Math. Phys. 30 (1994) 189.
\bibitem{MartinSaleur1} P.P. Martin and H. Saleur, \textit{On an
  algebraic approach to higher-dimensional statistical mechanics}, Comm. Math. Phys. 158 (1993) 155.

\bibitem{FanGreen} C.K. Fan and R.M. Green, \textit{On the affine
  Temperley--Lieb algebras}, arXiv:q-alg/9706003.
\bibitem{Green} R.M. Green, \textit{On representations of affine Temperley--Lieb algebras}, Algebras and Modules 
II, CMS Conference Proceedings, vol. 24, Amer. Math. Soc., Providence, RI, 1998, 245-261.

\bibitem{[Br]} R. Brauer, \textit{On Algebras Which are
  Connected with the Semisimple Continuous Groups}, Ann. of
  Math. (1937) 38 N4, 857-872.
\bibitem{Jones}  V.F.R.~Jones, \textit{Quotient of the affine Hecke
  algebra in the Brauer algebra}, L'Ens. Math. 40 (1994) 313.

\bibitem{GL0} J.J.~Graham and G.I.~Lehrer, \textit{Cellular algebras},
  Invent. Math. 123 (1996), 1-34.

\bibitem{Donkin} S.~Donkin, \textit{The q-Schur Algebra},  
London Mathematical Society Lecture Note Series, 1998.

\bibitem{GJSV} A.M. Gainutdinov, J.L. Jacobsen, H. Saleur, R. Vasseur, \textit{A physical approach to the classification of indecomposable Virasoro representations from the Blob algebra}, arxiv:1212.0093.
\bibitem{AF-book} F.W. Anderson and K.R. Fuller, \textit{Rings and Categories of Modules}, Graduate Texts in Math. 13,
Springer-Verlag, NY, 1992.

\bibitem{BFGT}P.V. Bushlanov, B.L. Feigin, A.M. Gainutdinov,
 I.Yu. Tipunin,
\textit{Lusztig limit of quantum $s\ell(2)$ at root of
unity and fusion of $(1,p)$ Virasoro logarithmic minimal models},
 Nucl. Phys. B 818 [FS] (2009) 179-195.

\bibitem{[M]} S.~MacLane, \textit{Homology}, Springer-Verlag, 1963.

\bibitem{PasquierSaleur} V. Pasquier and H. Saleur, \textit{Common structures
  between finite systems and conformal field theories through quantum
  groups}, Nucl. Phys. B 330, 523 (1990).

\bibitem{ChPr} V.~Chari, A.~Pressley, \textit{A guide to quantum groups}, CUP, 1994 667pp.

\bibitem{DFMC} T. Deguchi, K. Fabricius and B. Mc Coy, \textit{The
  $sl(2)$ loop algebra symmetry of the six-vertex model at roots of unity}, J. Stat. Phys. 102 (2001) 701.


\end{thebibliography}
\end{document}